\documentclass[lineno]{biometrika}
\usepackage{hyperref}

\hypersetup{
  colorlinks=true,
  citecolor = blue,
  linkcolor=blue
}

\usepackage{amsmath}

\usepackage{caption}

\usepackage{amssymb,mathtools}

\usepackage{graphicx}
\usepackage{times}
\usepackage{bm}
\usepackage{subcaption}
\usepackage{appendix}

\usepackage{enumitem}
\usepackage{pifont}

\usepackage{cleveref}

\usepackage{comment}
\usepackage{tabularx}
\usepackage{multirow}
\usepackage{makecell}
\usepackage{multicol, multirow, colortbl}

\DeclareMathAlphabet{\mathbbold}{U}{bbold}{m}{n}
\newcommand*{\one}{1}

\newcommand{\quotes}[1]{``#1''}

\newcommand\independent{\perp\kern-6pt\perp}

\newcolumntype{C}{>{\centering\arraybackslash}X}
\newcommand{\cmark}{\ding{51}}
\newcommand{\xmark}{\ding{55}}

\DeclareMathOperator{\Var}{Var}
\DeclareMathOperator{\Bias}{Bias}

\let\E\relax
\DeclareMathOperator{\E}{E}
\let\R\relax
\DeclareMathOperator{\R}{R}
\let\O\relax
\DeclareMathOperator{\O}{O}
\let\P\relax
\DeclareMathOperator{\P}{P}
\let\PP\relax
\DeclareMathOperator{\PP}{P}
\let\GG\relax
\DeclareMathOperator{\GG}{G}

\usepackage[plain,noend]{algorithm2e}

\makeatletter
\renewcommand{\algocf@captiontext}[2]{#1\algocf@typo. \AlCapFnt{}#2} 
\def\@algocf@capt@plain{top}
\renewcommand{\algocf@makecaption}[2]{%
  \addtolength{\hsize}{\algomargin}%
  \sbox\@tempboxa{\algocf@captiontext{#1}{#2}}%
  \ifdim\wd\@tempboxa >\hsize
    \hskip .5\algomargin%
    \parbox[t]{\hsize}{\algocf@captiontext{#1}{#2}}
  \else%
    \global\@minipagefalse%
    \hbox to\hsize{\box\@tempboxa}
  \fi%
  \addtolength{\hsize}{-\algomargin}%
}
\makeatother

\usepackage[normalem]{ulem}

\begin{document}

\jname{Biometrika}
\jyear{2024}
\jvol{103}
\jnum{1}
\cyear{2024}
\accessdate{Advance Access publication on 31 July 2023}

\received{2 January 2017}
\revised{1 August 2023}

\markboth{Y. Zhang \and Q. Zhao}{Biometrika style}

\title{An average-case sensitivity analysis for unmeasured confounding}

\author{Yao Zhang}
\affil{Department of Statistics and Data Science,  National University of Singapore,\\ 
Block S16, Level 7, 6 Science Drive 2, Singapore
 \email{yaozhang@nus.edu.sg}}

\author{\and     Qingyuan Zhao }
\affil{Statistical Laboratory, Centre for Mathematical Sciences, University of Cambridge,\\ Wilberforce Road, Cambridge CB3 0WB, U.K. \email{qyzhao@statslab.cam.ac.uk}}

\maketitle

\begin{abstract}
Sensitivity analysis for the unconfoundedness assumption is crucial in
observational studies. For this purpose, the marginal sensitivity
model gained popularity recently due to good interpretability and
mathematical properties. However, most existing models only consider a
worst-case parameter that bounds the logit difference between the
observed and full data propensity scores, which may not fully capture
the extent of unmeasured confounding. We propose a new sensitivity model that is parameterized by the second moment of the propensity score
ratio, requiring only the average strength of unmeasured confounding
to be bounded. By characterizing the associated sensitivity analysis
as an optimization problem, we derive sharp closed-form bounds of the
average potential outcomes under our model.  We propose efficient
one-step estimators for these bounds based on the corresponding
efficient influence functions. Additionally, we apply multiplier
bootstrap to construct simultaneous confidence bands to cover the
sensitivity curve that consists of bounds at different values of the sensitivity
parameters. Through a real-data study, we illustrate how this
average-case sensitivity analysis can provide tighter bounds and
facilitate calibration of the results using observed covariates. 
\end{abstract}


\section{Introduction}


Unconfoundedness (also referred to as ignorability) is a critical assumption for causal inference from observational studies
\citep{rosenbaum1983central, rosenbaum1984reducing}. In a canonical
setting with two treatment levels (let $0$ denote control and $1$
denote treated), unconfoundedness is satisfied if a study measures a
rich set of covariates $X$ such that the (real-valued) potential
outcomes under different levels of the treatment, $Y(0)$ and $Y(1)$,
are independent of the treatment $Z$ given $X$:
\begin{assumption}[Unconfoundedness]
\label{assume:xyzu}
    $Y(1), Y(0)\independent Z\mid X.$
\end{assumption}
This assumption essentially means that the treatment in an observational study can be regarded as randomized for the purpose of
identifying and estimating the average treatment effect (ATE) $\E\{Y(1) -
Y(0)\}$ and other causal quantities. To see this, unconfoundedness
means the ``observed data'' and ``full data'' propensity scores,
defined respectively as
\[
  e(x) = \PP\{Z=1\mid X=x\} \quad  \text{and} \quad  e(x,y) =
  \PP\{Z=1\mid X=x,Y(1)=y\},
\]
are always equal: $e(x) = e(x,y)$ for all $x$ and $y$. With this, we
can identify the expectation
of $Y(1)$ given $X$ by inverse-probability weighting (IPW) \citep{horvitz1952generalization}:
\begin{equation}\label{equ:first_formula}
\E\left\{ Y(1) \mid X\right\}  =\E\left\{  \frac{Z Y(1) }{e(X,Y(1))} \
  \bigg| \ X \right\}
= \E\left\{ \frac{ZY}{e(X)} \ \bigg| \ X \right\},
\end{equation}
where the first equality is a change of measure (which requires
$e(X,Y(1)) > 0$) and the second equality uses Assumption
\ref{assume:xyzu}. By taking the averaging of
\eqref{equ:first_formula} over $X$, we can then identify
$\E\{(Y(1))\}$ and similarly $\E\{Y(0)\}$.

To make an observational study credible, it is thus crucial to
reason about the unconfoundedness assumption using practical
context. Although unconfoundedness is untestable using just observational
data, we can use sensitivity analysis to assess whether the
conclusions of an observational study would change significantly under
certain violations of unconfoundedness. Sensitivity analysis for
observational studies can be traced back to
\citet{cornfield1959smoking}; the methodology and results in that
article played an instrumental role in the debate about smoking as a
major cause of lung cancer. Since then, a variety of sensitivity
analysis models and methods have been proposed
\citep{rosenbaum1983assessing,rosenbaum1987sensitivity,robins2000sensitivity,paul2002observation,imbens2003sensitivity,vanderweele2017sensitivity,cinelli2020making,bonvini2021sensitivity}. See
\Cref{sec:related-works} for a brief review of this literature.

This article is concerned with the marginal sensitivity models. Such
models are called \emph{marginal} by \citet{zhao2019sa}, because they
compare the full data propensity score
$e(x, y)$ with its marginal counterpart $e(x)$. Several recent
articles \citep{zhao2019sa,dorn2021sharp,dorn2021doubly} considered the following marginal sensitivity model \citep{tan2019marginal}
indexed by a worst-case parameter $\Gamma \geq 1$ chosen by the user:
\begin{equation}\label{equ:lambda_e}
     \Gamma^{-1}\leq\frac{e(X)/\{1-e(X)\}}{e(X,Y(1))/\{1-e(X,Y(1))\}} \leq \Gamma.
\end{equation}
When studying the population sensitivity analysis problem (basically before \Cref{sect:stat_inference}), we will treat $e(x)$ as known. Putting equation \eqref{equ:lambda_e} in a different way, it bounds the $L^{\infty}$-norm (essential supremum) of the logit difference between $e(x)$ and $e(x,y)$ by $\log(\Gamma)$. This model is similar but
different from another popular sensitivity model proposed by
\citet{rosenbaum1987sensitivity} that bounds the $L^{\infty}$-norm of
the logit difference between $e(x,y_1)$ and $e(x,y_2)$, which is
particularly convenient for matching methods.

Let $h(x,y):=e(x)/e(x,y)$ denote the ratio of the propensity
scores. Equation \eqref{equ:lambda_e} can be alternatively viewed as a
constraint on $e(x,y)$ or $h(x,y)$, as $e(x)$ can be estimated from the
data.
Besides the constraint above, by definition the propensity score ratio $h$ is bounded from below,
\begin{equation}\label{equ:add_1}
h(X,Y(1))\geq e(X),
\end{equation}
and needs to
satisfy an additional ``marginalization constraint'',
 \begin{equation}\label{equ:mean_1}
\E\{h(X,Y)\mid X=x,Z=1\} = 1,
 \end{equation}
under the usual consistency assumption (also called SUTVA), i.e.,
$Y = Y(1)$ almost surely given $Z=1$. 
The worst-case marginal sensitivity model can thus be defined as
\begin{equation}\label{equ:h_infinity}
\mathcal H_{\text{wst}}(\Gamma) = \big\{ h(x,y):  h \text{ satisfies }
\eqref{equ:lambda_e}, \eqref{equ:add_1},
\eqref{equ:mean_1} \text{ almost surely}     \big\}.
\end{equation}

Define the first moment of $Y$ conditional on $X$ and $Z=1$:
\[
 \mu_{1,h}(X) :=  \E \{h(X,Y)Y\mid X, Z=1 \}.
\]
\citet{dorn2021sharp} derived
sharp bounds for $\mathbb E\{Y(1)\}$ by solving the optimization problem
\begin{equation}\label{equ:msm_opt}
\text{minimize or maximize} \quad \E \left\{  \mu_{1,h}(X)  \right\}
\ \text{ subject to } \ h\in \mathcal H_{\text{wst}}(\Gamma),
\end{equation}

The objective function in \eqref{equ:msm_opt} is a reformulation of
the middle quantity in \eqref{equ:first_formula}. Like Rosenbaum's
sensitivity model, this model is sometimes criticized for being too
pessimistic because the parameter $\Gamma$ can be determined by a single
value of $X$. See \Cref{example:one} for an illustration.

In this article, we introduce an alternative sensitivity model that
instead restricts the difference between the full and observed
data propensity scores in an average sense. This model is motivated by
the mean one constraint of in \eqref{equ:mean_1} and
the observation that $h(X,Y(1)) = 1$ almost surely when
unconfoundedness (Assumption \ref{assume:xyzu}) holds.
When unconfoundedness is violated, the variance of $h(X,Y(1))$ is larger than 1.
Define the second moment of $h(X,Y(1))$ conditional on $X, Z=1 $:
\[
 \nu_{1,h}(X) :=  \E \{h^2(X,Y)\mid X, Z=1 \}.
\]
We can restrict the degree of confounding by bounding the expectation of this second moment:
\begin{equation}\label{equ:add_2}
\E \left\{  \nu_{1,h}(X)   \right\} \leq \Sigma.
\end{equation}
The associated average-case sensitivity model is then given by
\[
\mathcal H_{\text{avg}}(\Sigma) = \big\{ h(x,y):  h \text{ satisfies }  \eqref{equ:add_1},\eqref{equ:mean_1}, \eqref{equ:add_2}    \text{ almost surely}     \big\}.
\]
The parameter $\Sigma \geq 1$ plays a similar role to $\Gamma$ in the
worst-case model and is also chosen by the user. Sensitivity analysis
under this model then amounts to solving
\begin{equation}\label{equ:marginal_opt}
\text{minimize or maximize} \quad \mathbb E\{ \mu_{1,h}(X) \}
\ \text{ subject to } \ h\in \mathcal H_{\text{avg}}(\Sigma).
\end{equation}
Compared to the worst-case model, the average-case model depends less on extreme but rare confounding. We next illustrate this via a
numerical example {\color{black}in which an observed and an unobserved confounder interact,
creating a nonlinear effect on the treatment assignment mechanism.
}




\begin{example}
    \label{example:one}
Let $X\sim \text{Uniform}(0,1)$ and $U\sim \text{N}(0,1)$ be independent covariates,
and let the binary treatment $Z$ be drawn from a Bernoulli distribution with probability
\begin{equation}\label{equ:pzxu}
  \P(Z = 1 \mid X, U) =
  b\left(\frac{e^{c X U}}{1+e^{p X U}}\right),
\end{equation}
where $b(\cdot) = \min (\max(\cdot,0.05),0.95)$ clips the probability
at $0.05$ and $0.95$.
{\color{black}
The treatment assignment mechanism in \eqref{equ:pzxu} reflects a continuous
``gated confounding'' structure. As $X$ (e.g., the amount of time, access, or
flexibility an individual has to participate in a job training program)
increases, the unmeasured confounder $U$ (e.g., motivation or engagement)
has a stronger influence on treatment assignment through the term $c X U$.
When $X$ is small, treatment uptake is largely unrelated to $U$, so
$\P(Z=1\mid X,U)$ is close to $1/2$. The parameter $c$ controls this dependence and thus determines the overall
strength of unobserved confounding in the study.
}


The potential and observed outcomes are generated as follows:
\[
  Y(1)  = X- p + (U+U^3),\ Y(0) = Y(1)-5\ \text{ and }\ Y=
  ZY(1) + (1-Z)Y(0).
\]
{\color{black}
In this outcome model, the term $(U+U^3)$ introduces heterogeneity through the
unobserved confounder $U$, while the treatment effect is constant at $5$ for all
units.
}

\begin{figure}[t]
    \centering
\begin{subfigure}{.32\textwidth}
\includegraphics[width=\linewidth]{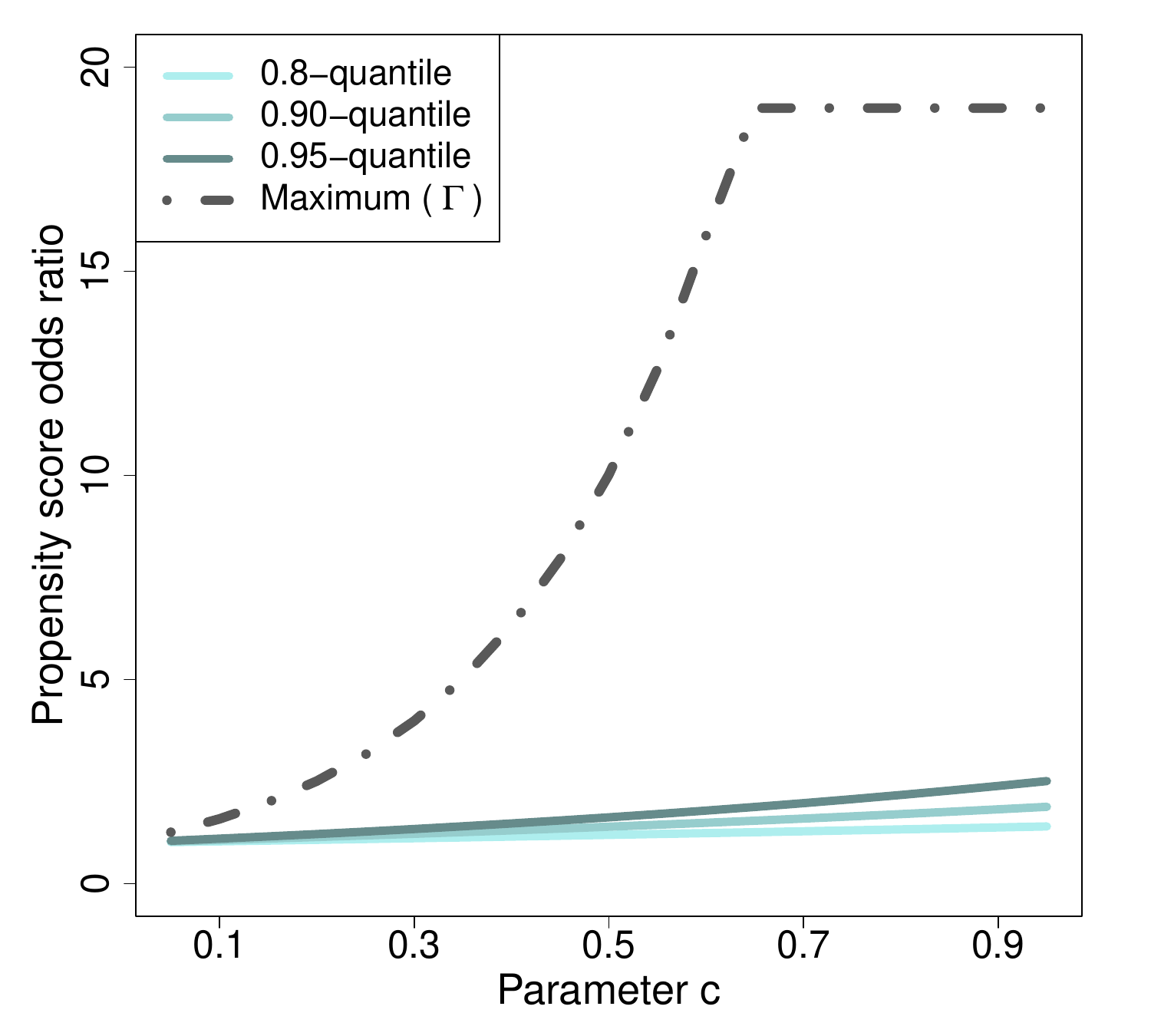}
    \caption{Worst-case PS odd ratios.}
    \label{subfig:toy_exp_infinity}
\end{subfigure}
\begin{subfigure}{.32\textwidth}
\includegraphics[width=\linewidth]{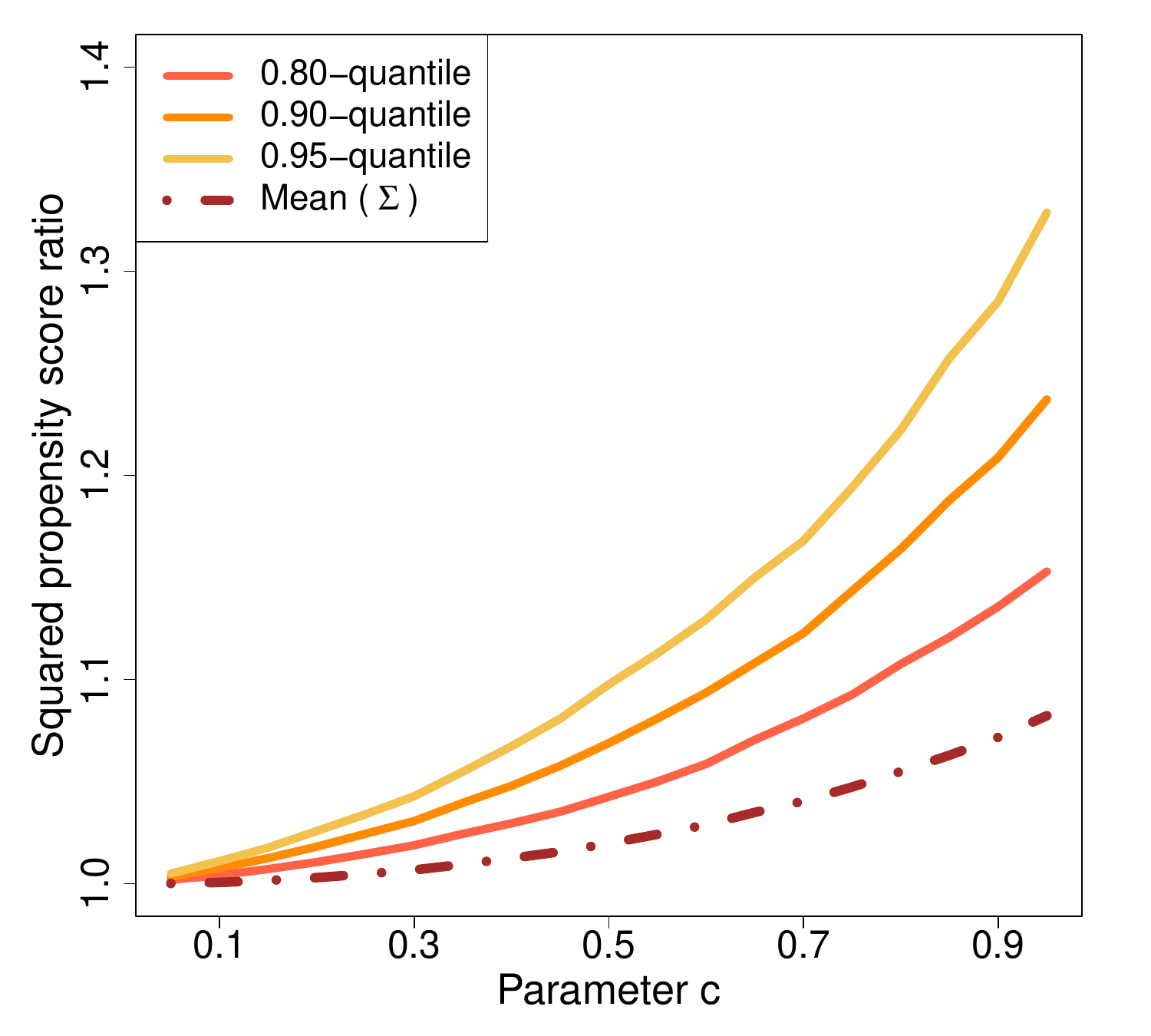}
    \caption{Average-case squared PS ratios.}
    \label{subfig:toy_exp_2}
\end{subfigure}
\begin{subfigure}{.32\textwidth}
\includegraphics[width=\linewidth]{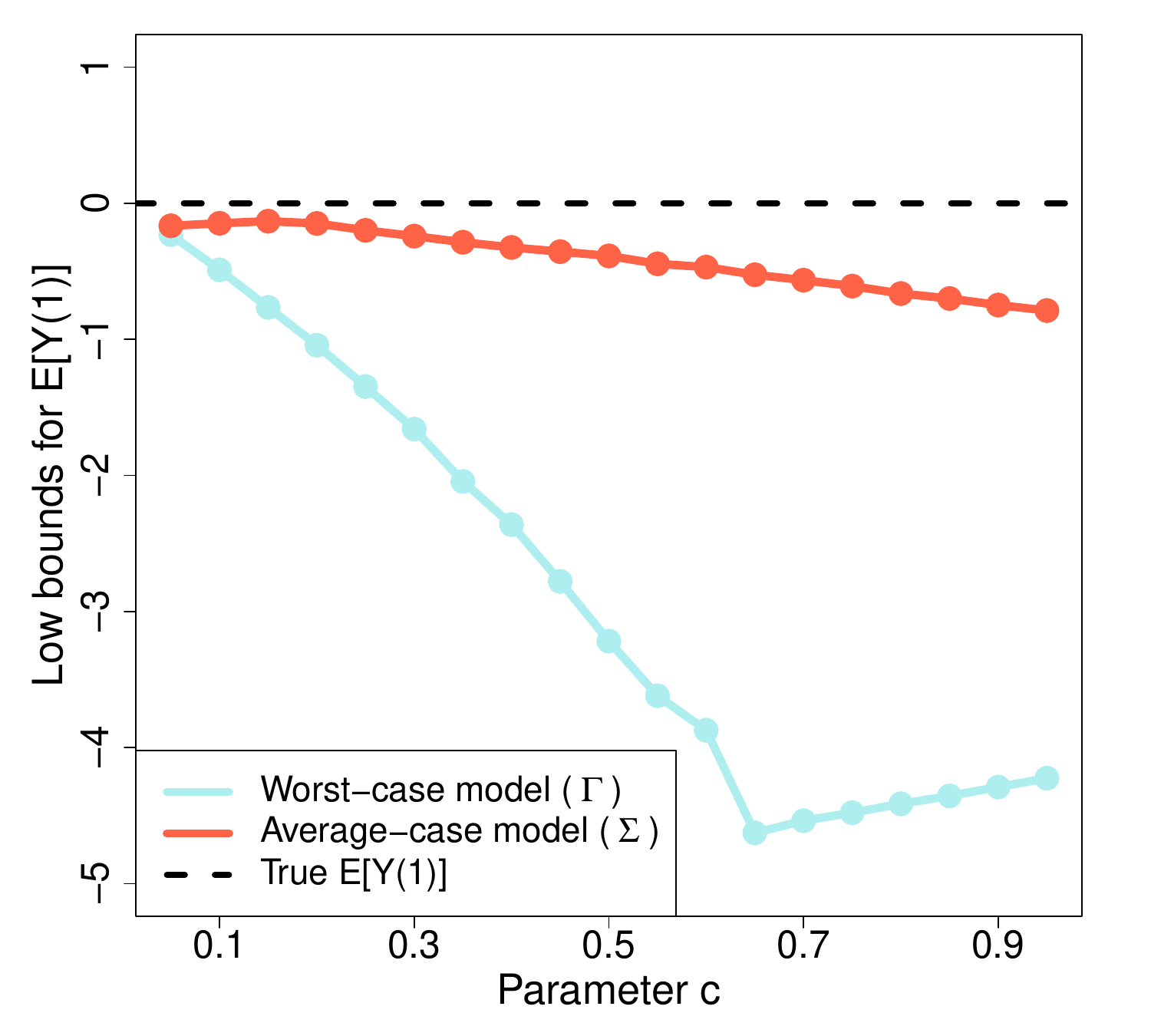}
    \caption{Lower bounds for $\E[Y(1)]$.}
    \label{subfig:toy_exp_bounds}
\end{subfigure}
\caption{
{\color{black}Numerical results for \Cref{example:one}. In panels (a) and (b), the maximum, quantiles, and mean of the (squared) propensity score (PS) ratios increase with the parameter $c$, which controls the strength of the unmeasured confounder $U$. In panel (c), the 
average-case model $\mathcal H_{\text{avg}}(\Sigma)$ yields a much tighter lower bound for $\E[Y(1)]$ than the worst-case model $\mathcal H_{\text{wst}}(\Gamma)$.}}
    \centering
\label{fig:toy_example}
\vspace{-10pt}
\end{figure}

{\color{black}
\Cref{subfig:toy_exp_infinity} shows the quantiles and the maximum of the propensity score odds ratio at varying levels of $c$. The quantiles represent the strength of $U$ for most units, while the maximum corresponds to the sensitivity parameter $\Gamma$ used in the worst-case model $\mathcal H_{\text{wst}}(\Gamma)$. The large gap between them highlights that $\Gamma$ is a pessimistic measure of the strength of $U$ for most units.

In contrast, the mean squared propensity score ratio $\Sigma$ in the average-case model $\mathcal H_{\text{avg}}(\Sigma)$ has a much smaller gap from the quantiles, as shown in \Cref{subfig:toy_exp_2}. This suggests that $\Sigma$ provides a more meaningful summary of unmeasured confounding for most units than the maximum $\Gamma$.

In \Cref{subfig:toy_exp_bounds}, we verify this by solving the optimization problems \eqref{equ:msm_opt} and \eqref{equ:marginal_opt} 
(with $\Gamma$ and $\Sigma$ set to the values in \Cref{subfig:toy_exp_infinity,subfig:toy_exp_2})
to derive lower bounds for $\E[Y(1)]$. By definition, $\E[Y(1)]=0$ for any value of $c\in(0,1)$. The average-case lower bound is much closer to $0$ than the worst-case lower bound. The worst-case lower bound increases near the end because $\Gamma$ remains at $0.95/0.05=19$ after $c\approx 0.6$, as shown in \Cref{subfig:toy_exp_infinity}.
}

\end{example}

{\color{black}
As alluded to in this example, this article generalizes the popular worst-case sensitivity model to a new average-case model, which allows for more optimistic sensitivity analysis.
To this end, we analytically solve the optimization problem \eqref{equ:marginal_opt} involved in the average-case model. Beyond simple differences in model formulation, we show these solutions and the resulting bounds provide a deeper and unifying view of the worst- and average-case models.

For estimation and inference, we derive the efficient influence functions (EIFs) for the bounds and sensitivity values in the worst-case and average-case sensitivity models. These functions can be used to construct efficient estimators of the average treatment effect. Moreover, they enable the use of the multiplier bootstrap \citep{belloni2018uniformly,kennedy2019nonparametric} to construct simultaneous confidence bands for sensitivity curves (i.e., sequences of bounds on the average treatment effect).
We use numerical simulations to study the finite-sample performance of our EIF-based estimators and confidence bands. We compare the average-case model with the worst-case model using a real-world dataset. The bounds from the average-case model admit a more optimistic interpretation when the sensitivity parameter is calibrated using observed covariates.
}

In this article, we let $p_{Y\mid X,Z=z}(\cdot) :=
p_{Y\mid X,Z}(\cdot\mid X,z)$ denote the density of $Y$ conditional on $X$ and $Z = z\in \{0,1\}$.
We use $a\lesssim b$ denote $a\leq  C b$ for some constant $C>0$. We
denote $a\wedge b := \min(a,b)$ and $a\vee b := \max(a,b).$ For a nuisance function $\eta(\cdot)$ and its estimator $\hat{\eta}(\cdot)$, we denote the root mean squared error by 
$\| \hat\eta - \eta \| := (\E\{[\hat\eta(O) - \eta(O)]^2 \mid \hat\eta \})^{1/2}$ 
and the supremum norm error by 
$\| \hat\eta - \eta \|_{\infty} := \sup_{o}|\hat\eta(o) - \eta(o)|$.

\section{Related work}
\label{sec:related-works}

Many sensitivity analysis methods have been developed since
\citet{cornfield1959smoking}. One common approach is to augment the
statistical model with more sensitivity parameters in relation to some unmeasured variable $U$
\citep{rosenbaum1983assessing,imbens2003sensitivity,vanderweele2011bias,vanderweele2017sensitivity,cinelli2020making}.
Although such models are easy to interpret, the sensitivity parameters
are often identifiable (sometimes just partially or weakly) from the observed data and cannot be arbitrarily chosen by the user \citep{scharfstein1999adjusting,gustafson2018sensitivity}. Other authors have also sought to relax the unconfoundedness assumption through
specifying a contrast $\delta$ between the
counterfactual distribution of $Y(z)$ given $X,Z=1-z$ and the factual
distribution of $Y(z)$ given
$X,Z=z$
\citep{robins2000sensitivity,birmingham2003pattern,vansteelandt2006ignorance,scharfstein2021semiparametric},
or between the distributions of $Z$ given $X,Y(z)$ and the distribution of $Z$ given $X$
\citep{scharfstein1999adjusting,gilbert2003sensitivity,gilbert2013sensitivity,franks2019flexible,masten2024assessing}.
The causal effect can then be identified for any fixed value of
$\delta$, but choosing an appropriate $\delta$ remains challenging.

In the pair-matched setting, the methods developed by
\citet{rosenbaum1987sensitivity,paul2002observation} provide simple
and interpretable sensitivity analysis. \citet{yadlowsky2022bounds}
considered Rosenbaum's model in the i.i.d.\ setup, and developed a method based on empirical loss minimization to derive bounds on the (conditional) average treatment effect.
Some recent articles attempted to relax the worst-case bounds in
Rosenbaum's model
\citep{hasegawa2017sensitivity,fogarty2019extended}. \citet{bonvini2021sensitivity} proposed an alternative sensitivity model parameterized by the proportion of confounded observations in the study.

The marginal sensitivity model $\mathcal H_{\text{wst}}(\Gamma)$ is first considered by \citet{tan2019marginal} and has gained popularity partly due to its close connection to distributionally robust optimization \citep{rockafellar2000optimization}. \citet{zhao2019sa} solved the
empirical version of the problem \eqref{equ:msm_opt} by
linear fractional programming, but neglecting the marginalization
constraint \eqref{equ:mean_1} and proposed a percentile bootstrap
method to construct confidence intervals of the
bounds. \citet{dorn2021sharp} derived closed-form solution to
\eqref{equ:msm_opt} and \citet{dorn2021doubly} derived
EIFs of the optimal values. 
{\color{black}
\citet{chernozhukov2022long} proposed a sensitivity model that restricts the strength of unmeasured confounding through coefficients of determination in outcome regression and treatment assignment models. These coefficients are different from the measures of unmeasured confounding used in marginal sensitivity models, including our average-case model.
} 


Our article is most closely related to a new strand of literature that
considers generalizations to the worst-case marginal sensitivity
model. \citet{jin2022sensitivity} proposed to use $f$-divergence
to measure the difference between the distributions of $Y(z)$ given $X,Z=1-z$ and $Y(z)$ given $X,Z=z$, but they still consider a worst-case bound over $X$. They also did not consider
the boundedness constraint in \eqref{equ:add_1}.
\citet{ishikawa2023convex} further parameterized their sensitivity
model by the expected $f$-divergence and solved an optimization problem relaxed by Jensen's inequality.
In another related work,
\citet{huang2022variance} considered the confounding bias in
estimating the average treatment effect on the treated (ATT) under a sensitivity model constrained by 
\[
\Var\{ \omega (X)\mid Z=0\}/\Var\{  \omega (X,U) \mid Z=0\} \leq 1/(1- R^2),
\]
where $  \omega (X) =e(X)/[1-e(X)] $ and $  \omega (X,U) = e(X,U)/[1-e(X,U)]$.
The sensitivity parameter $R^2$ in their model can be understood as the proportion of variation in $W(X,U)$ that is not explained in $W(X)$. \citet{huang2022variance} derived an upper bound on the
confounding bias in this model, but the bound is not sharp.
Thus, all existing methods for generalized marginal sensitivity models relax some constraints to make the optimization problem tractable, resulting in non-sharp bounds. In contrast, our average-case sensitivity analysis problem in \eqref{equ:marginal_opt} takes the form of a quadratic program (over the function $h$) and admits closed-form solutions, as will be shown in \Cref{sect:section_2}. These solutions enable us to derive efficient estimators for our bounds in
\Cref{sect:stat_inference}. When closed-form solutions are unavailable, constructing efficient estimators based on numerical derivatives can introduce approximation errors; We refer readers to
\citet{jordan2022data} and the references therein (Section 3) for further discussion on this issue.

\section{Population solution to the average-case sensitivity
  model}\label{sect:section_2}

\subsection{Lagrangian formulation}\label{sect:second_formulation}

The worst-case optimization problem in \eqref{equ:msm_opt} is tractable, as it can be solved separately for each value of $X$. In comparison, our  optimization problem in \eqref{equ:marginal_opt} is more difficult to solve because the constraint on confounding strength in \eqref{equ:add_2} is marginalized over $X$, while the other
constraints need to hold for all values of $X$; see \Cref{sect:connections} for further discussion.

To address this, we consider two reformulations of
\eqref{equ:marginal_opt}. We will focus on the minimization problem below; the
maximization problem can be addressed similarly, as the objective function in \eqref{equ:marginal_opt} is a linear functional of $h$. Our results require the following assumption.

\begin{assumption}\label{assumption:first}
        The outcome $Y$ is a continuous, real-valued random variable
        with finite variance and a positive probability density function conditional on $X$ and $Z$ almost surely. It satisfies the
          consistency assumption: $Y =
          Y(z)$ if $Z=z$ for any treatment value $z\in \{0,1\}.$
          The propensity score $e(X) = \P(Z = 1 \mid X)$ satisfies the
          strong positivity/overlap assumption: $b \leq e(X)\leq 1-b$ for some constant
          $b\in (0,1/2)$.
\end{assumption}

Our first reformulation considers minimizing the Lagrange function
corresponding to \eqref{equ:marginal_opt}:
\begin{equation}\label{equ:double_ipw}
\begin{split}
\text{minimize}  & \quad \        \frac{1}{2 }\E\left\{ \nu_{1,h}(X)   \right\}+\lambda \E \left \{ \mu_{1,h}(X)  
                   \right \}  \\
 \text{subject to}  & \quad \ \E\left \{ h(X,Y) \mid X, Z = 1\right \} = 1, \\
 & \quad \ h(X,Y)\geq e(X).
\end{split}
\end{equation}
This resembles the classical solution to the portfolio problem in
finance \citep{markowitz7portfolio}, where $\E \left \{ \mu_{1,h}(X)  \right\}$ corresponds to the ``return'' of $h$ and $\E\left\{ \nu_{1,h}(X)   \right\}$ corresponds to the ``risk'' of
                 $h$. 
The next result provides the solution to the Lagrangian problem.

\begin{proposition}\label{prop:double_ipw}
Let Assumption \ref{assumption:first} be given. For a given value $\lambda > 0$, the
optimization problem \eqref{equ:double_ipw} is solved by
\begin{equation}\label{equ:solution_double_1_1}
 h_*(X,Y)  = e(X) + \lambda g(X,Y),~\text{where}~g(X,Y) =
 (\xi_X - Y ) \one_{ \{Y\leq \xi_X       \} },
\end{equation}
and $\xi_X$  is the unique root of the following strictly increasing
function:
\begin{equation}\label{equ:unique}
f_{\lambda,X}(\xi):=  \E \left\{ (\xi - Y ) \one_{ \{Y\leq \xi        \}
  } \mid X, Z = 1 \right\} - \{1-e(X)\}/\lambda.
\end{equation}
\end{proposition}
For the solution $h_*$ in \eqref{equ:solution_double_1_1}, define
the population-level ``sensitivity value'' and optimal value as
\begin{equation}\label{equ:psi_12}
\psi_1(\lambda) := \E\left\{ \nu_{h^*,1}(X) \right\} \ \text{ and }\ \psi_2(\lambda) := \E \left\{ \mu_{h^*,1}(X) \right\}.
\end{equation}
The ``sensitivity curve'' for $\E\{Y(1)\} $ is then
defined as $(\psi_1(\lambda ),\psi_2(\lambda ))$ for a range
of $\lambda$.

The next result confirms that the solution path of
\eqref{equ:double_ipw} with varying $\lambda$ recovers that of
\eqref{equ:marginal_opt} with varying average-case sensitivity
parameter $\Sigma$.

\begin{proposition}\label{prop:same_curve}
For any $\lambda > 0$, the optimal value of
\eqref{equ:marginal_opt} with $\Sigma = \psi_{1}(\lambda)$ is given
by $ \psi_{2}(\lambda).$
\end{proposition}


\subsection{Sensitivity value formulation}\label{sect:svf}

We next consider a different reformulation of
\eqref{equ:marginal_opt}:
\begin{equation}\label{equ:new_ipw}
\begin{split}
\text{minimize} & \quad \       \frac{1}{2} \nu_{1,h}(X) 
                   \\
 \text{subject to} & \quad \  \mu_{1,h}(X) 
                     \leq  \E(Y \mid X, Z = 1) - \theta,  \\
 & \quad \ \E\left\{h(X,Y) \mid X, Z = 1\right\} = 1, \\
 & \quad \ h(X,Y)\geq e(X).
\end{split}
\end{equation}
Heuristically, the parameter $\theta > 0$ bounds the
confounding bias and is chosen by the user. We call
\eqref{equ:new_ipw} the ``sensitivity value formulation'' because it
finds the minimum value of the average-case sensitivity parameter for which the confounding bias of the naive estimator $\E(Y \mid X, Z = 1)$
in estimating $\E\{Y(1) \mid X\}$ is uniformly bounded by $\theta$
across all values of $X$. Note that if $\theta\leq 0$, the problem in \eqref{equ:new_ipw} is solved
trivially by $h_*(X,Y)=1$.

\begin{proposition}\label{prop:solutions_new_1}
Let Assumption \ref{assumption:first} be given. For a given value
$\theta > 0$, the optimization problem \eqref{equ:new_ipw} is solved by
\begin{equation}\label{equ:solution_new_1_1}
 h_*(X,Y)  = e(X) + \lambda_X g(X, Y),
\end{equation}
where $g(X,Y) = (\xi_X - Y ) \one_{ \{Y\leq \xi_X       \} }$,
$\lambda_X = \{1-e(X)\}/\E\{g(X,Y) \mid X, Z = 1\}$, and
$\xi_X$ is the unique root of the following strictly increasing
function:
 \begin{equation}\label{equ:ratio}
f_{\theta,X}( \xi):= \frac{ \E \{(\xi -  Y) Y \cdot \one_{\{Y\leq \xi\}} \mid
  X, Z = 1 \} }{ \E \{ (\xi -  Y) \cdot \one_{\{Y\leq \xi\}} \mid X, Z
  = 1\} } - \E(Y \mid X, Z = 1) + \frac{\theta}{1-e(X)}.
\end{equation}
\end{proposition}

For the solution  $h_{*}$ in
\eqref{equ:solution_new_1_1}, we define the sensitivity value as
\begin{equation}\label{equ:psi_0}
\psi_3(\theta) :=\E[ \nu_{1,h_*}^2(X) ],
\end{equation}
which is the minimum value of the average-case sensitivity parameter, provided that
the constraints in \eqref{equ:new_ipw} hold for all values of $X$. 

As we vary $\theta$, the solution path of
\eqref{equ:new_ipw} will generally be different from that of
\eqref{equ:marginal_opt}. Nevertheless, the sensitivity value
reformulation can be useful if the observational study only involves a
unmeasured confounder $U$ that is independent of $X$ and has an
additive effect on $Y(1)$, or equivalently, if $Y(1) - \E\{Y(1) \mid
X\} \independent X$. In this case, the confounding bias $\E(Y \mid X,
Z = 1) - \E\{Y(1) \mid X\} = \theta$ does not depend on $X$, and
\eqref{equ:new_ipw} and \eqref{equ:marginal_opt} will have the same
solution path. This proves the following proposition.

\begin{proposition}\label{prop:indep}
Suppose $Y(1) - \E\{Y(1) \mid X\} \independent X$.
For any $\theta > 0$, the optimal value of
\eqref{equ:marginal_opt} with the average-case sensitivity parameter
$\Sigma = \psi_{3}(\theta)$ is $\E\left\{ \E\left[Y \mid X,Z=1 \right]\right\}
- \theta$.
\end{proposition}

\subsection{Worst-case vs.\ average-case sensitivity analysis}\label{sect:connections}

We next compare our solution to the reformulations of the
average-case problem in \eqref{equ:marginal_opt} with that to the
worst-case problem in \eqref{prop:quantile_balancing}. \citet{dorn2021sharp} obtained
closed-form solution to \eqref{prop:quantile_balancing} by rewriting
the constraint in \eqref{equ:lambda_e} as
\begin{equation}\label{equ:lambda_h}
 W_{-}(X)  \leq h(X,Y) \leq  W_{+}(X),
\end{equation}
where $W_{-}(X) := (1-1/\Gamma)e(x)+1/\Gamma$ and $W_{+}(X) := (1-\Gamma)e(X)+\Gamma.$
With this, \eqref{equ:msm_opt} becomes a linear program that can be
solved using the Neyman-Pearson Lemma \citep{neyman1933ix} after
the transformation $h(X,Y)$ to $\{h(X,Y) - W_-(X) \}/ \{W_+(x) -
W_-(X) \}$. We restate their result here and give an alternative proof
in the Supplementary Materials for completeness using the
Karush-Kuhn-Tucker condition for optimiality.


\begin{proposition}\label{prop:quantile_balancing}
  The maximization problem in \eqref{equ:msm_opt} is solved by
\begin{equation}\label{equ:h_0}
 h_*(X,Y) =  \begin{cases}
                W_{-}(X),        \text{ if } Y< Q(X),
                   \\
                W_{+}(X),               \text{ if } Y> Q(X),
 \end{cases}
\end{equation}
where $Q(X)$ is the $\Gamma/(1+\Gamma)$-quantile of $Y$ given $X$ and $Z = 1$,
\[
  Q(X) = \inf\left\{y\in \mathcal{Y}: \P(Y \leq y \mid X, Z = 1) \geq
    \Gamma/(1+\Gamma)\right\}.
\]
The minimization problem is solved by \eqref{equ:h_0} after swapping
$W_{-}$ and $W_{+}$ and using the $1/(1+\Gamma)$-quantile
of $Y$ given $X$ and $Z = 1$ as $Q(X)$.
\end{proposition}

\begin{figure}[t]
\centering
\begin{subfigure}[t]{.325\textwidth}
    \centering
    \includegraphics[width=.95\linewidth]{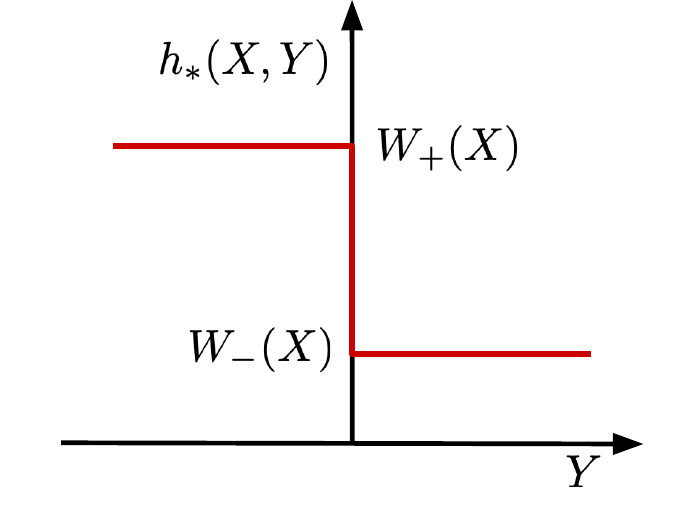}
    \caption{Solution to \eqref{equ:msm_opt}.}
    \label{subfig:sols_1_1}
  \end{subfigure}
  \begin{subfigure}[t]{.325\textwidth}
    \centering
    \includegraphics[width=.95\linewidth]{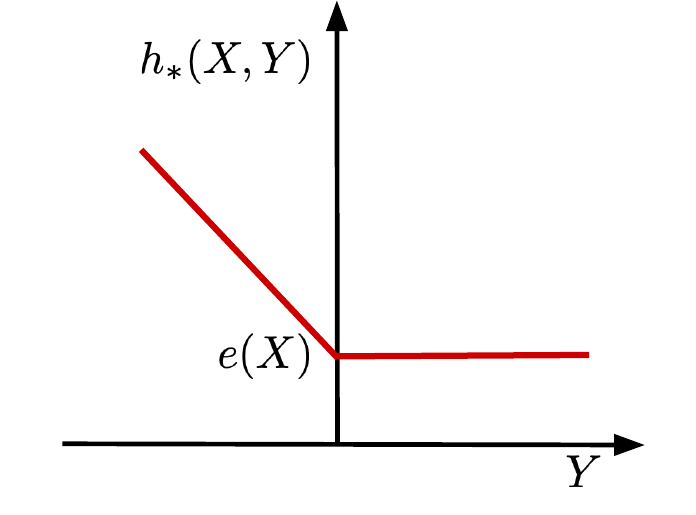}
   \caption{Solution to \eqref{equ:double_ipw} or \eqref{equ:new_ipw}.}
    \label{subfig:sols_1_2}
\end{subfigure}
\begin{subfigure}[t]{.325\textwidth}
    \centering
    \includegraphics[width=.95\linewidth]{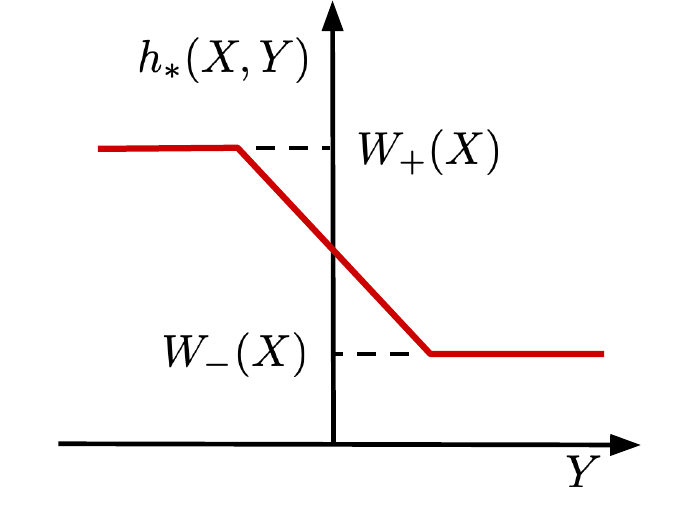}
    \caption{Solution to \eqref{equ:new_ipw} with
        \eqref{equ:lambda_h}.}
    \label{subfig:sols_1_3}
    \end{subfigure}
\caption{Schematic comparison of solutions to different marginal
  sensitivity analysis (minimization) problems. The thresholds
  $W_+(X)$ and $W_-(X)$ are given below \eqref{equ:lambda_h}.
  }
\label{fig:sols_}
\end{figure}


\Cref{fig:sols_} gives a schematic comparison of the
solutions. The optimal solution to the worst-case problem, as given by
\eqref{equ:h_0}, is a piece-wise constant function of $Y$
(\Cref{subfig:sols_1_1}). In
comparison, the optimal solution to the average-case problem, as given by \eqref{equ:solution_double_1_1} or \eqref{equ:solution_new_1_1}, is a
piece-wise linear function of $Y$
(\Cref{subfig:sols_1_2}). Moreover, although we do not pursue this generalization here, we prove in Appendix \ref{appendix:solutions_new_1} that if both average-case and worst-case constraints are included in optimization, the solution is still piecewise constant (\Cref{subfig:sols_1_3}). The last solution is reminiscent of the winsorizing technique
\citep{dixon1960simplified} and the derivative of Huber's loss
\citep{huber1964}.

As mentioned in \Cref{sec:related-works}, existing works in the sensitivity analysis literature sometimes drop or relax certain
constraints in the model, leading to non-sharp bounds.
Since our bounds are based on closed-form solutions, it is not
difficult to verify sharpness of the corresponding bounds following
the idea in Proposition 5 of \citet{dorn2021sharp} and Chapter 5 of
\citet{manski2003partial}.

Before stating the sharpness result formally, let us first explain how
our methods can be applied to get bounds for $\E\{Y(0)\}$ and the ATE
$\E\{Y(1) - Y(0)\}$. To bound $\E\{Y(0)\}$, we use a slightly modified
average-case marginal sensitivity model in which the function
$h'(X,Y(0)) = (1 - e(X)) / (1 - e(X,Y(0)))$ for $e(X,Y(0)) = \P(Z = 1 \mid X,
Y(0))$ (note the abuse of notation) is required to satisfy the
marginalization constraint \eqref{equ:mean_1} and counterparts of \eqref{equ:add_1} and
\eqref{equ:add_2}.
The bounds on
$\E\{Y(0)\}$ can be obtained in a same way as in
\Cref{sect:second_formulation,sect:svf}.
Subsequently, a lower bound for the ATE can be obtained by
subtracting the upper bound of $\E\{Y(0)\}$ from the lower bound
of $\E\{Y(1)\}$; an upper bound for the ATE can be obtained
similarly.

The next result confirms this for our Lagrangian formulation indeed
obtains sharp lower and upper bounds for $\E\{Y(1)\}$ and
$\E\{Y(0)\}$, by showing that any value between these bounds can be
attained. When $Z=0$, $h(X,Y(0))$ refers to the ratio $[1-e(X)]/[1-\mathbb{P}\{Z=1  \mid X,Y(0)\}]$.
A similar result for the sensitivity value formulation can be found in Appendix \ref{appendix:sharpcate}.

\begin{proposition}\label{prop:sharp}
Under Assumption \ref{assumption:first}, for $z\in \{0,1\}$ and any distribution $\PP$ with $h(X,Y(z))$ satisfying the constraints in \eqref{equ:double_ipw},
there exists a distribution $\tilde \PP$ of $(X,Z,Y,Y(z))$ satisfying that
\begin{enumerate}[label=(\arabic*)]
                \item its marginal distribution of
               $(X,Y,Z)$ matches the
                   data distribution under $\PP $.
                \item its propensity score ratios
                  satisfy all constraints in the program
                  \eqref{equ:double_ipw}.
                \item the expectation of $Y(z)$ under $\tilde{\PP}$ is
                  equal to that under $\PP$.
        \end{enumerate}
\end{proposition}

\Cref{prop:sharp} 
confirms that by including all essential constraints in optimization, the average-case sensitivity model leads to sharp bounds for $\E[Y(0)]$ and $\E[Y(1)]$, respectively.
In the proof, we show that the lower bound of $\mathbb E[Y(1)]$ is obtained by lower bounding the true counterfactual outcome by the truncated factual outcome as follows:
\[
\mathbb E[Y(1)\mid X,Z=0] \geq \E[g_*(X,Y)Y\mid X,Z=1],
\]
where $g_*(X,Y) = g(X,Y)/\E[g(X,Y)\mid X,Z=1]$
for the function $g$ defined in  \eqref{equ:solution_double_1_1}. 
Under the worst-case sensitivity model, the sharp bound of $\E[Y(1)]$ is  obtained similarly:
\[
\mathbb E[Y(1)\mid X,Z=0] \geq    \E\big[\Gamma^{\text{sign} \{Q(X)-Y\}   } Y \mid X,Z=1 \big].
\]
\citet[Theorem 2]{dorn2021sharp} showed that the bounds for
$\E[Y(0)]$ and $\E[Y(1)]$ are simultaneously attainable under the worst-case model, thereby yielding sharp bounds for the ATE $ = \E[Y(1) - Y(0)]$. However, this may not hold for the average-case model considered here. To see the difference, note that the weighting function in the worst-case model,
$\Gamma^{\text{sign} \{Q(X)-Y\}   } $, takes on only two extreme values, either $\Gamma$ or $1/\Gamma$. In contrast, the weight function $g_*(X,Y)$ in the average-case model varies with $Y$ continuously, making it more difficult to establish sharpness for the ATE bound.  A similar limitation was noted by \citet[Section 3.4]{yadlowsky2022bounds} in their method based on Rosenbaum’s sensitivity model.

\section{Estimation of bounds and sensitivity
  values}\label{sect:stat_inference}

\subsection{Setup \& Notation}\label{sect:notation_if}

We next introduce our one-step estimators for the population-level bounds and
sensitivity values introduced in the previous section. More specifically, we
will consider estimating $\psi_1(\lambda)$ and $\psi_2(\lambda)$, as
defined in \eqref{equ:psi_12} for the Lagrangian formulation of the
average-case problem, and $\psi_3(\theta)$, as defined in
\eqref{equ:psi_0} for the sensitivity value formulation.
We also include analogous results for the worst-case sensitivity model in \Cref{sect:additional}. After describing the estimators, we show how to use them in multiplier bootstrap to construct confidence bands for sensitivity curves.

In what follows, we consider an observational study with
$n$ units, where the observed data  $O_{[n]} := (O_1, \dotsc,O_n)$ are assumed to be i.i.d. draws from the same distribution $\PP$ as the random variable $O = (X,Y,Z)\in \O$ which satisfies Assumption \ref{assumption:first}.

Using these observations, we can directly estimate our target parameters, e.g. $\psi_1(\lambda)$, as follows: first estimate the function $h_*^2$ in \eqref{equ:psi_12} nonparametrically and then compute a sample-average estimator for its expectation conditional on $Z=1$.
We call this the direct plug-in estimator of $\psi_1(\lambda)$.
Unfortunately, it is known from nonparametric regression
\citep{wasserman2006all} that the estimator of $h_*^2$ would converge at a rate slower than $n^{-1/2}.$  Consequently, the sample-average estimator is unable to achieve the central limit theorem (CLT) for asymptotic inference because its bias goes to 0 slower than $n^{-1/2}.$ One-step estimation is a well-known technique \citep{bickel1993efficient,kennedy2022semiparametric} to address this problem by debiasing estimators and relaxing the required rate conditions for inference.
We refer to \citet{kennedy2022semiparametric} for a brief introduction to one-step estimation.
Here, we directly introduce one-step estimation with some new notation.

In this article, our one-step estimator is written as a sample average of the uncentered efficient influence function (EIF) $\phi_*$ of a target parameter $\psi_* = \psi_*(\beta)$, where $\beta$ is a pre-specified sensitivity parameter.
The EIF $\phi_*(O) = \phi_*(O;\beta,\eta)$ is a real-valued function of $O$, where $\eta$ denote the nuisance parameters involved, such as $h_*^2$ mentioned earlier. We would keep the parameters $\beta$ and/or $\eta$ implicit in our notation when they are irrelevant for our discussion.

Suppose we divide the $n$ observations into $K$ disjoint folds, with every fold consisting of observations from $m=n/k$ units.
Let $\PP_{n}^{(k)}$ denote the empirical measure of the data in the $k$-th fold, and $\hat \eta_{-k}$ denote the estimator of $\eta$ fitted to the data in the other folds, $[K]\setminus \{k\}.$
Let $\hat \eta \equiv\eta_{-K}$ and
$\hat \phi_* (O_i) \equiv \phi_* (O_i;\beta,\hat \eta)$.
The one-step estimator of $\psi_*$ is given by
\begin{equation}\label{equ:one_estimator}
\hat \psi_* := \PP_{n}^{(K)} \hat \phi_*  = m^{-1}\sum_{i=n-m+1}^{n} \hat \phi_* (O_i).
\end{equation}
To remedy the efficiency loss due to sample-splitting,  we can use the popular $K$-fold cross-fitting strategy in \citet{schick1986asymptotically,chernozhukov2018double}. The cross-fitted estimator  of $\hat \psi_*$ is obtained by averaging $K$ cross-fitted one-step estimators:
\begin{equation}\label{equ:cf_estimator}
\hat \psi_{*,\text{cf}} :=  \frac{1}{K}\sum_{k=1}^{K}\hat \psi^{(k)} = \frac{1}{K}\sum_{k=1}^{K} \PP_n^{(k)} \hat \phi_*^{(k)}  = \PP_n \hat \phi_{*,\text{cf}}.
\end{equation}
This cross-fitted estimator can be written as an average of  $\hat \phi_*(O_1; \hat \eta_{-k_1} ),\dotsc, \hat \phi_*(O_n; \hat \eta_{-k_n} ),$ where $k_i$ is the fold that contains $i$-th observation $O_i$, i.e.,  the fold that is not used by $\hat \eta_{-k_i}.$ This average is denoted by $\PP_n \hat \phi_{*,\text{cf}}$ above.
Cross-fitting allows us to make use of all observations in computing a single average.
Under some achievable rate conditions to be introduced below,
$\hat \psi_{*,\text{cf}} $ can be a root-$n$ consistent and asymptotically normal (CAN) estimator,
\begin{equation}\label{equ:can}
\sqrt{n}\big(\hat \psi_{*,\text{cf}} - \psi_*\big) \xrightarrow[]{d} \mathcal{N}\left(0,\sigma^2:=\Var\left[\phi_*(O)\right]\right).
\end{equation}
The uncentered EIF $\phi_*$ has the lowest variance in the CLT, which implies that
$\hat \psi_{*,\text{cf}}$ is an efficient estimator of $\psi$ \citep{bickel1993efficient,tsiatis2006semiparametric}.

\subsection{One-step estimation: average-case model }\label{sect:estimation_l_2}

The construction of one-step and cross-fitted estimators in \eqref{equ:one_estimator} and \eqref{equ:cf_estimator} is the same across all parameters. Thus, our introduction below will
focus on the EIFs and the nuisance parameters. We present the EIFs for
the bounds and sensitivity values under the average-case sensitivity model.

\begin{theorem}\label{prop:eif_small}
The uncentered EIFs of $\psi_1$ and $\psi_2$ in \eqref{equ:psi_12}  are given by
\begin{align*}
\phi_1(O) = &\          \frac{Z}{e(X)}\left\{    2[1-e(X)]\Pi_h(X,Y)+
              h_*^2  (X,Y)- \E\left[h_*^2(X,Y) \mid X, Z = 1\right]
              \right\}          \\[2pt]
& + 2[1-e(X)][Z-e(X)+ \Pi_e(X,Y)]  +\E[h_*^2(X,Y) \mid X, Z = 1] \  \text{ and}    \\[5pt]
 \vspace{5pt}
\phi_2(O) = &\   \frac{Z}{ e(X) }\big\{    \Pi_h (X,Y)  \mathbb
              E  \big[Y \one_{\{Y\leq \xi_X \}} \mid X,
              Z = 1  \big]    - \mathbb E \big[ h_*(X,Y) Y \mid X, Z =
              1 \big]         \\[2pt]
&\  +  h_*(X,Y)Y   \big\}  + \Pi_e (X,Y)  \mathbb E  \big[Y
  \one_{\{Y\leq \xi_X \}} \mid X, Z = 1  \big]  \\[3pt]
&\ +[Z-e(X)]\mathbb E [Y \mid X, Z = 1] +  \mathbb E  \big[ h_*(X,Y) Y
  \mid X, Z = 1\big],
\end{align*}
where $\Pi_e(X,Y)$ and $\Pi_h(X,Y)$ are two mean-zero random variables
defined as 
\[
\Pi_e(X,Y)       = \frac{        e(X)-Z }{ \PP (
    Y\leq  \xi_X \mid X, Z = 1) }
\ \text{ and }\
\Pi_h(X,Y)       = \frac{  1 -  h_*(X,Y)                }{ \mathbb
  \PP ( Y\leq  \xi_X \mid X, Z = 1) }.
\]
\end{theorem}

\begin{theorem}\label{prop:eif_large}
The uncentered EIF of $\psi_3$ in \eqref{equ:psi_0} is given by
\begin{align*}
\phi_3(O) =  & -  \frac{Z}{e(X)}           \Big\{      2\lambda_X
               [1-e(X)] \left(Y - \E(Y \mid X, Z = 1)        \right)  +      \lambda_X^2\big( g^2(X,Y) \\
&+  \E[g^2(X,Y) \mid X, Z = 1] \big)\Big\} + 2\lambda_X[Z-e(X)]
  \E\left[(Y-\xi_X) \one_{\{Y > \xi_X \}} \mid X, Z = 1\right]   \\
&                  + \E[h_*^2(X,Y) \mid X, Z = 1].
\end{align*}
\end{theorem}
To use the EIFs above, we estimate the expectations of $Y, \one_{\{Y\leq \xi_X \}},$
$Y\one_{\{Y\leq \xi_X \}}$ and $Y^2 \one_{\{Y\leq \xi_X \}}$
conditional on $X$ and $Z=1$, using a conditional density estimator $\hat p_{Y\mid X,Z=1}$.
For example, we estimate the second moment term by $\int y^2 \one_{\{y \leq
 \hat  \xi_X\}} \hat p_{Y\mid X,Z=1}(y)  dy.$ In the literature, there are many advanced methods
\citep{chernozhukov2013inference,belloni2019conditional}
and models \citep{meinshausen2006quantile,friedman2020contrast} to estimate conditional densities. For simplicity, here we consider the Nadaraya-Watson (NW) kernel estimator,
\begin{equation}\label{equ:kde}
\hat p_{Y\mid X,Z=1}(y )    = \frac{\sum_{i=1}^{n-m}Z_i K_{1}(X-X_i)K_{2}(y -Y_i)  }{\sum_{j=1}^{n-m}Z_j K_{1}(X-X_j)   },
\end{equation}
where $K_{1}$ and
$K_{2}$ are two continuous and
nonnegative kernel functions.
Alternatively, we can use an additive regression model: given $Z_i=1,$
$Y_i = f(X_i)+\epsilon_i\sigma^2(X_i),$
where $\epsilon_i\sim \mathcal N(0,1).$ We can first fit a model $\hat f$ to estimate $f$.
Based on a mean estimate $\hat f$ and a variance estimate using the residuals, we can express the (truncated) first and second moments of the Gaussian outcome $Y$ given $X,Z=1$ in closed form. We can also combine both techniques by modelling the density of the residual $Y_i - \hat{f}(X_i)$ using the NW estimator; see Section 4 of \citet{hansen2004nonparametric} and Section 3 of \citet{fan1998efficient} for this two-step modelling approach.

To estimate the root $\xi_X$, we first estimate $f_{\lambda,X}( \xi )$ in \eqref{equ:unique} or $f_{\theta,X }(\xi)$ in \eqref{equ:ratio}
using $\hat e$ and the moments estimated by $\hat p_{Y\mid X,Z=1}$.
The estimator is an increasing function of $\xi $ when  $\hat p_{Y\mid X,Z=1}>0.$ We
can find $\hat \xi_{X}$ via the bisection method; the root-finding error decays exponentially fast with the number of iterations. The results below show that one-step estimation can improve the accuracy
of the direct plug-in estimator that converges at a rate slower than $n^{-1/2}.$

\begin{assumption}\label{assumption:nuisance_error}
With probability 1, $\hat e\in (0,1)$ and $\hat p_{Y\mid X,Z=1}:\mathcal{X}\times \mathcal{Y}\rightarrow(0,\infty)$ is a bounded and continuous function. The nuisance estimator  $\hat \eta = (\hat e, \hat p_{Y\mid X,Z=1}  )$ satisfies that
\[ \ \  \| \hat e - e\|_{\infty} = o_{\PP}(n^{-1/4})  \ \text{ and } \
 \| \hat p_{Y\mid X,Z=1} - p_{Y\mid X,Z=1} \|_{\infty} =o_{\PP}(n^{-1/4}).
\]
\end{assumption}
\begin{proposition}\label{thm:double_bias}
Under Assumptions \ref{assumption:first} and \ref{assumption:nuisance_error},  $\Bias(\hat \phi_j \mid \hat \eta) =o_{\PP}(n^{-1/2}) $ for $j=1,2,3.$
\end{proposition}
The bias function in \Cref{thm:double_bias} is given by $\E[ \psi(O_i);\hat \eta) - \psi(O_i;\eta) \mid \hat \eta ]$.
where $O_i$ is an observation from the $K$-th fold and $\hat \eta$ is the nuisance estimator fitted to other folds, as mentioned above \eqref{equ:one_estimator}. When the rate condition above is satisfied, the cross-fitted estimator $\hat \psi_{*,\text{cf}}$ in \eqref{equ:cf_estimator} can obtain the CLT in \eqref{equ:can}, allowing us to define a valid (1-$\alpha$)-confidence interval (CI) for $\psi_{*}(\beta)$ as
\begin{equation}\label{equ:ci_pointwise}
\hat C_{*,\text{cf}}(\beta) = \big[\hat \psi_{*,\text{cf}}(\beta)\pm z_{\frac{\alpha}{2}}\hat \sigma_{*,\text{cf}}(\beta)/\sqrt n \big],
\end{equation}
where $\hat \sigma_{*,\text{cf}}^2(\beta)$ is the cross-fitted estimator of the variance of the EIF $\phi(O;\beta).$

\subsection{One-step estimation: worst-case model}\label{sect:additional}

We next describe our one-step estimation result for
the upper bound for $\E[Y(1)]$ derived under the worst-case sensitivity model in \eqref{equ:h_0}.
This bound can be written as  $\psi= \psi_+  +  \psi_-$ with
\begin{equation}\label{equ:psi}
\psi_+  := \E\big[W_{+}(X) \mu_{+}(X)\big] \quad \text{and} \quad   \psi_-  :=    \E\big[ W_-(X)        \mu_{-}(X) \big],
\end{equation}
For a fixed value of $X$, the quantities $\mu_{+}(X)$ and $\mu_{-}(X) $ are expected shortfalls, defined as
\[
  \mu_{+}(X) = \E\big[Y \one_{\{Y> Q(X)\} } \mid X, Z = 1 \big]
\quad \text{and} \quad \mu_{-}(X) = \E\big[Y \one_{\{Y<
  Q(X)\}} \mid X, Z = 1 \big],
\]
where $Q(X)$ is defined in \Cref{prop:quantile_balancing}.
Denote $\eta  = (e,  Q,\mu_+, \mu_-)$ and their estimators $\hat \eta = (\hat e,\hat Q,\hat{\hat{\mu}}_+,\hat{\hat{\mu}}_- )$ fitted to the first $K-1$ folds of the observations; we estimate $W_{+}$ and $W_{-}$ following their definitions below \eqref{equ:lambda_h}.
We use the following EIFs for estimating $\psi= \psi_+  +  \psi_-$.
\begin{theorem}\label{thm:if_psi_plus}
The uncentered EIF of $\psi$ is given by $\phi(O) = \phi_+(O) +\phi_-(O)$ with
\begin{align*}
\phi_+(O) = &\ \frac{Z W_+(X)}{e(X)}\left[\left(1-\alpha_* - \one_{\{Y>Q(X)\}} \right)Q(X) + Y  \one_{\{Y>Q(X)\}}- \mu_+(X) \right]  \\
&\ + \left[(1-\Gamma)Z+\Gamma \right]\mu_+(X)\quad \text{and}\\
\phi_-(O)  =  &\  \frac{Z W_{-}(X)}{e(X)}\left[ \left(\alpha_* - \one_{\{Y<Q(X)\}} \right)Q(X) + Y  \one_{\{Y<Q(X)\}} - \mu_-(X) \right] \\
&\ + \left[(1-\Gamma^{-1})Z+\Gamma^{-1} \right]\mu_-(X).
\end{align*}
\end{theorem}

\citet{dorn2021doubly}  proposed a \quotes{doubly-valid/doubly-sharp} (DVDS) estimator of their sharp bound under a distributional shift
formulation of the worst-case sensitivity model $\mathcal H_{\text{wst}}$ in \eqref{equ:h_infinity}. Assuming a weighted linear outcome quantile model, \citet{tan2022model} developed relaxed population bounds under $\mathcal H_{\text{wst}}$ along with their doubly robust estimators.
In contrast, our EIFs above are derived under the original formulation of $\mathcal H_{\text{wst}}$, involving slightly different nuisance parameters and requiring distinct proofs for the theoretical result below.

\begin{assumption}\label{assumption:consistency}
With probability 1,  $(Q,\mu_{+},\mu_{-},\hat Q,\hat{\hat{\mu}}_{+},\hat{\hat{\mu}}_{-})$ are bounded and $\hat e\in (0,1)$. Furthermore, the errors
$\|Q -\hat Q \|$, $\| \hat{\hat{\mu}}_{+} -\mu_{+} \|$, $\|\hat{\hat{\mu}}_{-} -\mu_{-} \|$ and $\|\hat e -\hat e \|$
are $o_{\PP}(n^{-1/4})$.
\end{assumption}

\begin{proposition}\label{thm:bias}
Under Assumptions \ref{assumption:first} and \ref{assumption:consistency}, $ \Bias(\hat \phi \mid \hat \eta)  = o_{\PP}(n^{-1/2})$.
\end{proposition}
Similar to \Cref{thm:double_bias}, \Cref{thm:bias} implies that when the rate condition above is satisfied,
we can define a valid (1-$\alpha$)-confidence interval (CI) for $\psi$ as in \eqref{equ:ci_pointwise}.

\subsection{Simultaneous confidence bands}\label{sect:uni}

In sensitivity analysis, it is often desirable to report multiple bounds under
different levels of unmeasured confounding.
\citet{belloni2018uniformly, bonvini2021sensitivity} proposed an inference procedure based on multiplier bootstrap (MB) \citep{gine1984some,vaart1996weak} to construct simultaneous confidence bands for sensitivity curves, e.g., sequences of bounds. Here we apply MB to construct confidence bands for the curves,
\begin{equation}\label{equ:curves}
    \begin{split}
        \Psi(\mathcal{D}) &:= \left\{(\Gamma, \psi(\Gamma))      : \Gamma \in \mathcal{D}\subset [1,\infty)\right\}, \\
    \Psi_{12}(\mathcal{D}_{12}) &:=\{(\psi_{1}(\lambda),\psi_2(\lambda)):\lambda\in \mathcal{D}_1\subset [0,\infty)\}, \\
        \Psi_3(\mathcal{D}_{3}) &:=\{(\psi_{3}(\theta),\E\left\{ \E\left[Y \mid X,Z=1 \right]\right\}-\theta) :\theta\in \mathcal{D}_3\subset [0,\infty)\},
    \end{split}
\end{equation}
for some given ranges $\mathcal{D}$'s.
\citet[Theorems 3 and 4]{kennedy2019nonparametric}
proved the validity of MB for influence function-based estimators. We prove similar theoretical results in  \Cref{sect:theory.mb}, verifying the regularity conditions required to apply MB in our setting.

Here we describe the procedure of MB for constructing a confidence band for the sensitivity curve $ \Psi_*(\mathcal{D}_*) := \left\{(\beta, \psi_*(\beta)) : \beta \in \mathcal{D}_*\right\}$.
This notation follows from the one in \Cref{sect:notation_if}.

We first note that taking a union of the CIs
$\hat C_{*,\text{cf}}(\beta)$ in \eqref{equ:ci_pointwise}
for all $\beta\in \mathcal{D}_*$
may not yield a valid confidence band for $ \Psi_*(\mathcal{D}_*)$ because these CIs lack uniform validity.
To address this, MB generalizes the Gaussian approximation based on
the CLT in \eqref{equ:can} as follows. It first approximates the distribution of the supremum of Gaussian process $|\mathbb {G}_*(\beta)| := \left|\phi_*(O;\beta ) - \psi_*(\beta )\right|/ \sigma_*(\beta )$ for all $\beta\in \mathcal D_*$. Then it increases the $z$-score in the pointwise CIs to a critical value $\hat q_{0,\alpha}$ that upper bounds the supremum with probability $1-\alpha$. Formally, $\hat q_{*,\alpha}$ is defined as the $(1-\alpha)$-quantile of the supremum of the multiplier bootstrap process as follows:
\[
\P\bigg\{\sup_{\beta\in \mathcal{D}_*} \left|\sqrt{n}\PP_n \left[ A \left(\hat \phi_{*,\text{cf}}(\beta) -  \hat  \psi_{*,\text{cf}}(\beta)          \right)/\hat \sigma_{*,\text{cf}}(\beta) \right]                   \right|         \geq \hat q_{*,\alpha} \        \Big| \         O_{[n]} \bigg\} = \alpha,
\]
where the average under $\PP_n$ is defined in the same way as the one in \eqref{equ:cf_estimator}, and $A_{[n]}$ are i.i.d Rademacher variables drawn independently of $O_{[n]}$. This equation means that
 after scaling by $\sqrt{n}$, the supremum of the empirical average of $A$ multiplied by the normalized influence function is larger than $\hat q_{0,\alpha} $ with probability $\alpha$.
Taking the union of the CIs in \eqref{equ:ci_pointwise} using this new critical value $\hat q_{0,\alpha}$ leads to an asymptotically valid confidence band for $\Psi_*(\mathcal{D}_*).$ When the sensitivity curves involve two unknown parameters for the same $\beta$, e.g., $\Psi(\mathcal{D}_{12})$, we combine the lower confidence band
$\big(\hat \psi_{1,\text{cf}}(\lambda )- \hat q_{1,\alpha}\hat \sigma_{1,\text{cf}}(\lambda)/\sqrt n,
\hat \psi_{2,\text{cf}}(\lambda )- \hat q_{2,\alpha}\hat \sigma_{2,\text{cf}}(\lambda)/\sqrt n \big)$ for all $\lambda$
to construct a valid lower confidence band for $\Psi_{12}(\mathcal{D}_{12}).$ We can compute the band for $\Psi_{3}(\mathcal{D}_{3})$ similarly.

\section{Simulation study}\label{sect:simulation}

In this section, we examine the finite-sample performance
of our proposed estimators and confidence bands on i.i.d. data ({$n=300,400,500$}) simulated as follows:
{ \begin{align*}
&  U_i \sim \mathcal{N}(0,1), \ X_{i,j}\sim \mathcal{N}_{[-1,1]}(0,1),\ j =1,\dots,10, \\
& e(X_i) = 1/[1+\exp(-X_{i,1} -X_{i,1}^2 ) ],\  Z_i \sim  \text{Bern}(e(X_i)),     \\
&  Y_i(0) = X_{i,1} + U_i + \one_{\{X_{i,1}>0\}}U_i,\ Y_i(1) = Y_i(0) + 0.5,\ Y_i = Z_i Y_i(1) + (1-Z_i)Y_i(0),
\end{align*}}
where $\mathcal{N}_{[-1,1]}(0,1)$ is the standard normal distribution
truncated to $[-1,1].$ 

The true propensity score $e(X_i)$ has a quadratic term, and the outcomes are simulated with heteroscedastic noise. We estimate the propensity score using a logistic regression model without the quadratic term. We estimate the other nuisance parameter $Q$, $\mu_{+}, \mu_-$ and $p_{Y\mid X,Z=1}$ using an additive regression model with Gaussian error, as described below \eqref{equ:kde}.
These estimators are misspecified slightly, which allows us to investigate the advantages of our proposed one-step estimators compared to the direct plug-in estimators without using EIFs. We use 10-fold cross-fitting in both cases.  We consider estimating all the parameters $\psi$ and $\psi_j$ for $j\in [3].$ We also include the one-step estimators using (part of) the true nuisance parameters, which serve as ``oracle'' estimators with better efficiency; the true nuisance parameters are denoted in the second row of \Cref{tab:rmse}. We assess the estimators by their root-mean-squared errors over 500 runs, e.g., the RMSE of the estimator $\hat \psi_{1,\text{cf}}(\lambda )$ is given by
$\text{RMSE}\big\{
\hat \psi_{1,\text{cf}}(\lambda ) \big\} = \big\{\frac{1}{500}\sum_{j=1}^{500} \big(\hat\psi_{1,\text{cf}}^{(j)}(\lambda ) - \psi_1(\lambda ) \big) \big\}^{1/2}.$
All the results are reported in \Cref{tab:rmse} below. Comparing the second and third columns, we can see that one-step estimators converge faster than direct plug-in estimators.
In the other columns, the RMSEs generally drop as we plug in the true nuisance parameter. These results show that the EIFs we derive can improve efficiency, confirming our theoretical results in  Propositions \ref{thm:double_bias} and \ref{thm:bias}.

\begin{table}[t]
    \caption{RMSEs of direct and one-step estimators over 500 simulations.}
    \label{tab:rmse}
\begin{center}
\resizebox{0.58\columnwidth}{!}{
    \begin{tabular}{ |c|c|c|c|c| }
\hline
\rule{0pt}{10pt}  Methods   & Direct  &   \multicolumn{3}{c|}{One-step (Ours)} \\              [1pt]
\hline
\multirow{1}{*}{Nuisance}  & \multirow{1}{*}{ $\hat  e, \hat p_{Y\mid X,Z=1}$}
      & \multirow{1}{*}{$\hat  e, \hat p_{Y\mid X,Z=1}$}     & \multirow{1}{*}{$\hat  e, p_{Y\mid X,Z=1}$ }        & \multirow{1}{*}{$e, p_{Y\mid X,Z=1}$ } \\
\hline
\rule{0pt}{10pt} $n$   &
    \multicolumn{4}{c|}{$\Gamma = 5,\ \psi(\Gamma)= 1.224$} \\     [1.0pt]
\hline
300 &  0.506 &  0.204 &   0.154  &  0.142 \\
400 &  0.491 & 0.185 &   0.126  & 0.122 \\
500 &  0.504 &  0.183  &   0.113 &  0.109 \\
\hline
\rule{0pt}{10pt}      $n$  &
    \multicolumn{4}{c|}{$\lambda  = 1,\ \psi_1(\lambda)=1.509$}\\     [1pt]
\hline
300   & 0.376 & 0.150   & 0.114 & 0.103  \\
400   & 0.385 &  0.142 & 0.100 & 0.093 \\
500 & 0.396 & 0.146  &   0.086 &  0.079 \\
\hline
\rule{0pt}{10pt}     $n$   &
    \multicolumn{4}{c|}{$\lambda  = 1,\ \psi_2(\lambda)=-0.334$}\\     [1pt]
\hline
300   & 0.735 & 0.347  & 0.157 & 0.149  \\
400   & 0.761 &  0.356 & 0.133 & 0.136 \\
500 &  0.769 & 0.351 &   0.121 &  0.120 \\
\hline
\rule{0pt}{10pt}     $n$   &
    \multicolumn{4}{c|}{$\theta  = 0.5,\ \psi_3(\theta)=1.179$}\\     [1pt]
\hline
300   & 0.127 & 0.107   & 0.045 & 0.035 \\
400   & 0.128 & 0.105 & 0.040 & 0.030 \\
500 &  0.129 &  0.104   &  0.034 &  0.027 \\
\hline
    \end{tabular}
 }
    \end{center}
\end{table}

\begin{table}[t]
    \caption{Simultaneous coverage rates with and without multiplier
      bootstrap (MB), where $\hat{q}_{\alpha},\hat{q}_{1,\alpha},\hat{q}_{2,\alpha}$ and $\hat{q}_{3,\alpha}$ are the critical values chosen by MB with 2500 bootstrap samples. The results are averaged over 500 simulations.}
    \label{tab:uniform_cover}
\begin{center}
\resizebox{0.63\columnwidth}{!}{%
    \begin{tabular}{ |c|c|c|c|c|c|c| }
\hline
Domain  &
    \multicolumn{3}{c|}{$\{ \psi(\Gamma) :\Gamma\in \mathcal{D}\}$} &     \multicolumn{3}{c|}{$ \{\psi_1(\lambda): \lambda \in \mathcal{D}_{12} \}$}                                                                                     \\ [1pt]
    \hline
MB    &  \xmark & \multicolumn{2}{c|}{\cmark }  &    \xmark & \multicolumn{2}{c|}{\cmark } \\
\hline
$n $     &  Coverage &  $\hat q_{\alpha}$  &    Coverage & Coverage &  $\hat q_{1,\alpha}$  &    Coverage \\
\hline
300 &  0.884 & 2.454 & 0.964 & 0.968  &2.425 & 0.996   \\
400  & 0.908 & 2.478 & 0.980 & 0.974  &2.447 & 1.000   \\
500  & 0.885 & 2.496 & 0.976 & 0.976  &2.466 & 1.000  \\
\hline
Domain  &
      \multicolumn{3}{c|}{$ \{\psi_2(\lambda): \lambda \in \mathcal{D}_{12} \}$}                       &  \multicolumn{3}{c|}{$\{ \psi(\theta ) :\theta \in \mathcal{D}_3\}$}                                                         \\ [1pt]
    \hline
MB    &  \xmark & \multicolumn{2}{c|}{\cmark }   &    \xmark & \multicolumn{2}{c|}{\cmark } \\
\hline
$n $     &  Coverage &  $\hat q_{2,\alpha}$  &    Coverage & Coverage &  $\hat q_{3,\alpha}$  &    Coverage  \\
\hline
300 &  0.934 & 2.449 & 0.982 & 0.746 & 2.393 &  0.908 \\
400  & 0.878 & 2.468 & 0.978 & 0.780  & 2.454 & 0.928 \\
500  & 0.914 & 2.488 & 0.980 & 0.818  & 2.475 & 0.936 \\
\hline
    \end{tabular}
 }
 \end{center}
\end{table}

Next, we demonstrate the uniform validity of our confidence bands in comparison to point-wise confidence intervals.  We remove the quadratic term and the heteroscedastic noise in the simulation. Then our models are specified correctly.
We consider covering the sensitivity curves $  \Psi(\mathcal{D}),\Psi_{12}(\mathcal{D}_{12})$ and $   \Psi_3(\mathcal{D}_{3})$ in \eqref{equ:curves} for the parameter ranges $\mathcal{D} = \{2,3,\dotsc,10,11\}, \mathcal{D}_{12} = \{0.2, 0.4, \dotsc, 1.8, 2.0\}$
 and  $\mathcal{D}_3= \{0.03, 1.0, \dotsc, 0.27, 0.3\}.$
\Cref{tab:uniform_cover} shows that as the sample size increases, most of the confidence bands using MB can achieve approximately 95\% coverage for the parameters, while the point-wise CIs fail to do so.

\section{Real data study}\label{sect:real_data}

We next compare the sensitivity models on an observational study \citep{zhao2018cross} for estimating the ATE of fish consumption on the blood mercury level.
The outcome variable \quotes{blood mercury} is obtained from the individuals who answered questionnaires about seafood consumption in the National Health and Nutrition Examination Survey (NHANES) 2013-2014. The binary \quotes{treatment} variable indicates if an individual has consumed more than 12 servings of fish or shellfish in the previous month. The study has 234 treated individuals (high consumption) and 873 controls (low consumption), and 8 covariates: gender, age, income, whether income is missing, race, education, ever smoked, and the number of cigarettes smoked last month.

We use the same nuisance estimators as in the last section.  We one-step estimate all the sensitivity curves with 9-fold cross-fitting.
The introduction of our method above focuses on deriving the lower bound for $\mathbb{E}[Y(1)]$.
Fixing the sensitivity parameter at the same value, we can derive the upper bound of $\E[Y(0)]$ as follows.
We first change $Z=1$ to $Z=0$ and $Y$ to $-Y$ in the dataset. Then we apply our proposed estimator for $\E[Y(1)]$ and multiply the estimator by $-1$ at the end. In the average-case analysis, we sum up the treated and control sensitivity values, which measures the total deviation from unconfoundedness.
Connecting the ATE bounds across different values of the sensitivity parameter forms the sensitivity curve for the ATE. We construct lower confidence bands for these sensitivity curves by applying multiplier bootstrap (MB) to the efficient influence function (EIF) of the ATE lower bounds, which is given by the EIF of the treated outcome's lower bound minus the EIF of the control outcome's upper bound.

\begin{figure}[t]
    \centering
\begin{subfigure}[b]{.3\textwidth}
\centering
    \includegraphics[width=\linewidth]{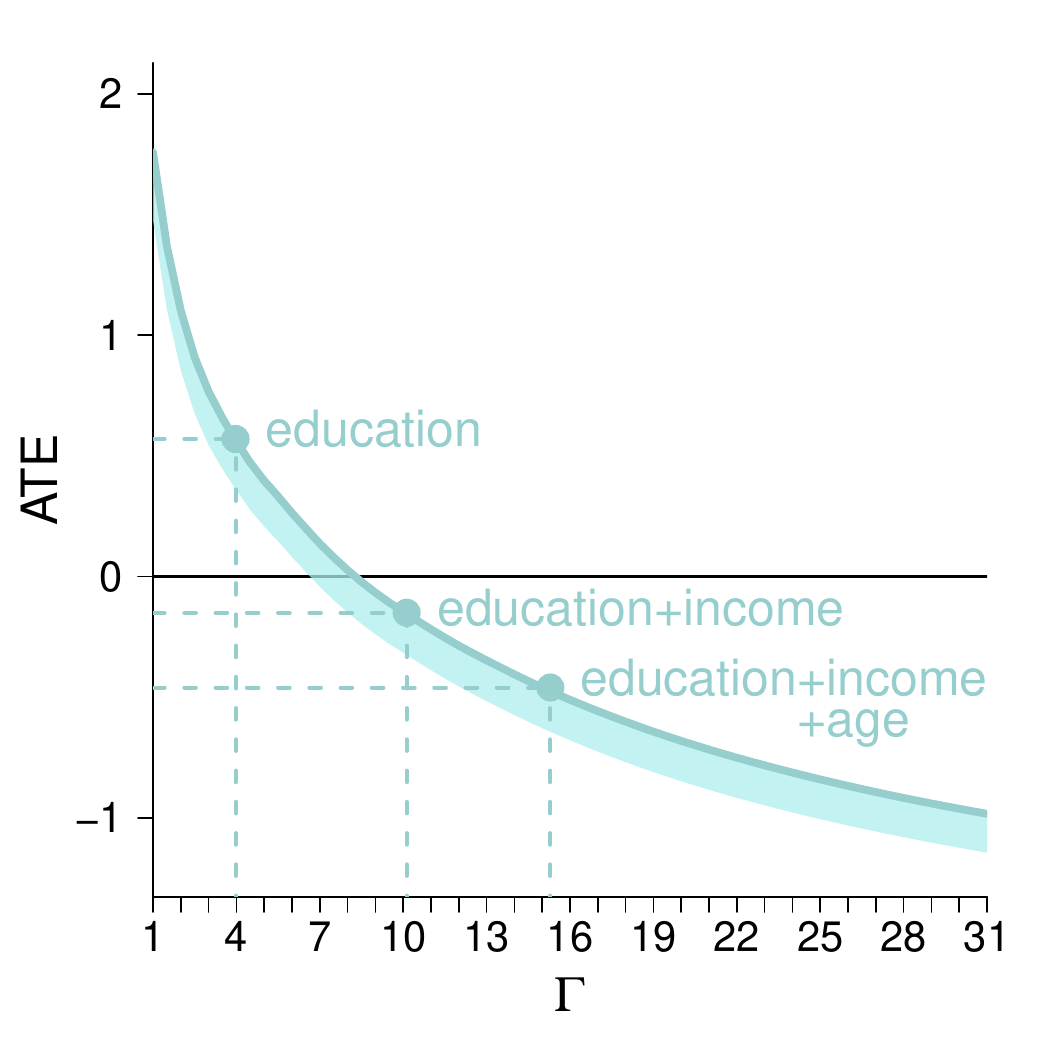}
        \caption{Worst-case model}
    \label{fig:real_curves_a}
\end{subfigure}
 \hfill
\begin{subfigure}[b]{.3\textwidth}
\centering
    \includegraphics[width=\linewidth]{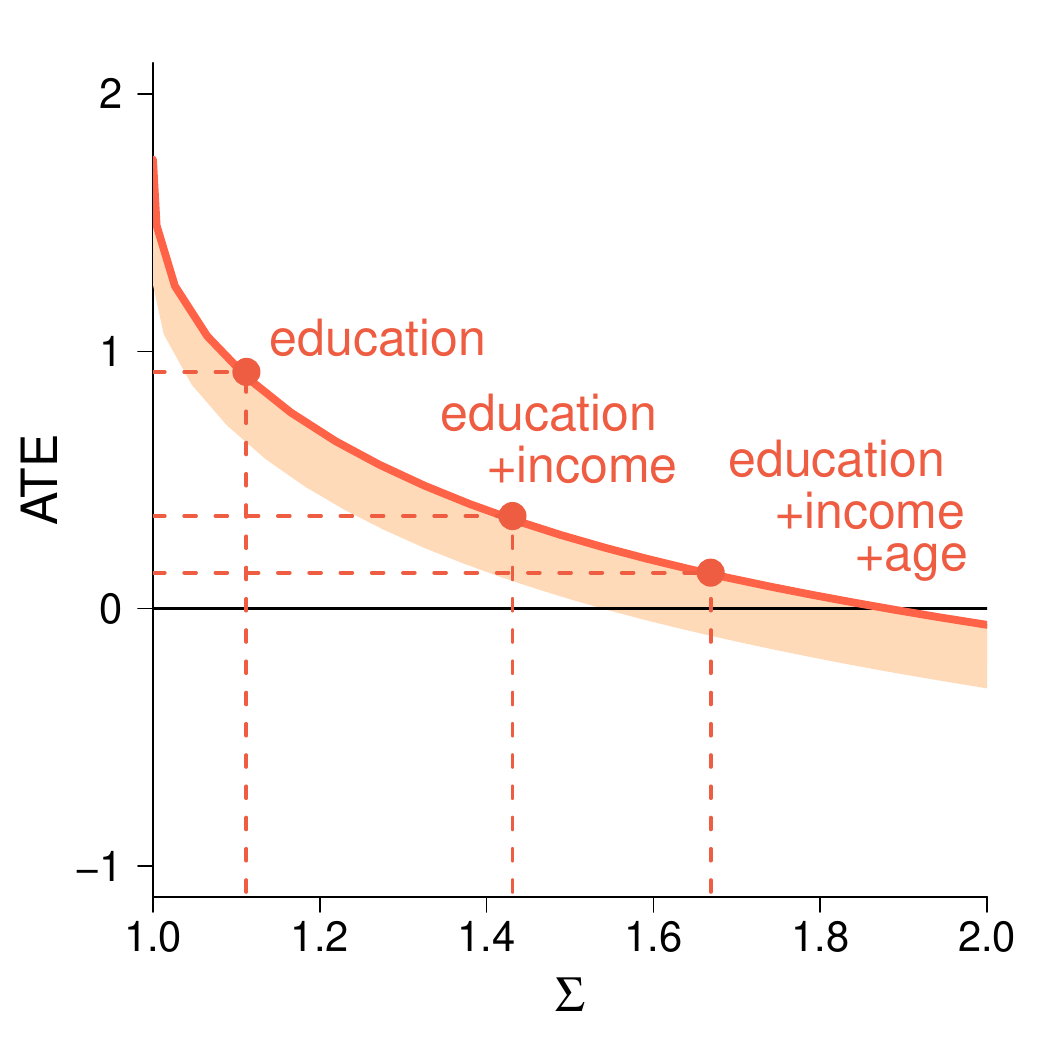}
    \caption{Average-case model}
    \label{fig:real_curves_b}
\end{subfigure}
 \hfill
\begin{subfigure}[b]{.3\textwidth}
\centering
    \includegraphics[width=\linewidth]{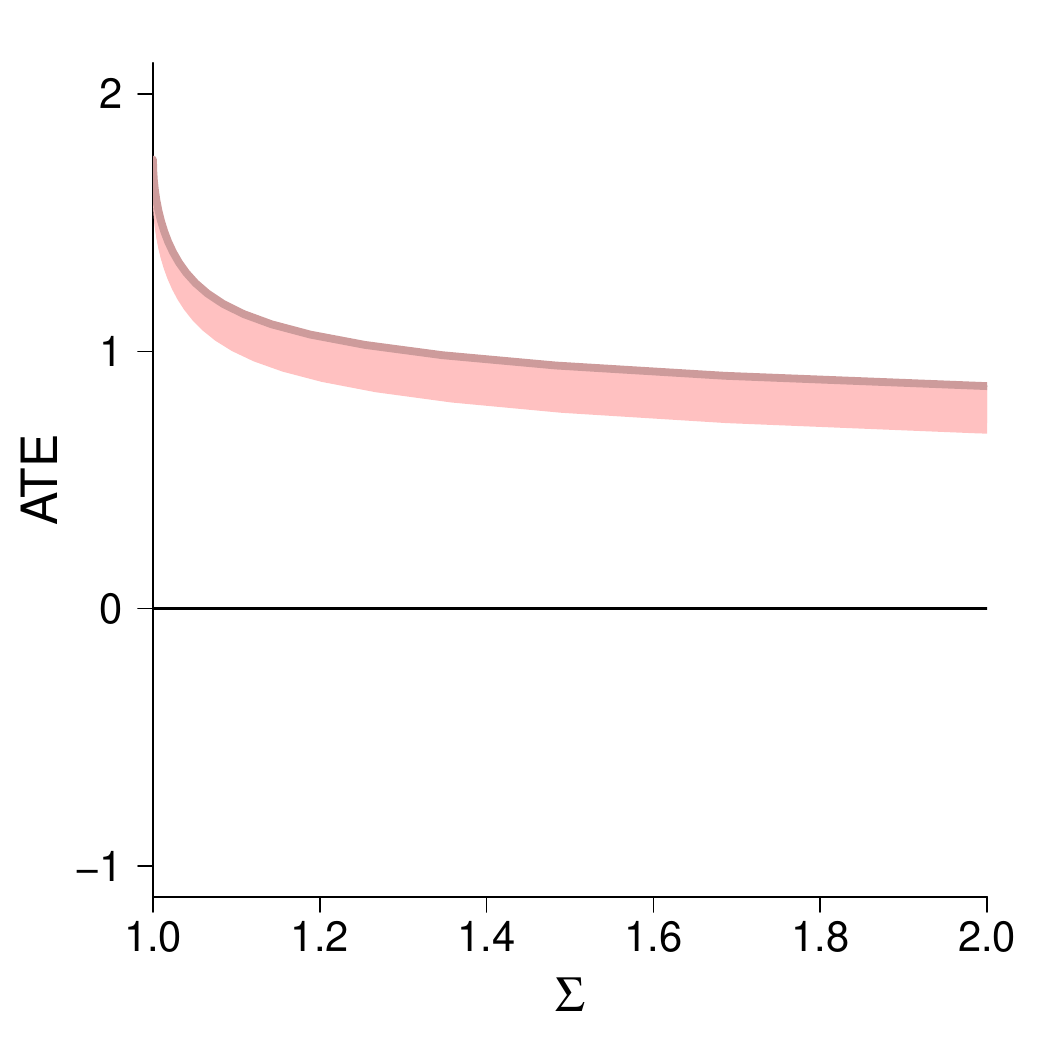}
    \caption{Average-case model (Indep. $U$)}
   \label{fig:real_curves_c}
\end{subfigure}
\caption{ Sensitivity curves of the ATE of fish consumption on blood
  mercury in the real-data study. The shaded
  regions are 90\% lower confidence bands from multiplier
  bootstrap. The average-case model is more optimistic than the worst-case model in calibration analysis based on observed covariates. It generates even tighter bounds in panel (c) by further assuming that the unmeasured confounder satisfies the
    independence in \Cref{prop:indep}. }
\label{fig:real_curves}
\end{figure}

\Cref{fig:real_curves} illustrates the sensitivity curves obtained under the worst-case sensitivity model and the average-case sensitivity model. The sensitivity value is defined as the value of the sensitivity parameter at which the corresponding sensitivity curve intersects the line $\text{ATE} = 0$.
To help interpret sensitivity values, we follow the approach proposed by \citet{imbens2003sensitivity} and \citet{hsu2013calibrating}, which suggests comparing the values with the strength of the observed covariates by pretending the covariates are unmeasured. This calibration strategy has been widely adopted in applied research.
Specifically, we assess the strength of the unmeasured confounder $U$ given the full covariates $X$, and compare it to the strength of the covariate $X_{j}$ given the other covariates $X_{-j}$.
To do so, we fit logistic regression models to estimate the propensity scores $\mathbb P\{Z_i=1\mid X_i\}$ and $\mathbb P\{Z_i=1\mid X_{i,-j}\}$ for every individual $i$, respectively.
In the worst-case sensitivity analysis, we compute the most extreme odds ratio of these propensity score estimates across all individuals in the study. We emphasize that this is only an approximation, since the worst-case model requires the constraint to hold for all possible covariate values, including those unobserved.
For the average-case analysis, we compute the empirical second moment of the ratio between these two propensity score estimates.
We note that a more formal calibration method is described by \citet{cinelli2020making}, which involves comparing the strength of $U$ given $X_{-j}$ with that of $X_j$ given $X_{-j}$. We leave the implementation of this approach as future work.

Panel (a) of \Cref{fig:real_curves} shows that the ATE remains positive if the worst-case model $\mathcal{H}_\infty(\Gamma)$ in \eqref{equ:h_infinity} holds for $\Gamma \leq 8$. However, omitting observed covariates such as \textit{education}, \textit{income} (and \textit{age}) yields an extreme odds ratio exceeding 8. This suggests that the study’s conclusion could be overturned if an unmeasured confounder were as influential as these covariates.
In contrast, panel (b) presents the sensitivity curve under our average-case sensitivity model, based on the Lagrangian formulation in \Cref{sect:second_formulation}.
When the same covariates,  \textit{education},   \textit{income} and \textit{age}, are omitted, the empirical average of the propensity score ratios remains below the sensitivity value $\Sigma = 1.9$.
This indicates that the average-case model is more optimistic than the worst-case model in sensitivity value calibration.
Panel (c) shows the sensitivity curve derived using the sensitivity value formulation in \Cref{sect:svf}, which is even more optimistic due to the additional independence assumption in \Cref{prop:indep}.
Together, the results in \Cref{fig:real_curves} suggest that the conclusion \quotes{fish consumption causally increases blood mercury levels} is more robust under the average-case sensitivity model than under the worst-case model.

\section{Discussion}\label{sect:conc}


Under the unconfoundedness assumption, the IPW estimator and its
variants have been applied in a variety of observational studies.
Recently, marginal sensitivity models has extended to deal with sequential unconfounding
\citep{bonvini2022sensitivity,tan2023sensitivity}.
It would be interesting to extend the average-case model to observational studies over time and space \citep{brumback2004sensitivity,lee2021network}.
The instability issue of IPW estimators and
the associated sensitivity analysis are often addressed by trimming
the propensity scores, which changes the estimand. It would be also
useful to consider relaxations of the optimization problems considered
here and develop simpler surrogates. Such relaxations may allow us to
include more meaningful constraints in optimization, leading to tighter bounds for partial identification.

The term \quotes{sensitivity value} was first introduced in \citet{zhao2018sensitivity} within the framework of \citet{paul2002observation}'s sensitivity model for pair-matched observational studies. The sensitivity values proposed in this article are applicable to a broader range of statistical problems. For example, they can be used in constructing confidence intervals \citep{owen2001empirical,duchi2021statistics} and evaluating stability \citep{gupta2021s}. Building on the connection,
our results on semiparametric inference may also have broader applications beyond causal inference.

\bibliographystyle{biometrika}
\bibliography{references}

\newpage

	\appendix
\appendixpage

\renewcommand{\theequation}{\arabic{equation}}
\setcounter{equation}{0}
\counterwithout{equation}{section}

Below we let $\E_{Y\mid X,Z=z}(\cdot) :=
\E(\cdot\mid X,Z=z)$ and $\P_{Y\mid X,Z=z}(\cdot) := \P(\cdot\mid
X,Z=z)$ to simplify notation.

\section{Extension to \citet{huang2022variance} }

As discussed in the Related Work section,
\citet{huang2022variance} proposed a sensitivity model that restricts the strength of unmeasured confounding through a constraint on the variance of the propensity score odds ratio. Their paper focuses on estimating the average treatment effect on the treated (ATT). Here we consider estimating the expected outcome $\mathbb E\{Y(1)\}$ for the ATE, and applying our optimization results to bounds derived under similar odds ratio--based models. In contrast, we solve the corresponding optimization problem analytically, without dropping any necessary constraints.

To begin with, we denote the propensity score odds ratio by
\[
\omega(x) =\frac{e(x)}{1-e(x)},  \qquad \omega(x,y) = \frac{e(x,y)}{1-e(x,y)},
\]
where $e(x) = \P\{Z=1\mid X=x\} $ and $e(x,y) = \P\{Z=1\mid X=x,Y(1)=y\}$.
The constraint in the sensitivity model of \citet{huang2022variance} can be defined as
\begin{equation}\label{equ:hp}
				\frac{\text{Var}(\omega(X) \mid Z=1)}{  \text{Var}(\omega(X,Y) \mid Z=1 )} \leq \rho, 								\tag{HP}
\end{equation}
where the sensitivity parameter $\rho\in [0,1]$. Observe that
\[
\omega(x,y) = \frac{1}{e^{-1}(x,y)-1} = \frac{1}{h(x,y)/e(x)-1},
\]
which is a nonlinear function of $h(x,y)$. Thus, the variance of $\omega(x,y)$ cannot be simplified using the marginalization constraint in \eqref{equ:mean_1}: $\E_{Y\mid X,Z=1}\{h(X,Y)\}=1$. Our results cannot be directly applied to solve this nonlinear optimization problem.
In response, we modify the variance ratio constraint above and define slightly different models that allow the optimization problem to be solved analytically.

\subsection{Odds ratio model I}
We first reformulate our sensitivity model via the second moment (i.e., variance) of the odds ratio instead of the probability ratio
$h(x,y) = e(x)/e(x,y)$.
We first express the odds ratio using $h(x,y)$:
\[
\frac{\omega(x)}{\omega(x,y)} = \frac{e(x)[1-e(x,y)] }{[1-e(x)]e(x,y)} = 1 + \frac{e(x) - e(x,y) }{[1-e(x)]e(x,y)}
=1+  \frac{h(x,y)-1}{1-e(x)}.
\]
Then by the marginalization constraint in \eqref{equ:mean_1}, we have
\[
\text{Var}_{Y\mid X,Z=1}(\omega(X)/\omega(X,Y)) = \frac{\E_{Y\mid X,Z=1}[h^2(X,Y)] }{[1-e(X)]^2}
\equiv \frac{\nu_{h,1}(X)}{[1-e(X)]^2}.
\]
Suppose that we solve the problem in \eqref{equ:marginal_opt} subject to $\mathbb E\{ \nu_{h,1}(X)/ [1-e(X)]^2 \} \leq \tilde \Sigma$ instead of the constraint in \eqref{equ:add_2}. As discussed in \Cref{sect:second_formulation}, optimizing the marginal and conditional Lagrangian functions of this problem would have the same solutions since all the constraints are conditional on $X$. Consider
\begin{align*}
	\text{minimize}  & \quad \frac{1}{2}\nu_{1,h}(X)/[1-e(X)]^2 +\lambda \mu_{1,h}(X)  \\
	\text{subject to}  & \quad \E_{Y\mid X,Z=1}\left \{ h(X,Y)\right \} = 1, \\
	& \quad h(X,Y)\geq e(X).
\end{align*}
The solution resembles the one in \Cref{prop:double_ipw} with $\lambda$ changed to $\lambda[1-e(X)]^2$:
\[
h_*(X,Y)  = e(X) + \lambda [1-e(X)]^2g(X,Y), \ \text{ where }\ g(X,Y) =
(\xi_X - Y ) \one_{ \{Y\leq \xi_X\} },
\]
and $\xi_X$ is the unique root of the following strictly increasing function:
\[
f_{\lambda,X}(\xi):=  \E_{Y\mid X,Z=1} \left\{ (\xi - Y ) \one_{ \{Y\leq \xi\}
} \right\} - \{1-e(X)\}^{-1}/\lambda.
\]
Then, as explained in \Cref{prop:same_curve}, solving the marginal optimization problem with
$ \tilde \Sigma = \E\{\nu_{1,h_*}(X)/[1-e(X)]\}$ would have the same optimal value $\E\{\mu_{1,h_*}(X)\}$.

\subsection{Odds ratio model II}

Observe that the variance $\tilde \nu_{h,1}(X) =\omega^2(X)\,\text{Var}_{Y\mid X,Z=1}(\omega^{-1}(X,Y))$
depends on the variance of $\omega^{-1}(X,Y)$, while the variance ratio constraint in \eqref{equ:hp} is defined using the variance of $\omega(X,Y)$. The latter leads to the nonlinear optimization problem of $h(x,y)$ mentioned above. To address this, we define a similar model using the variance of the inverse odds ratio.
By the marginalization constraint in \eqref{equ:mean_1}, 
\[
\E_{Y\mid X,Z=1}\{	\omega^{-1}(X,Y)	\} = \E_{Y\mid X,Z=1}\left\{1/e(X,Y) - 1\right\}  = 	e^{-1}(X)	-	1 = 	\omega^{-1}(X),
\]
and
\begin{align*}
	\Var_{Y\mid X,Z=1}(\omega^{-1}(X,Y) ) & = \Var_{Y\mid X,Z=1}\left(1/e(X,Y) \right)	\\
	&\ =
	\E_{Y\mid X,Z=1}\{e^{-2}(X,Y) \} - \E_{Y\mid X,Z=1}^2\{e^{-1}(X,Y) \} \\
	&\ = e^{-2}(X)\E_{Y\mid X,Z=1}\{h^{2}(X,Y) \}  - e^{-2}(X)	\\
	&\ = e^{-2}(X) \left[   \nu_{1,h}(X)-1         \right].
\end{align*}
Then by the law of total variance,
\begin{align*}
	\Var(\omega^{-1}(X,Y) \mid Z=1 ) & = \E \{	\Var_{Y\mid X,Z=1}(\omega^{-1}(X,Y) ) 	\mid Z=1\} + \Var(	\omega^{-1}(X)\mid Z=1)	\\
	& = \P^{-1}\{Z=1\} \E \{e^{-1}(X)[\nu_{h,1}(X)-1]\}+ \Var(	\omega^{-1}(X)\mid Z=1).
\end{align*}
We change the variance ratio constraint in \eqref{equ:hp} to
\begin{equation}\label{equ:A2}
	\frac{	\Var(\omega^{-1}(X,Y) \mid Z=1 ) 			}{	\Var(\omega^{-1}(X) \mid Z=1 ) 		} \leq \rho. \tag{A2}
\end{equation}
We note that the inverse odds ratio carries the same interpretation as the odds ratio, but in the opposite direction.
This constraint is equivalent to
\begin{align*}
	& \E \{e^{-1}(X)[\nu_{h,1}(X)-1]\}\leq (\rho-1)\,\P\{Z=1\}\,\Var(\omega^{-1}(X) \mid Z=1 )	\\
	\Leftrightarrow\ \ &
	\E \{e^{-1}(X)\nu_{h,1}(X)\}\leq (\rho-1)\,\P\{Z=1\}\,\Var(\omega^{-1}(X) \mid Z=1 ) - \E \{e^{-1}(X)\} \ := b(\rho).
\end{align*}
The right-hand side of this constraint (i.e., $b(\rho)$) can be estimated from the observed data. As discussed in the last subsection,
we can solve the problem in \eqref{equ:marginal_opt} subject to the constraint in \eqref{equ:A2}  via
\begin{align*}
	\text{minimize}  & \quad \frac{1}{2} e^{-1}(X) \nu_{1,h}(X)+\lambda \mu_{1,h}(X)  \\
	\text{subject to}  & \quad \E_{Y\mid X,Z=1}\left \{ h(X,Y)\right \} = 1, \\
	& \quad h(X,Y)\geq e(X).
\end{align*}
The solution resembles the one in \Cref{prop:double_ipw} with $\lambda$ changed to $\lambda e(X)$:
\[
h_*(X,Y)  = e(X) + \lambda e(X)g(X,Y), \ \text{ where }\ g(X,Y) =
(\xi_X - Y ) \one_{ \{Y\leq \xi_X\} },
\]
and $\xi_X$ is the unique root of the following strictly increasing function:
\[
f_{\lambda,X}(\xi):=  \E_{Y\mid X,Z=1} \left\{ (\xi - Y ) \one_{ \{Y\leq \xi\}
}  \right\} -  [e^{-1}(X)-1]/\lambda.
\]
Let $b(\rho^*) = \E\{e^{-1}(X)\nu_{1,h_*}(X)\}$. By the definition $b(\rho^*)$ above, this equality is achieved when
\[
\rho^* = 1+ \frac{ \E\{e^{-1}(X)\nu_{1,h_*}(X)\} +  \E \{e^{-1}(X)\} }{ \P\{Z=1\}\,\Var(\omega^{-1}(X) \mid Z=1 ) }.
\]
As explained in \Cref{prop:same_curve}, solving the marginal optimization problem subject to the constraint in \eqref{equ:A2} with $\rho = \rho^*$
would have the same optimal value $\E\{\mu_{1,h_*}(X)\}$.


\section{Proof of Proposition \ref{prop:double_ipw}}\label{appendix:double_ipw}
\begin{proof}
 The objective function of \eqref{equ:double_ipw} is marginal over $X$,
while the constraints are conditional on $X$. It makes no difference to solve this problem conditionally. The conditional Lagrangian is given by
\begin{align*}
\mathcal{L} =   &\ \mathbb E_{Y\mid X,Z=1}\big\{ h^2(X,Y)/2 + \lambda  h(X,Y) Y \big\}  + \lambda_{X,2}\left[ 1 - \mathbb E_{Y\mid X,Z=1}\left\{ h(X,Y) \right\} \right]  \\
&\ - \mathbb E_{Y\mid X,Z=1}\left\{ \lambda_{Y,3} [h(X,Y)-e(X)] \right\}.
\end{align*}
Setting the functional derivative of $\mathcal{L}$ w.r.t. $h$ to $0$, we obtain the Euler-Lagrangian equation, i.e., the stationarity condition in the KKT conditions,
\[
h(X,Y) + \lambda Y - \lambda_{X,2}  - \lambda_{Y,3}  = 0.
\]
By complementary slackness, $\lambda_{Y,3} [ h(X,Y)-e(X)]  =0$.
By dual feasibility, $\lambda_{Y,3}  \geq0$. If $\lambda_{Y,3} =0$,
\[
h(X,Y) = \lambda_{X,2} - \lambda Y = e(X) +  \lambda(\xi_X - Y),
\]
where $\xi_X:=[\lambda_{X,2}  - e(X)]/ \lambda$. By primal feasibility and $\lambda>0$, $h(X,Y)\geq e(X) \Leftrightarrow Y\leq \xi_X.$
If $\lambda_{Y,3}>0$, $h(X,Y)= e(X)$.  By the stationarity condition and $\lambda>0$,
that is when $Y>\xi_X.$ Now we have proven the solution in
\eqref{equ:solution_double_1_1}. Finally, $f_{\lambda,X}(\xi)$ in  \eqref{equ:unique} is attained  by substituting \eqref{equ:solution_double_1_1} into the first constraint in \eqref{equ:double_ipw}. Under Assumption \ref{assumption:first}, $\E_{Y\mid X,Z=1} \big\{( \xi -Y) \one_{ \{Y\leq \xi   \} }\big\}$ is a strictly increasing and positive function of $\xi$. Then, $ f_{\lambda,X}(\xi)$  has a unique root $\xi_X$ for any $\lambda>0$.
\end{proof}

\section{Proof of Proposition \ref{prop:same_curve}}\label{appendix:prop:same_curve}

\begin{proof}
Denote the objective function of  \eqref{equ:double_ipw} by
\[
D(h) := \E \left[ \nu_{1,h}(X)  \right]/2+\lambda \E \left[  \mu_{1,h}(X) \right].
\] 
For the solution $h_*(X,Y)$ defined in \eqref{equ:solution_double_1_1}, we obtain the optimal value
$D(h_*) = \psi_{1}(\lambda)/2 + \lambda \psi_{2}(\lambda).$ It is easy to see that $h_*$ is a feasible point of \eqref{equ:marginal_opt}. 
If there is another solution $h_{**}$ of \eqref{equ:marginal_opt}, it must satisfy
\[
  \E\left\{\nu_{1,h_{**}}(X) \right\}\leq \psi_{1}(\lambda) \  \text{ and } \  \E \left\{  \nu_{1,h_{**}}(X) \right\}\leq \psi_2(\lambda ),
\]
which implies that $D(h_{**})\leq  D(h_*).$ Since $h_{**}$ is also a feasible point of \eqref{equ:double_ipw}, it holds that $D(h_{**})\geq D(h_{*}).$ Given that $D(h_{**})= D(h_{*})$, equivalently,
\[
\E \left \{\nu_{1,h_{**}}(X) \right\}/2 +\lambda \E \left\{ \mu_{1,h_{**}}(X)  \right\} = \psi_{1}(\lambda)/2 +\lambda \psi_2(\lambda ).
\]
Thus, the penultimate equation can only hold with equalities.
\end{proof}

\section{Proof of Proposition \ref{prop:solutions_new_1}}\label{appendix:solutions_new_1}

For $\theta>0,$  the solution of the program \eqref{equ:new_ipw} with an additional constrain in \eqref{equ:lambda_h}
is given by
\begin{equation}\label{equ:h_star}
\begin{split}
 h_*(X,Y)  = \begin{cases}
 W_{+}(X),        \quad \quad \quad \text{if } Y<  \xi_X - W_{+}(X)/\lambda_{X},\\
\lambda_{X} (\xi_X -  Y),     \ \   \text{ if }\xi_X - W_{+}(X)/\lambda_{X}  \leq  Y \leq  \xi_X - W_{-}(X)/\lambda_{X},    \\
  W_{-}(X), \quad       \quad \quad     \text{if } Y> \xi_X - W_{-}(X)/\lambda_{X},
 \end{cases}
\end{split}
\end{equation}
where $\xi_X = \lambda_{X,2}/\lambda_X$ is defined in \eqref{equ:mid_step} below. When $\Gamma=\infty,$
$W_-(X) = e(X)$ and $W_+(X) = \infty.$ By re-defining $\xi_X$ as $\xi_X =  \lambda_{X,2}/\lambda_X - e(X)/\lambda_X$, \eqref{equ:h_star} reduces to the solution \eqref{equ:solution_new_1_1} in the main manuscript.

\begin{proof} The Lagrangian function of \eqref{equ:new_ipw} further subject to $h\in \mathcal{H}_{\infty}(\Gamma)$ is given by
\begin{align*}
\mathcal{L} =   &\ \frac{1}{2} \mathbb E_{Y\mid X,Z=1}\big[ h^2(X,Y)\big] + \lambda_{X}\left[   \mathbb E_{Y\mid X,Z=1}\left\{ h(X,Y) Y\right\} - \mathbb E_{Y\mid X,Z=1}\left\{ Y\right\}  + \theta  \right]   \\
& + \mathbb E_{Y\mid X,Z=1}\big\{ \lambda_{Y,3}\left[W_-(X) - h(X,Y) \right] + \lambda_{Y,4}\left[h(X,Y) - W_+(X) \right] \big\}\\
& + \lambda_{X,2}\left[ 1 - \mathbb E_{Y\mid X,Z=1}\left\{ h(X,Y) \right\} \right].
\end{align*}
Setting the functional derivative of $\mathcal{L}$ w.r.t. $h$ to $0$, we obtain the Euler-Lagrangian equation (i.e. the stationarity condition in the KKT conditions),
\[
h(X,Y) +\lambda_{X} Y - \lambda_{X,2}  - \lambda_{Y,3} +\lambda_{Y,4} = 0.
\]
By complementary slackness, we have $ \lambda_{X}\left[ \mathbb E_{Y\mid X,Z=1}\left\{ h(X,Y) Y\right\} -\mathbb E_{Y\mid X,Z=1}\left\{ Y\right\} + \theta  \right]=0,$
\[
\lambda_{Y,3}\left[W_-(X) - h(X,Y) \right] =0 \quad \text{and} \quad  \lambda_{Y,4}\left[h(X,Y) - W_+(X) \right]=0.
\]
By dual feasibility, $\lambda_{X},\lambda_{Y,3}, \lambda_{Y,4}\geq 0$.
Since $ W_+(X)> W_-(X)$, it is impossible that $\lambda_{Y,3},\lambda_{Y,4}>0$.
When $\lambda_{Y,3}> 0 $ and $\lambda_{Y,4}= 0$, $h(X,Y) = W_-(X).$  Further if $\lambda_{X}=0$, by the stationarity condition,
\[
\lambda_{X,2} =  W_-(X)  - \lambda_{Y,3} +  \lambda_{Y,4} < W_-(X).
\]
When $\lambda_{Y,4}> 0$ and $\lambda_{Y,3}= 0$, $h(X,Y) = W_+(X).$ Further if $\lambda_{X}=0$,
\[
\lambda_{X,2} =  W_+(X)  - \lambda_{Y,3} +  \lambda_{Y,4} > W_+(X).
\]
Because $W_+(X)> W_-(X)$, $\lambda_{X,2} < W_-(X)$ contradicts with  $\lambda_{X,2} > W_+(X)$ in the last two equations, so
we cannot have $h(X,Y)= W_-(X)$ and $W_+(X)$ for two different values of $Y$ if $\lambda_{X}=0$. If
$\lambda_{Y,3}$ or $\lambda_{Y,4}$ is always positive,  i.e.,
$h(X,Y)$ is always equal to $W_-(X)$ or $W_+(X)$, $h(X,Y)$  does not satisfy the equality constraint $\mathbb E_{Y\mid X,Z=1}\left[ h(X,Y) \right]=1$. So we know that $\lambda_{X}\neq 0$ unless $\lambda_{Y,3} = \lambda_{Y,4}= 0.$

 If  $\lambda_{X} = \lambda_{Y,3} = \lambda_{Y,4}= 0$, $h(X,Y) = \lambda_{X,2}$ by the stationarity condition. By primal feasibility, $\mathbb E_{Y\mid X,Z=1}\left[ h(X,Y) \right] = \lambda_{X,2} =1\Rightarrow h(X,Y)=1,$ then
 \[
\mathbb E_{Y\mid X,Z=1}\left[ Y \right]  \leq  \mathbb E_{Y\mid X,Z=1}\left[ Y \right]  - \theta \Rightarrow \theta \leq 0,
 \]
and $W_-(X) \leq   1 \leq  W_+(X).$ This completes the proof for $h_*(X,Y)=1 $ if $\theta \leq 0.$

We now consider the case that $\lambda_{X}>0$. When $\lambda_{Y,3} = \lambda_{Y,4}= 0$, the stationarity condition implies that
\begin{equation}\label{equ:mid_step}
h(X,Y) = \lambda_{X,2} - \lambda_{X} Y  = \lambda_{X} (\xi_X-Y) \text{ with } \xi_X :=\lambda_{X,2}/\lambda_{X},
\end{equation}
By primal feasibility, we have
\[
W_-(X)  \leq   \lambda_{X} (\xi_X-Y)  \leq  W_+(X)  \ \Leftrightarrow   \       \xi_X - W_{+}(X)/\lambda_{X}  \leq  Y \leq  \xi_X - W_{-}(X)/\lambda_{X},
\]
If $\lambda_{X},\lambda_{Y,3}> 0 $ and $\lambda_{Y,4}= 0$, $h(X,Y) = W_-(X) $. It follows from the stationarity condition that
\[
 W_-(X) + \lambda_{X} Y - \lambda_{X,2}   >  0 \ \Leftrightarrow \ Y > [\lambda_{X,2}- W_-(X)]/\lambda_{X} = \xi_X - W_-(X)/\lambda_{X},
\]
If $\lambda_{X},\lambda_{Y,4}> 0 $ and $\lambda_{Y,3}= 0$, $h(X,Y) = W_+(X) $, then
\[
W_+(X) + \lambda_{X} Y - \lambda_{X,2}  <  0 \ \Leftrightarrow \ Y <  [\lambda_{X,2} - W_+(X) ]/\lambda_{X} = \xi_X - W_+(X)/\lambda_{X}.
\]
The last three equations complete the proof for \eqref{equ:h_star}.

We now verify the uniqueness of $\xi_X$ in the solution \eqref{equ:solution_new_1_1} when $\theta >0.$ In this case, we have proved that $\lambda_{X}> 0.$ Then by complementary slackness, $\E_{Y\mid X,Z=1}[h_*(X,Y)Y] = \E_{Y\mid X,Z=1}[Y] -\theta$.
Dividing it by $\lambda_X\E_{Y\mid X,Z=1}[g(X,Y)]=1-e(X)$ removes $\lambda_{X}$, then we can find $\xi_X$ by solving the equation
\[
\E_{Y\mid X,Z=1}\left[(\xi_X-Y)Y\one_{\{Y\leq \xi_X\}}\right]/\E_{Y\mid X,Z=1}\left[(\xi_X-Y)\one_{\{Y\leq \xi_X\}}\right] = \E_{Y\mid X,Z=1}[Y] - \theta/[1-e(X)].
\]
This leads to the definition of $ f_{\theta,X} (\xi ) $ in
\eqref{equ:ratio}. The function $ f_{\theta,X} ( \xi  ) $ has a positive derivative $ d f_{\theta,X} (\xi) /d \xi  $:
\begin{align*}
        &\  \frac{ \E_{Y \mid X,Z=1}\big[Y \one_{ \{Y\leq \xi\}} \big] \times \E_{Y\mid X,Z=1,Y\leq \xi}\big[ \xi-Y \big]- \PP_{X,Z=1}\{Y\leq  \xi    \}\times \E_{Y\mid X,Z=1,Y\leq \xi }\big[ (\xi-Y )Y    \big]}{\E_{Y\mid X,Z=1,Y\leq \xi }^2\big[\xi-Y \big]    \PP_{X,Z=1}\{Y\leq \xi  \}   }\\
        =&\ \frac{   \xi \E_{Y\mid X,Z=1,Y\leq \xi}\big[Y \big] -  \E_{Y\mid X,Z=1,Y\leq \xi}^2 \big[Y \big]         - \xi \E_{Y\mid X,Z=1,Y\leq \xi}\big[Y \big] +  \E_{Y\mid X,Z=1,Y\leq \xi}\big[Y^2 \big]     }{\E_{Y\mid X,Z=1,Y\leq \xi }^2\big[\xi-Y \big]}        \\
        =& \Var_{Y\mid X,Z=1,Y\leq \xi } \big[Y\big]/\mathbb E_{Y\mid X,Z=1,Y\leq \xi }^2 \big[ \xi -Y\big] >0.
\end{align*}
The function $f_{\theta,X} (\xi)$ can be written as
\[
f_{\theta,X} (\xi)      = \frac{\E_{Y\mid X,Z=1}\left[          (\xi    -Y )(Y - \E_{Y\mid X,Z=1}[Y]+ \theta/[1-e(X)]     )\one_{\{Y\leq \xi                              \}}     \right]}{\E_{Y\mid X,Z=1}\big[           (\xi    -Y )\one_{\{Y\leq \xi                           \}}\big]}  .
\]
If $\xi< \E_{Y\mid X,Z=1}[Y] - \theta/[1-e(X)]$,
\begin{align*}
&\ (\xi   -Y )(Y- \E_{Y\mid X,Z=1}[Y] + \theta/[1-e(X)]     )\one_{\{Y\leq \xi                              \}} \\
=  &\ \big\{-( \xi  -Y )^2 +  (\xi  -Y )(\xi - \E_{Y\mid X,Z=1}[Y] +\theta/[1-e(X)]   ) \big\}\one_{\{Y\leq \xi                               \}} <  0.
\end{align*}
Since $\E_{Y\mid X,Z=1}\big[              (\xi    -Y )\one_{\{Y\leq  \xi                          \}}\big]> 0, f_{\theta,X} (\xi )< 0$ if $\xi < \E_{Y\mid X,Z=1}[Y] - \theta/[1-e(X)]$. Together with the positive derivative above, we know that $f_{\theta,X} \big(\xi          \big)$ has a unique root $\xi_X$. Finally, the expression of $\lambda_{X}$ is derived from the equality constraint $\E_{Y\mid X,Z=1}[h_*(X,Y)]=1.$\end{proof}

\section{Proof of Proposition \ref{prop:quantile_balancing}}\label{appendix:quantile_balancing}
\begin{proof}
The Lagrangian function of \eqref{equ:msm_opt} is
\begin{align*}
\mathcal{L}=   &\ \mathbb E_{Y\mid X,Z=1}\big\{ -h(X,Y) Y + \lambda_{Y,1}\left[W_-(X) - h(X,Y) \right] + \lambda_{Y,2}\left[h(X,Y) - W_+(X) \right] \big\}     \\
&\ + \lambda_{X,3}\left\{ 1 - \mathbb E_{Y\mid X,Z=1}\left[  h(X,Y) \right]  \right\}.
\end{align*}
Setting the functional derivative of $\mathcal{L}$ w.r.t. $h$ to $0$, we obtain the Euler-Lagrangian equation, i.e., the stationarity condition in the KKT conditions,
\[
-Y -\lambda_{Y,1} + \lambda_{Y,2} - \lambda_{X,3} = 0 \quad \Leftrightarrow \quad  \lambda_{Y,1} - \lambda_{Y,2}=  -\lambda_{X,3} -Y.
\]
By complementary slackness,
\[
\lambda_{Y,1}\left[W_{-}(X)-h(X,Y)\right] = 0 \quad \text{and}\quad  \lambda_{Y,2}\left[ h(X,Y) - W_{+}(X)\right] = 0.
\]
Combined this with the dual feasibility ($\lambda_{Y,1},\lambda_{Y,2}\geq 0$), we have
\begin{equation*}
[\lambda_{Y,1},\lambda_{Y,2},h(X,Y) ]=
\begin{cases}
[-\lambda_{X,3}-Y, 0, W_{-}(X)], \text{ if }  Y< -\lambda_{X,3}, \\
[0, \lambda_{X,3}+Y, W_{+}(X)],\hspace{13pt} \text{if }  Y> -\lambda_{X,3}.\\
        \end{cases}
\end{equation*}
Let $\alpha_*  :=\PP \{Y<-\lambda_{X,3}   \mid X=x,Z=1\} $.
By the primal feasibility,
\begin{align*}
        \mathbb E_{Y\mid X,Z=1} \left\{ h(X,Y)\right\}  = \alpha_* W_{-}(X)+ (1-\alpha_*) W_{+}(X) = 1,
\end{align*}
which implies that
\[
\alpha_* = \frac{ 1- W_+(X)}{W_-(X)- W_+(X)}  = \frac{1- (1-\Gamma )e(X) - \Gamma }{ (1-\Gamma^{-1} )e(X) + \Gamma^{-1} -  (1-\Gamma )e(X) - \Gamma } = \frac{\Gamma}{1+\Gamma},
\]
then $-\lambda_{X,3}$ is the $\Gamma/(1+\Gamma)$-quantile $Q(X)$. The solution of \eqref{equ:msm_opt} with minimization can be derived in the same way after changing $- h(X, Y) Y$ to $h(X, Y) Y$ in the definition of $\mathcal{L}$.
\end{proof}

\section{Proof of Proposition \ref{prop:sharp} and a result for the sensitivity value formulation }\label{appendix:sharp_ate}

\subsection{Proof of Proposition \ref{prop:sharp} }

\begin{proof}
We first set up the notation of our proof. The lower and upper bounds of $\E[Y(1)]$ and  $\E[Y(0)]$
are  derived using the following four optimizers:
\begin{align*}
        & h_1^-(x,y) = e_1(x) + \lambda (\xi_{x,1}^{-}-y)\one_{\{y\leq  \xi_{x,1}^{-} \}}\equiv e_1(x)  +\lambda g_1^-(x,y), \\
        & h_1^+(x,y) = e_1(x) + \lambda (y-\xi_{x,1}^{+})\one_{\{y\geq \xi_{x,1}^{+} \}}\equiv e_1(x) +\lambda g_1^+(x,y),      \\
        & h_0^-(x,y) = e_0(x) + \lambda  (\xi_{x,0}^{-}-y)\one_{\{y\leq  \xi_{x,0}^{-} \}}\equiv e_0(x)  +\lambda g_0^-(x,y), \\
        & h_0^+(x,y) = e_0(x)+ \lambda (y-\xi_{x,0}^{+})\one_{\{y\geq \xi_{x,0}^{+}\}}\equiv e_0(x) +\lambda g_0^+(x,y).
\end{align*}
where the superscript \quotes{-} denote lower bound, \quotes{+} denote upper bound,
$e_1(x) \equiv e(x)$ and $e_0(x) \equiv 1-e(x).$
The function $h_1^-(x,y)$ is exactly the solution $h_*(x,y)$ in \Cref{prop:double_ipw}. It is the  lower bound of $\E[Y(1)\mid X].$
To upper bound  $\E[Y(1)\mid X],$ we derive  $h_1^+(x,y)$ by solving the program \eqref{equ:double_ipw} with $Y$ changed to $-Y.$ We find the root of the function
\[
\E_{Y\mid  X=x,Z=1}[(Y-[-\xi])\one_{\{Y\geq  -\xi\}}] - [1-e(x)]/ \lambda.
\]
If we define $\xi_{x,1}^{+}$ as the root multiplied by -1, we have the definition of $h_{1}^+(x,y)$ above. Similarly,  $h_0^-(x,y)$ and $h_0^+(x,y)$ are the minimizer and maximizer of $\E[Y(0)\mid X=x]$ respectively.

The identification regions of $\E[Y(1)]$  and $\E[Y(0)]$ are denoted by
$\mathcal{B}_1 = [b_1^-,b_1^+]$ and $\mathcal{B}_0 = [b_0^-,b_0^+]$,
where $b_1^- = \E_{X}[\mu_{1,h_{1}^-}(X)]$ and other $b$'s are defined similarly using the notation above.

To achieve the third condition in the proposition,
we need to find a distribution $\tilde{\PP}$ satisfying
\[
\tilde \E[Y(z)] = b_z, \forall  b_z\in  \mathcal{B}_z.
\]
We first find a distribution to achieve the two extreme points $b_z^-$ and $b_z^+$ of $\mathcal{B}_z$, respectively. Then we define a mixture $\tilde \PP_{z}$
of the two extremal distributions to attain any point $b_z \in  \mathcal{B}_z$. In what follows,
we use the notation \quotes{$\tilde{\hspace{6pt}  }$} to indicate all the quantities defined under $\tilde \PP_{z}$.

First, we define the potential outcomes under $\tilde \PP_{z}$ to satisfy the standard consistency assumption:
\begin{equation}\label{equ:chain}
\tilde  p_{Y(z)\mid X}(y\mid x) = \tilde  p_{Y(z)\mid X,Z}(y\mid x,z) = \tilde  p_{Y\mid X,Z}(y\mid x,z)  := p_{Y\mid X,Z}(y\mid x,z).
\end{equation}
The last equality is required in the first condition. Similarly, we keep the other distributions the same:
\[
\tilde  p_X(x) := p_X(x) \  \text{ and }\ \tilde  e_z(x)  := e_z(x).
\]
We now turn to the second condition. For $t\in \{+,-\}$, we define the counterfactual outcome distribution
\[
 \tilde  p_{Y(z)\mid X,Z}( y\mid x,1-z) := \frac{ g_z^{t}(x,y)p_{Y\mid X,Z}(y\mid x,z)}{\int g_z^{t}(x,y')p_{Y\mid X,Z}(y'\mid x,z) dy'}
 = \frac{\lambda g_z^{t}(x,y)}{e_{1-z}(x)}p_{Y\mid X,Z}(y\mid x,z),
\]
by the definition of $g_z^{t}(x,y)$ and $\xi_{x,z}^{t}$.
above. The last equation and \eqref{equ:chain} imply that
\[
\frac{p_{Y(z)\mid X,Z}( y\mid x,1-z)}{p_{Y(z)\mid X,Z}( y\mid x,z)
} = \frac{\lambda g_z^{t}(x,y)}{e_{1-z}(x)}.
\]
Using Bayes' rule, we can rewrite $ \tilde  e_z(x,y)$ as
\begin{align*}
 \tilde  e_z^{t}(x,y)  = \mathbb{\tilde  P}_{Z\mid X,Y(z)}(z\mid x,y)
= \frac{e_z(x)  }{              e_z(x)  + \lambda  g_z^{t}(x,y)}      = \frac{e_z(x)  }{      h_z^{t}(x,y)},
\end{align*}
Given $\tilde  e_z(x)  = e_z(x),$ this implies that
$   \tilde  h_z(x,y)  := \tilde  e_z(x) / \tilde  e_z(x,y) = h_z^{t}(x,y).$
As an optimizer,  $h_z^{t}(x,y)$ satisfies the constraint of \eqref{equ:double_ipw}, so does $\tilde  h_z(x,y)$, which is the requirement in the second condition.
To verify the third condition, we first note that
\[
b_z^{t} = \E_{X,Y\mid Z=z}[h_z^{t}(X,Y)Y] =   \tilde \E_{X,Y\mid Z=z}\left[ \frac{\tilde e_z(X)Y}{\tilde e_z^t(X,Y) } \right] = \tilde \E_z^t[Y(z)].
\]
The second equality is achieved by the first condition and the definition of $ \tilde  h_z(x,y)$.
The third equality is attained by the IPW formula in \eqref{equ:first_formula}.
Given any point $b_z\in  \mathcal{B}_z$, we can write it as a convex combination,
\[
b_z = \tilde w_z^{+}b_z^{+} + \tilde w_z^{-}b_z^{-} = \tilde w_z^{+}\tilde \E_z^+[Y(z)] + \tilde w_z^{-}\tilde \E_z^-[Y(z)]= \tilde \E_z[Y(z)],
\]
by  defining the mixture $\tilde  \PP_z =  w_z^{+}\tilde  \PP_z^+ + w_z^{-}\tilde  \PP_z^-$. By construction, the three conditions hold for $\tilde  \PP_z$ because they hold for the two extremal distributions $\PP_z^+ $ and $\PP_z^-$.
\end{proof}

\subsection{Proof for the sensitivity value formulation}\label{appendix:sharpcate}

\begin{proposition}\label{prop:sharp_cate}

Under Assumption \ref{assumption:first},
for $z\in \{0,1\}$ and any distribution $\PP$ with $h(X,Y(z))$ satisfying the constraints in \eqref{equ:new_ipw}, there exists a distribution $\tilde \PP$ of $(X,Z,Y,Y(z))$ satisfying that
        \begin{enumerate}
                \item its marginal distribution of
               $(X,Y,Z)$ matches the observed
                   data distribution under $\PP $.
                \item its propensity score is the solution of \eqref{equ:new_ipw}.
        \end{enumerate}
\end{proposition}

\begin{proof}
        The proof is similar to the one for \Cref{prop:sharp} above.
        In the program \eqref{equ:new_ipw} for $Z=0$ and $1$, we minimize the second moments to obtain two sensitivity values. We denote the two optimizers,
\begin{align*}
        & h_1^*(X,Y) = e(X) + \lambda_{X,1} (\xi_{X,1}-Y)\one_{\{Y\leq  \xi_{X,1} \}},   \\
        & h_{0}^*(X,Y) = 1-  e(X) + \lambda_{X,0}  (Y-\xi_{X,1})\one_{\{Y\geq \xi_{X,0}\}},
\end{align*}
where $h_1^*(X,Y)$ is exactly $h_{*}(X, Y)$ in \eqref{equ:solution_new_1_1}.
We denote the two optimizers above as $h_{z}^*(X,Y) = e_z(X) + \lambda_{X,z} g_z(X,Y)$ for $z=0,1.$
The sensitivity values are the lower bound of the second moment. We define the identification regions as
$\mathcal{B}_1 = [b_{1}^*,\infty]$ and $ \mathcal{B}_0 = [b_{0}^*,\infty],$ where
\[
b_z^* = \E_{Y\mid X,Z=z}\{ [h_z^*(X,Y)]^2 \}.
\]
We construct an extremal distribution $\tilde \PP_z$  that can attain the extreme point $b_{z}^*$ of the region $\mathcal{B}_z$.
First, we define the potential outcomes under $\tilde \PP_{z}$ to satisfy the standard consistency assumption:
\[
\tilde  p_{Y(z)\mid X}(y\mid x) = \tilde  p_{Y(z)\mid X,Z}(y\mid x,z) = \tilde  p_{Y\mid X,Z}(y\mid x,z)  := p_{Y\mid X,Z}(y\mid x,z).
\]
Similarly, we let  $\tilde  p_X(x) := p_X(x)$ and $\tilde  e_z(x)  := e_z(x)$.
Define the counterfactual outcome distribution
\[
 \tilde  p_{Y(z)\mid X,Z}( y\mid z, 1-z) := \frac{g_z(x,y)p_{Y\mid X,Z}(y\mid x,z)}{\E_{Y\mid X,Z=z}[g_z(X,Y)] }= \frac{\lambda_{x,z}  g_z(x,y)}{e_{1-z}(x)}p_{Y\mid X,Z}(y\mid x,z),
\]
by the definition of $\lambda_{x,z}$ in \Cref{prop:solutions_new_1}. Then,
\begin{align*}
 \tilde  e_z(x,y)     =  \frac{e_z(x)  }{              e_z(x)  + \lambda_x  g_z(x,y)}
 = \frac{e_z(X)  }{     h_z^*(x,y)},
\end{align*}
Given $\tilde  e_z(x)  = e_z(x),$ this implies that $\tilde h_z(x,y) := \tilde e_z(x) /\tilde  e_z(x,y)  =  h_z^*(x,y).$ Then,
\begin{align*}
\mathbb{\tilde E}_{Y \mid X,Z=z}\big[\tilde h_z(X,Y)Y\big] & = \mathbb  E_{Y \mid X,Z=z}\big[h_z^*(X,Y)Y \big]  = \mathbb  E_{Y \mid X,Z=z}[Y] -(2z-1) \theta, \\
\mathbb{\tilde E}_{Y \mid X,Z=z}\big[\tilde h_z^2(X,Y)\big] & = \mathbb  E_{Y \mid X,Z=z} \big\{ [h_z^*(X,Y)]^2 \big\} = b_z^*.
\end{align*}
The second equality in the first line holds because the first constraint in \eqref{equ:new_ipw} is met with equality.
\end{proof}

\section{Background: influence functions}\label{sect:background_if}

The theory of influence functions \citep{van2002lectures} is crucial to removing the first-order bias in semiparametric estimation.
Let $\R$ denote the set of real numbers.
We treat a one-dimensional parameter $\tau$ as a mapping $\tau(\PP)$  from $\mathcal{P}$ to  $\R$ where $\mathcal{P}$ is the set of all possible observed data distributions. For any $\PP\in \mathcal{P}$, we define a path through $\PP$ as a one-dimensional submodel that passes through
$\PP$ at $\epsilon=0$ in the direction of a zero-mean function $s$ satisfying that $\|s\|_{2}\leq C'$ and $\epsilon\leq 1/C'$ for some constant $C'>0$. The submodel $\PP_{\epsilon}$ takes a density
$p_{O,\epsilon}(o) := p_{O}(o)[1+\epsilon s(o)]$ for $o\in \O$.
The tangent space $\mathcal{S}$ is defined as the set of zero-mean functions $s$ for any paths through $\PP$; $\mathcal{S}$ is known to be the Hilbert space of zero-mean functions when we use a nonparametric model.
Suppose $\tau(\PP)$ is differentiable at $\PP$ relative to $\mathcal{S}$, i.e., there is a linear mapping $\dot{\tau}(\cdot;\PP): \mathcal{S}\rightarrow \R$ such that for any $S\in \mathcal{S}$ and submodel $\PP_{\epsilon}$, we have
\begin{equation}\label{equ:if_definition}
\dot{\tau}(S;\PP)=\frac{d\psi_{\epsilon}}{d\epsilon}  \Big|_{\epsilon=0} = \E\big[ EIF(\tau)(O)  S(O)  \big],
\end{equation}
where $\tau_{\epsilon} : = \tau(\PP_{\epsilon})$ and $S(o) := \frac{d}{d\epsilon}\log p_{O,\epsilon}(o)\big|_{\epsilon=0}$.
In this article,  we let the subscript \quotes{$\epsilon$} denote the distribution shift from  $\PP$ to $\PP_{\epsilon}$ and \quotes{$S$} denote the score function for a random variable.
The second equality above is established by the Riesz representation theorem for the Hilbert space $\mathcal{S}$ that expresses $\dot{\tau}(\cdot;\PP)$ as an inner product with the unique {\it efficient influence function} $EIF(\tau):\O\rightarrow \R$ that lies in the closed linear span of $\mathcal{S}$. The efficiency of $EIF(\tau)$ implies that it has lower variance than any other influence functions, i.e., any measurable function $IF(\tau):\O\rightarrow \R$ whose projection onto the closed linear span of $\mathcal{S}$ is $EIF(\tau)$.
Some basic EIFs of expectation, conditional expectation and truncated expectations are given below, and we will use them in the following sections.

\begin{lemma}\label{lemma:eif_mean}
        Given one random variable $A$,  $EIF(\mathbb E[A] ) = A - \mathbb {E}[A].$
\end{lemma}
\begin{proof}
By the zero-mean property of influence functions,       it is straightforward to verify the EIF definition in \eqref{equ:if_definition}, i.e.,  $\frac{d}{d\epsilon}\mathbb E_{\epsilon}[A]|_{\epsilon=0} = \E\left[(A - \mathbb {E}[A])S(A) \right]$.
\end{proof}

\begin{lemma}\label{lemma:eif}
        Given two random variables $A$ and $B$,
\[
                EIF    \left\{\E[A\mid B=b] \right\} = \frac{\one_{\{B=b\}}}{p_{B}(b)}\left(    A -     \E[A\mid B=b]   \right).
\]
\end{lemma}
\begin{proof}
We  can directly verify the EIF definition in \eqref{equ:if_definition}:
\begin{align*}
&\hspace{12pt}  \frac{d}{d\epsilon}\E_{\epsilon}[A\mid B=b]\big|_{\epsilon=0}  = \int a \frac{d}{d\epsilon} \big[\log p_{A,B,\epsilon}(a,b) -        \log p_{B,\epsilon}(b)  \big]\big|_{\epsilon=0} p_{A\mid B}(a|b) da \\
                & =   \int a S(a,b) p_{A\mid B }(a\mid b) da - \int a  p_{A\mid B}(a|b) da \times \frac{d}{d\epsilon}  \log p_{B,\epsilon}(b) \big|_{\epsilon=0}        \\
                & = \int\int \frac{\one_{\{b'=b\}}}{p_{B}(b')} a S(a,b') p_{A, B }(a, b') da db' - \E[A\mid B=b]  \int \frac{\one_{\{b'=b\}}}{p_{B}(b')} \Bigg(  \frac{d}{d\epsilon}\log p_{B,\epsilon}(b') \big|_{\epsilon=0}          \\
                &\hspace{10pt} + \underbrace{\int  \frac{d}{d\epsilon}\log p_{A\mid B,\epsilon}(a\mid b') \big|_{\epsilon=0}p_{A\mid B}(a\mid b')     da  }_{=0}\Bigg)          p_{B}(b')d b'    \\
                                & = \int\int \frac{\one_{\{b'=b\}}}{p_{B}(b')} a S(a,b') p_{A, B }(a, b') da db' - \E[A\mid B=b]  \int \int \frac{\one_{\{b'=b\}}}{p_{B}(b')}  S(a,b')p_{A,B}(a, b')              da   d b'      \\
                & = \int\int \frac{\one_{\{b'=b\}}}{p_{B}(b')} \big(a - \E[A\mid B=b']  \big) S(a,b') p_{A, B }(a, b') da db'           = \E\big[             EIF\left\{\E[A\mid B=b]\right\}  S(A,B) \big],
\end{align*}
as required.
\end{proof}

\begin{lemma}\label{lemma:eif_2}
For $\psi$ defined in \eqref{equ:psi}, the EIF of $\E_{Y\mid X=x,Z=1}[Y\one_{\{Y\leq Q(X)\}}]$ is given by
\begin{align*}
          \frac{\one_{\{X=x,Z=1\}}}{p_{X,Z}(X,Z=1)}
        \big\{\left(\alpha_*-  \one_{\{Y\leq Q(X)\}} \right) Q(X)+Y  \one_{\{Y\leq Q(X)\}}  -\E_{Y\mid X=x,Z=1}[Y  \one_{\{Y\leq Q(X)\}}]       \big\}.
\end{align*}
\end{lemma}
\begin{proof}
By the definition of the quantile function,
\begin{align*}
&\ \alpha_* = \int^{Q_{\epsilon}(x)}  p_{Y\mid X,Z,\epsilon}(y|X,Z=1) dy  \\
&\ \hspace{6pt}0  = p_{Y\mid X,Z}(Q(x)\mid X,Z=1) \frac{d}{d\epsilon} Q_{\epsilon}(x)      \Big|_{\epsilon=0} + \frac{d}{d\epsilon} \int \one_{\{y\leq Q(x)\} }p_{Y\mid X,Z,\epsilon}(y|X,Z=1) dy   \big|_{\epsilon=0}\\
\Leftrightarrow  &\ \   p_{Y\mid X,Z}(Q(x)\mid X,Z=1) \frac{d}{d\epsilon} Q_{\epsilon}(x)          \Big|_{\epsilon=0} = - \frac{d}{d\epsilon} \int \one_{ \{y\leq Q(x)\}} p_{Y\mid X,Z,\epsilon}(y|X,Z=1) dy \big|_{\epsilon=0}.
\end{align*}
Using the Leibniz integral rule and the equation above,
\begin{align*}
& \frac{d}{d\epsilon} \E_{Y\mid X=x,Z=1,\epsilon}[Y\one_{\{Y\leq Q_{\epsilon}(X)\}}] \Big|_{\epsilon=0}
=  \frac{d}{d\epsilon}  \int^{Q_{\epsilon}(x)} y p_{Y\mid X,Z,\epsilon}(y|X,Z=1) dy   \Big|_{\epsilon=0} \\
=\ &   Q(x)p_{Y\mid X,Z}(Q(x)\mid X,Z=1) \frac{d}{d\epsilon} Q_{\epsilon}(x)       \Big|_{\epsilon=0} + \frac{d}{d\epsilon}  \int y \one_{\{y\leq Q(x)\}} p_{Y\mid X,Z,\epsilon}(y|X,Z=1) dy         \Big|_{\epsilon=0}     \\
        =\       &  \frac{d}{d\epsilon} \int \left[ - \one_{\{y\leq Q(x)\}} Q(x)+ y\one_{\{y\leq Q(x)\}} \right] p_{Y\mid X,Z,\epsilon}(y|X,Z=1) dy        \Big|_{\epsilon=0} \\
        =\       &  \int\int\int  \frac{\one_{\{x'=x,z'=1\}}}{p_{X,Z}(x',z')} \Big[ - \one_{\{y\leq Q(x')\}} Q(x') + y \one_{\{y\leq Q(x')\} } - \big(- \alpha_*        Q(x')    \\
                 &\hspace{38pt}  + \E_{Y\mid X=x',Z= 1}\big[Y \one_{ \{Y\leq Q(X)\} }\big] \big)        \Big] S(y,z',x') p_{y,z'x' }(y, z',x') dy dz'dx'                \\
=\               &  \E\big[             EIF\left\{\E_{Y\mid X=x,Z=1}[Y\one_{\{Y\leq Q(X)\}}]\right\} S(Y,Z,X)   \big],
\end{align*}
where the penultimate equality is obtained by the proof of \Cref{lemma:eif}  above.
\end{proof}

\section{Proof of theorem \ref{thm:if_psi_plus}}\label{appendix:if_psi_plus}

\begin{proof}
We first derive the expression of EIF via the calculus of IFs  \citep[Section 3.4.3]{kennedy2022semiparametric}, and then verify that the expression satisfies \eqref{equ:if_definition}. The IF of $\psi_+$ is given by
\begin{align*}
         &\  \sum_{x\in \mathcal{X}}IF\left\{p_{X}(x)\right\} W_+(x)\mu_+(x)   + \sum_{x\in \mathcal{X}}p_{X}(x)IF\left\{ W_+(x)\right\}\mu_+(x)   + \sum_{x\in \mathcal{X}}p_{X}(x) W_+(x) IF\left\{\mu_+(x)      \right\}                \\
        = &\    \sum_{x\in \mathcal{X}}\left[\one_{\{X=x\}} - p_{X}(x)\right] W_+(x) \mu_+(x)+  \sum_{x\in \mathcal{X}}p_{X}(x)\left(1-\Gamma\right) IF\left\{e(x)\right\}\mu_+(x)      \\
         &\ + \sum_{x\in \mathcal{X}}p_{X}(x) W_+(x)\Big( IF \left\{     \E_{Y\mid X,Z=1 }\left[Y  \right]\right\}   - IF \big\{  \E_{Y\mid X,Z=1} \big[Y \one_{\{Y\leq Q(X)\}}      \big]\big\}    \Big)   \\
        = &\ W_+(X) \mu_+(X) - \psi_+ + \left(1-\Gamma\right)\sum_{x\in \mathcal{X}}p_{X}(x) \mu_+(x)\frac{\one_{\{X=x\}}}{p_{X}(x)} \left[Z- e(X)      \right] \\
        &\ +\sum_{x\in \mathcal{X}}p_{X}(x) W_+(x)\frac{\one_{\{X=x,Z=1\}}}{p_{X,Z}(x,z)} \Big[Y- \E_{Y\mid X,Z=1}\left[Y\right]  -
         Y \one_{\{     Y\leq Q(X)\}} \\
        &\ \hspace{1cm}      + \E_{Y\mid X,Z=1}\left[Y \one_{\{  Y\leq Q(X)\}} \right]   - \left(\alpha_* - \one_{\{   Y\leq Q(X)\}} \right)   Q(X)    \Big] \\
        = &\ \left[(1-\Gamma)Z + \Gamma         \right]\mu_+(X)  -\psi_+  + \frac{W_+(X)Z}{e(X)}                \Big[   \left(1-  \alpha_* - \one_{\{   Y> Q(X)\}}  \right) Q(X) \\
        & \hspace{250pt} +      Y \one_{\{      Y> Q(X)\}} - \mu_+(X)   \Big].
\end{align*}
The first term in the third equality is attained by \Cref{lemma:eif_mean}. In the fourth equality, the second term is obtained by \Cref{lemma:eif}, and the third term is obtained by Lemmas \ref{lemma:eif} and \ref{lemma:eif_2}.
Next, we define
\begin{equation}\label{equ:f_+_}
f_+(X,Y) =   \left(1-  \alpha_* - \one_{\{      Y> Q(X)\}}  \right)Q(X)+ Y \one_{\{     Y> Q(X)\}}.
\end{equation}
Denote $\sum_{z=0}^{1}$ by  $\int dz$. We next verify the expression of $EIF(\psi_+)$:
\begin{align*}
       & \  \E\big[ EIF(\psi_+)(O)S(O)  \big] \\
        =     &         \int\int\int EIF(\psi_+)(y,z,x)\frac{d}{d\epsilon}\log p_{Y,Z,X,\epsilon}(y,z,x)\big|_{\epsilon=0} p_{Y,Z,X}(y,z,x)   dydzdx  \\
        &\hspace{-12pt}  =         \int\int\int EIF(\psi_+)(y,z,x)\frac{d}{d\epsilon} p_{Y\mid X,Z,\epsilon}(y\mid x,z)\big|_{\epsilon=0} p_{Z\mid X}(z\mid x)p_{X}(x) dydzdx   \\
                &  + \int\int\int EIF(\psi_+)(y,z,x)\frac{d}{d\epsilon} p_{Z\mid X,\epsilon}(z\mid x)\big|_{\epsilon=0} p_{Y\mid X,Z}(y\mid x,z)p_{X}(x) dydzdx \\
         & + \int\int\int EIF(\psi_+)(y,z,x)\frac{d}{d\epsilon} p_{X,\epsilon}(x)\big|_{\epsilon=0} p_{Y\mid X,Z}(y\mid x,z)p_{Z\mid X}(z\mid x) dydzdx         \\
        &\hspace{-12pt}  =
         \int \int   \left[W_+(x)   [f_+(x,y) -\mu_+(x)  ]              +      e(x)  \mu_+(x)        \right]  \frac{d}{d\epsilon} p_{Y\mid X,Z,\epsilon}(y\mid X,Z=1)\big|_{\epsilon=0}  p_{X}(x) dy dx    \\
        & + \int \int    \Gamma  \mu_+(x)    \frac{d}{d\epsilon} p_{Y\mid X,Z,\epsilon}(y\mid x,0)\big|_{\epsilon=0} \left[ 1-e(x)       \right] p_{X}(x)  dy dx         \\
        & +\underbrace{ \int  \mu_+(x)  \frac{d}{d\epsilon} e_{\epsilon}(x)\big|_{\epsilon=0} p_{X}(x) dx +  \int  \Gamma \mu_+(x)                               \frac{d}{d\epsilon}[1-  e_{\epsilon}(x)]\big|_{\epsilon=0} p_{X}(x) dx }_{(*)}\\
        & + \underbrace{\int \mu_+(x) e(x) \frac{d}{d\epsilon} p_{X,\epsilon}(x)\big|_{\epsilon=0} dx + \int  \Gamma \mu_+(x)\left[1-e(x)\right] \frac{d}{d\epsilon} p_{X,\epsilon}(x)\big|_{\epsilon=0} dx }_{(**)}    \\
        &\hspace{-12pt} = \int W_{+}(x) \frac{d}{d\epsilon} \mu_{+,\epsilon}(x)\big|_{\epsilon=0}       p_{X}(x)  dx + \underbrace{\frac{d}{d\epsilon}  \int    W_{+}(x) \mu_+(x)                                p_{X}(x)  dx   \big|_{\epsilon=0}      }_{=0}  \\
        & + \underbrace{ \int \mu_+(x)  \frac{d}{d\epsilon} W_{+,\epsilon}(x)\big|_{\epsilon=0} p_{X}(x) dx}_{(*)} + \underbrace{\int W_+(x)\mu_+(x) \frac{d}{d\epsilon} p_{X,\epsilon}(x)\big|_{\epsilon=0} dx}_{(**)}
 = \frac{d \psi_{+,\epsilon}}{d\epsilon}  \big|_{\epsilon=0},
\end{align*}
as required by \eqref{equ:if_definition}. The first term in the penultimate equality is attained by
\begin{align*}
&\ \int                  \left[f_+(x,y) - \mu_{+}(x)    \right]\frac{d}{d\epsilon} p_{Y\mid X,Z,\epsilon}(y\mid X,Z=1)\big|_{\epsilon=0} dy        \\
= &\ \int \int \int     \underbrace{\frac{\one_{\{x'=x,z=1\}} }{p_{X,Z}(x',z)} \left[f_+(x',y) - \mu_{+}(x')    \right] }_{ \text{  $=   EIF\left\{ \mu_{+}(x)\right\}$ by \Cref{lemma:eif}}} \log p_{Y\mid X,Z,\epsilon}(y\mid x',z)\big|_{\epsilon=0}                          p_{Y,Z,X}(y,z,x')              dy      dz              dx'     \\
 = &\   \int \int \int  EIF\left\{\mu_{+}(x)\right\} \underbrace{ \frac{d}{d\epsilon} \log p_{Y,Z,X,\epsilon}(y,z,x')\big|_{\epsilon=0}                 }_{=S(y,z,x')}   p_{Y,Z,X}(y,z,x')       dy      dz     dx'             \\
   &\ - \int \int \underbrace{ \int  EIF\left\{ \mu_{+}(x)\right\} p_{Y\mid X,Z} (y\mid x',z)dy}_{ = 0} \frac{d}{d\epsilon}  \log p_{X,Z,\epsilon}(x,z)\big|_{\epsilon=0}                p_{X,Z}(x',z)                   dz     dx'
 \\
=&\             \int \int \int  EIF\left\{\mu_{+}(x)\right\}S(y,z,x')
 p_{Y,Z,X}(y,z,x')       dy      dz     dx'                                                             = \frac{d  }{d\epsilon} \mu_{+,\epsilon}(x) \big|_{\epsilon=0}.
\end{align*}
The proof for $EIF(\psi_-) = \phi_-(O)-\psi_-$ follows the same steps so omitted here. The differences are replacing $\Gamma$ by $\Gamma^{-1}$ and truncating $Y$ below the quantile rather than above.
\end{proof}

\section{Proof of Proposition \ref{thm:bias}}\label{appendix:bias}

\subsection{Proof sketch}\label{equ:sketch}

Define the expected outcomes above and below the {\it estimated} quantile,
\[
\hat{\mu}_{+ } (X) = \E_{Y\mid X,Z=1}\big[Y \one_{\{Y>\hat Q(X)\}} \big]\ \text{ and }\ \hat{\mu}_{-} (X) =  \E_{Y\mid X,Z=1}\big[Y  \one_{\{Y<\hat Q(X)\}} \big].
\]
 In the subsection below, we will prove a bias decomposition
\[
\Bias(\hat \phi \mid \hat \eta)  = \Bias_1(\hat \phi\mid \hat \eta) +  \Bias_2(\hat \phi\mid \hat \eta),
\]
where $ \Bias_1(\hat \phi  \mid \hat \eta )$ is given by
\begin{align*}
& \Gamma   \E\left\{ \frac{\hat e(X)- e(X)}{\hat{e}(X)}   \left[ \Big(\PP_{Y\mid X,Z=1}\big\{Y>\hat Q(X)\big\} - [1-\alpha_* ] \Big) \hat Q(X)+  \hat{\hat{\mu}}_+(X) - \hat{\mu}_{+ } (X)   \right] \ \Big| \  \hat \eta  \right\}                \\
&\ + \Gamma^{-1}\E\left\{ \frac{\hat e(X)-\hat e(X)}{\hat{e}(X)}   \left[ \Big(\PP_{Y\mid X,Z=1}\big\{Y<\hat Q(X)\big\} -\alpha_* \Big) \hat Q(X) + \hat{\hat{\mu}}_-(X)  - \hat{\mu}_-(X)  \right]  \ \Big| \  \hat \eta  \right\},
\end{align*}
and
\begin{align*}
\Bias_2(\hat \phi \mid \hat \eta  )   =  &\ \mathbb E\Bigg\{  \bigg(W_+(X)-W_-(X)\bigg) \bigg( \big[\hat{Q}(X)- Q(X)\big] \big[ \P_{Y\mid X,Z=1} {\{Y<\hat Q(X)\}}-\alpha_*\big]             \\
&\hspace{15pt} - \E_{Y\mid X,Z=1}\left[|Y-Q(X)|\one_{\{ Q(X)\wedge \hat Q(X) <Y< Q(X)\vee\hat Q(X)   \}   }      \right] \bigg) \ \bigg| \  \hat \eta  \Bigg\} \\
\lesssim    & \             \mathbb E\big\{ [\hat Q(X) - Q(X)]^2  \mid \hat \eta   \big\}            =    o_{\PP}(n^{-1/2}),
\end{align*}
using the facts that
\[
W_+(X)-W_-(X) < \Gamma - \Gamma^{-1}, \
\PP_{Y\mid X,Z=1}\{Y<\hat Q(X)\}
-\alpha^* = O_{\PP}( |\hat Q(X) - Q(X)|  ),
\]
and $|Y-Q(X)|\leq |\hat Q(X) - Q(X)|$ under the event $Q(X)\wedge\hat Q(X)<Y< Q(X)\vee \hat Q(X).$ The probability of the event also decays as fast as  $|\hat Q(X) - Q(X)|$.
In $\Bias_1(\hat \phi  \mid \hat \eta )  $, we can rewrite
\[
\hat{\hat{\mu}}_{+}(X) - \hat{\mu}_{+} (X) =\hat{\hat{\mu}}_{+} (X) - \mu_{+}(X) + \mu_{+}(X) -  \hat{\mu}_{+} (X).
\]
By definition, $\hat{\mu}_{+}$ converges to $\mu_{+}$ as fast as $\hat Q$ converges to $Q$, i.e.,
 \[
\hat \mu_+(X) -  \mu_+(X) \lesssim  \E_{Y\mid X,Z=1}\big[Y \one_{\{ Q(X)\wedge \hat Q(X) <Y< Q(X)\vee \hat Q(X) \}}  \big]\lesssim |\hat Q(X) - Q(X)|.
 \]
The same argument applies to $\mu_{-} (X) -  \hat{\mu}_{-}(X)$. Using Cauchy-Schwarz and Assumption \ref{assumption:consistency}, we have $ \Bias_1(\hat \phi  \mid \hat \eta )    = o_{\PP}(n^{-1/2})$.

\subsection{Bias expressions}

\begin{proof}
By definition, $\phi(O;\hat \eta ) =  \phi_+(O;\hat \eta )  + \phi_-(O;\hat \eta ).$
Following \eqref{equ:f_+_}, we define
\[
\hat f_+(X,Y) =   \big(1-  \alpha_* - \one_{\{  Y> \hat Q(X)\}}   \big)\hat  Q(X)+ Y \one_{\{   Y> \hat Q(X)\}},
\]
We rewrite the uncentered EIF $\phi_+(O;\hat \eta )$ as
\begin{align*}
\phi_+(O;\hat{e},\hat{Q},\hat{\hat \mu}_+) := &\  \frac{Z \hat{W}_+(X)}{\hat{e}(X)}\left[\hat{f}_+(X,Y) - \hat{\hat{\mu}}_+(X) \right] + \left[(1-\Gamma)Z+\Gamma \right]\hat{\hat{\mu}}_+(X)  \\
        = &\  \frac{Z \hat{W}_{+}(X)}{\hat{e}(X)}\left[\hat{f}_+(X,Y) - \hat{\hat{\mu}}_+(X) \right] + \left[(1-\Gamma)Z+\Gamma \right]\left[   \hat{\hat{\mu}}_+(X) - \hat{f}_+(X,Y)   \right]\\
          &\ + \left[(1-\Gamma)Z+\Gamma \right] \hat{f}_+(X,Y).
\end{align*}
The difference $\phi_+(O;\hat{e},\hat{Q},\hat{\hat \mu}_+) - \phi_+(O;e,Q,\mu_+)$ can be written as
\begin{align*}
&\underbrace{ \phi_+(O;\hat{e},\hat{Q},\hat{\hat \mu}_+)
- \phi_+(O;e,\hat{Q},\hat{\hat \mu}_+) }_{(\text{a})} + \underbrace{ \phi_+(O;e,\hat{Q},\hat{\hat \mu}_+)
 -\phi_+(O;e,Q,\mu_+)                   }_{(\text{b})} \\
=\ & \underbrace{ \Gamma Z \left[1/\hat{e}(X) - 1/e(X) \right]\big[ \hat{f}_+(X,Y) - \hat{\hat{\mu}}_+(X) \big]}_{(\text{a}) \text{ by the first expression of $\phi_+$ above}} + \underbrace{\left[(1-\Gamma)Z + \Gamma\right]  \left[\hat{f}_+(X,Y) - f_+(X,Y) \right] }_{(\text{b-1})} \\
& +\underbrace{ \Gamma \left[Z/e(X) -1 \right]\big[\hat{f}_+(X,Y) - \hat{\hat{\mu}}_+(X)  -f_+(X,Y) + \mu_+(X)   \big]}_{(\text{b-2})}.
\end{align*}
Using $\mathbb E\{Z/e(X)-1 \mid X\}=0,$ we have
\begin{align*}
\text{Bias}(\hat \phi_+ \mid \hat \eta ) = &\   \E\bigg\{     \phi_+(O;\hat{e},\hat{Q},\hat{\hat \mu}_+)
- \phi_+(O;e,\hat{Q},\hat{\hat \mu}_+) + \phi_+(O;e,\hat{Q},\hat{\hat \mu}_+)  -\phi_+(O;e,Q,\mu_+)              \ \big| \   \hat \eta   \bigg\} \\
 = &\    \Gamma \E\left\{ \frac{ \hat e(X)-e(X)}{\hat e(X)}   \left[  \hat{\hat{\mu}}_+(X) -  \E_{Y\mid X,Z=1}[\hat f_+(X,Y) ]  \right]    \ \Big| \   \hat \eta  \right\}                \\
 &       + \E\left\{\frac{ZW_+(X)}{e(X)} [\hat{f}_+(X,Y) - f_+(X,Y) ]   \ \Big| \   \hat \eta  \right\}.
\end{align*}
Next, we define the functions
\begin{align*}
& f_-(X,Y) =   \big(\alpha_* - \one_{\{ Y< Q(X)\}} \big)Q(X)+ Y \one_{\{        Y< Q(X)\}},\\
&
 \hat f_-(X,Y) =  \big(\alpha_* - \one_{\{      Y< \hat Q(X)\}}  \big)\hat Q(X)+ Y \one_{\{     Y<  \hat Q(X)\}}.
\end{align*}
We can rewrite $\phi_-(O;\hat \eta )$ as
\[
\phi_-(O;\hat{e},\hat{Q},\hat{\hat \mu}_-) = \frac{Z \hat{W}_-(X)}{\hat{e}(X)}\left[\hat{f}_-(X,Y) - \hat{\hat{\mu}}_-(X) \right] + \left[(1-\Gamma^{-1})Z+\Gamma^{-1} \right]\hat{\hat{\mu}}_-(X),
\]
and show that
\begin{align*}
\Bias(\hat \phi_-\mid \hat \eta ) = &\  \Gamma^{-1} \E\left\{ \frac{\hat e(X)-e(X)}{\hat e(X)}   \left[ \hat{\hat{\mu}}_-(X) - \E_{Y\mid X,Z=1}[ \hat f_-(X,Y) ] \right]  \ \Big| \   \hat \eta  \right\}\\
& + \E\left\{ \frac{ZW_-(X)}{e(X)}                              \left[\hat{f}_{-}(X,Y) - f_{-}(X,Y) \right]  \ \Big| \   \hat \eta  \right\}.
\end{align*}
Summing up the first terms in $\Bias(\hat \phi_+\mid \hat \eta )$ and $\Bias(\hat \phi_-\mid \hat \eta )$ gives  $\Bias_1(\hat \phi\mid \hat \eta )$ in \Cref{equ:sketch}. Then,
$\Bias_2(\hat \phi\mid \hat \eta )$ in \Cref{equ:sketch} can be derived from the second terms,
\begin{equation}\label{equ:theorem_2_combine}
\E\left\{W_{+}(X)\left[\hat{f}_+(X,Y) - f_+(X,Y) \right]  + W_{-}(X)\left[\hat{f}_-(X,Y) - f_-(X,Y)\right] \ \Big| \  Z=1, \hat \eta \right\}.
\end{equation}
In this expectation, we first consider the terms of trimmed outcomes,
\begin{align*}
        &  W_+(X)\left[Y \one_{\{Y>\hat Q(X)\}} -       Y \one_{\{Y> Q(X)\}}     \right]                + W_-(X)\left[Y \one_{\{Y< \hat Q(X)\}} -       Y\one_{\{Y< Q(X)\}      } \right]                       \\
        = &\ \begin{cases}
                0,  \quad \quad \quad \quad \quad \quad \quad \quad  \hspace{11pt} \ \ \ \   \text{if }         Y>\hat Q(X), Y> Q(X); \\
                W_-(X)Y -W_+(X)Y,    \ \ \ \  \text{ if } Y<\hat Q(X), Y> Q(X);       \\
                W_+(X)Y -W_-(X)Y,     \ \ \ \   \text{ if } Y>\hat Q(X), Y< Q(X);    \\
                0, \quad \quad \quad \quad \quad \quad \quad \quad   \hspace{10.5pt}\ \ \ \   \text{if }
                Y<\hat Q(X), Y< Q(X). \end{cases}
\end{align*}
Then, we  consider the other terms,
\begin{align*}
          &\    \underbrace{ W_+(X)\Big[1-\alpha_* -\one_{\{Y>\hat{Q}(X)\} }
         \Big]\hat{Q}(X) + W_-(X)\Big[\alpha_* - \one_{ \{Y<\hat{Q}(X)\}}       \Big]\hat{Q}(X) }_{\text{(a)}}  \\
        &\      \underbrace{- W_{+}(X)\Big[1-\alpha_* -\one_{\{Y>\hat{Q}(X)\}} \Big] Q(X)       -       W_-(X)\Big[\alpha_* -\one_{\{Y<\hat{Q}(X)\} }\Big]Q(X)}_{\text{$-$(b)}} \\
        &\  \underbrace{ W_+(X)\Big[ 1-\alpha_* -\one_{\{Y>\hat{Q}(X)\}} \Big]Q(X)      + W_-(X)\Big[\alpha_* -\one_{\{Y<\hat{Q}(X)\} }\Big]Q(X) }_{\text{$+$(b)}}      \\
        &\ \underbrace{- W_+(X)\Big[ 1-\alpha_* -\one_{\{Y>Q(X)\}} \Big]Q(X)
        -       W_-(X)\Big[\alpha_* - \one_{ \{Y<Q(X)\}} \Big]Q(X)
        }_{\text{(c)}}  \\
        =&\ \underbrace{        W_+(X)\Big[1- \alpha_* -\one_{ \{Y>\hat{Q}(X)\} } \Big]\Big[\hat{Q}(X) - Q(X) \Big]     + W_-(X)\Big[\alpha_* -\one_{ \{Y<\hat{Q}(X)\}} \Big] \Big[\hat{Q}(X) - Q(X) \Big] }_{\text{(a)$-$(b)}}\\
        &\  \underbrace{+       W_+(X) \Big[-\one_{ \{Y>\hat Q(X)\}} + \one_{ \{Y>Q(X)\}        }       \Big]   Q(X)                                                             +      W_-(X) \Big[-\one_{\{Y<\hat Q(X)\} }+ \one_{ \{Y<Q(X)\}}                \Big]   Q(X)                    }_{\text{(b)$-$(c)}}                            \\
        =&\ \left[W_+(X)-W_-(X)\right]  \big[ \one_{ \{Y<\hat{Q}(X)\} } -\alpha_* \big]\big[\hat{Q}(X) - Q(X) \big]
                + W_+(X)        \one_{ \{Y=\hat{Q}(X)\} }\big[\hat{Q}(X) - Q(X) \big]\\
        &\      + \begin{cases}
       0,\hspace{132pt}    \text{if }         Y>\hat Q(X), Y> Q(X); \\
                W_{+}(X)Q(X)  -    W_{-}(X)Q(X)            ,    \hspace{7pt}     \text{ if } Y<\hat Q(X), Y> Q(X);      \\
                W_-(X)Q(X)              -W_{+}(X)Q(X)                   ,    \hspace{7pt}    \text{ if } Y>\hat Q(X), Y< Q(X);      \\
                0, \hspace{132pt}  \text{if }
                Y<\hat Q(X), Y< Q(X).
                \end{cases}
\end{align*}
Summing up the end of the last two equations shows that  \eqref{equ:theorem_2_combine} can be written as  the expectation of the product of $W_+(X)-W_-(X)$ and
\begin{align*}
 \big[\hat Q(X) - Q(X)\big]     \big[\one_{ \{Y<\hat{Q}(X)\} } -\alpha_* \big]   +      \begin{cases}
       0,\hspace{48pt}    \text{if }         Y>\hat{Q}(X), Y>Q(X); \\
                Q(X)    -Y      ,\hspace{10pt}    \text{if } Y<\hat{Q}(X), Y>Q(X);       \\
                Y- Q(X),        \hspace{9pt} \text{if } Y>\hat{Q}(X), Y<Q(X);   \\
                0, \hspace{49pt}  \text{if }
                Y<\hat{Q}(X), Y<Q(X), \end{cases}
\end{align*}
conditioning on $Z=1$ and $\hat \eta $. Further conditioning on $X$ gives $\Bias_2 (\hat \phi\mid \hat \eta )$ in \Cref{equ:sketch}.
\end{proof}

\section{Proof of Theorems \ref{prop:eif_small} and \ref{prop:eif_large}}\label{appendix:prop:eif_large}

The proofs presented below involve many calculations of functional derivatives. We first note that a brief introduction to influence functions is given in \Cref{sect:background_if}.
In the following proofs, we begin by calculating
the repeatedly used derivatives in \Cref{sect:H_1}. After that, we derive the efficient influence functions (EIFs) for $\psi_1$ and $\psi_2$ in \Cref{sect:psi_2_sect} and $\psi_3$ in \Cref{sect:psi_3_sect}, respectively.

\subsection{Two derivatives}\label{sect:H_1}

We start by deriving two derivatives.  To simplify the exposition,  we denote
\begin{align*}
g (X,Y) = (\xi_X-Y) \one_{\{Y \leq \xi_X \}} &  \  \text{ and } \
\bar g(X)  = \E_{Y\mid X,Z=1}\left[(\xi_X-Y) \one_{\{Y \leq \xi_X \}}\right],     \\
b  (X,Y) = (\xi_X-Y)Y \one_{\{Y \leq \xi_X \}} &  \  \text{ and } \
\bar b (X)  = \E_{Y\mid X,Z=1}\left[(\xi_X-Y)Y \one_{\{Y \leq \xi_X \}}\right].
\end{align*}
It holds that $b (X,Y) = g (X,Y)Y.$
First,
\begin{equation}\label{equ:derivative_kappa_0}
\begin{split}
\frac{d \bar g_{\epsilon} (X) }{d\epsilon}  \Big|_{\epsilon=0}  = &     \frac{d}{d\epsilon} \bigg\{      \int^{\xi_{X,\epsilon}}(\xi_{X,\epsilon} - y)  p_{Y\mid X,Z=1,\epsilon}(y)dy     \bigg\} \Big|_{\epsilon=0}                      \\
= &\ 0\times     \frac{d \xi_{X,\epsilon}}{d\epsilon} \Big|_{\epsilon=0} +  \frac{d \xi_{X,\epsilon}}{d\epsilon} \Big|_{\epsilon=0} \E_{Y\mid X,Z=1}\left[\one_{\{Y \leq  \xi_X \}}\right]  \\
&\ + \frac{d}{d\epsilon} \bigg\{         \int (\xi_X  - y)\one_{\{Y \leq \xi_X \}}      p_{Y\mid X,Z=1,\epsilon}(y)dy     \bigg\} \Big|_{\epsilon=0}  \\
= &\ \frac{d \xi_{X,\epsilon}}{d\epsilon} \Big|_{\epsilon=0}  \E_{Y\mid X,Z=1}\left[ Y \one_{\{Y \leq \xi_X \}}\right]  + \frac{d}{d\epsilon} \mathbb E_{Y\mid X,Z=1,\epsilon}\left[        g (X,Y)         \right]   |_{\epsilon=0}        \\
= &\  \frac{d \xi_{X,\epsilon}}{d\epsilon} \Big|_{\epsilon=0} \E_{Y\mid X,Z=1}\left[ \one_{\{Y \leq \xi_X \}}\right]  +  \E\left[ \frac{Z}{e(X)}\left[g  (X,Y)-  \bar g(X) \right]        S(O)\   \Big|\ X \right].
\end{split}
\end{equation}
The last line uses the proof of \Cref{lemma:eif}, which implies that
\begin{align*}
        \frac{d}{d\epsilon} \mathbb E_{Y\mid X=x,Z=1,\epsilon}\left[    g (X,Y)         \right]   |_{\epsilon=0} = \E\left[     \frac{Z}{e(X) }\left[g  (X,Y)-  \bar g  (X) \right]     S(O)    \Big|\ X=x               \right].
\end{align*}
Similarly, we can show that
\begin{equation}\label{equ:gamma_1_e}
\begin{split}
\frac{d \bar b_{\epsilon} (X) }{d\epsilon}  \Big|_{\epsilon=0}   =& \           \frac{d}{d\epsilon} \bigg\{      \int^{\xi_{X,\epsilon}}(\xi_{X,\epsilon} - y)y p_{Y\mid X,Z=1,\epsilon}(y)dy     \bigg\} \Big|_{\epsilon=0}\\
= &\ \frac{d \xi_{X,\epsilon}}{d\epsilon} \Big|_{\epsilon=0} \E_{Y\mid X,Z=1}\left[ Y \one_{\{Y \leq \xi_X \}}\right]  + \E\left[ \frac{Z}{e(X)}\left[b  (X,Y)-  \bar b (X) \right]       S(O)    \       \Big|\   X \right].
\end{split}
\end{equation}

\subsection{Proof of Theorem \ref{prop:eif_small}}\label{sect:psi_2_sect}

\begin{proof}
We first prove the EIF of $\psi_2$. We will use the notation in \Cref{sect:H_1}. First,
\begin{equation}\label{equ:g_g_bar}
\lambda g  (X,Y)-  \lambda  \bar g  (X)  = h_*(X,Y)-e(X) - [1-e(X)] = h_*(X,Y) - 1.
\end{equation}
Differentiating both sides of $\lambda \bar g  (X)  = 1-e(X) $  and plugging in, \eqref{equ:derivative_kappa_0}, we have
\begin{equation}\label{equ:lambda_xi}
\lambda  \frac{d \xi_{X,\epsilon}}{d\epsilon} \Big|_{\epsilon=0}
   =  \E\left\{\left[     \Pi_e(X,Y) +\frac{Z }{e(X)}\Pi_h(X,Y) \right]S(O)\    \Big|\ X \right\} .
\end{equation}
The term  $\Pi_h(X,Y)$ is attained using \eqref{equ:g_g_bar}. Then,
\begin{align*}
&\ \frac{d}{d\epsilon}\E_{Y\mid X,Z=1,\epsilon}\left[h_{*,\epsilon}(X,Y)Y\right] \big|_{\epsilon=0} \\
=&\ \E\left[    \frac{Z}{e(X)}\left( h_{*}(X,Y)Y -  \E_{Y\mid X, Z=1}[h_{*}(X,Y)Y]\right) S(O)    \       \Big|\   X \right]              \\
&\ +\frac{d  }{d\epsilon}\left\{ e_{\epsilon}(X)\E_{Y\mid X,Z=1}\left[ Y \right] + \lambda \xi_{X,\epsilon}  \E_{Y\mid X,Z=1}\left[ Y \one_{\{Y \leq \xi_X \}}\right] \right\} \big|_{\epsilon=0}                                                                                                                   \\
        =&\     \E\Bigg[        \bigg( \frac{Z}{e(X)}\Big[     \Pi_h(X,Y)\E_{Y\mid X,Z=1}\left[ Y \one_{\{Y \leq \xi_X \}}\right]     + h_{*}(X,Y)Y -  \E_{Y\mid X, Z=1}[h_{*}(X,Y)Y] \Big]  \\
&\ \hspace{15pt}        + \Pi_e(X,Y)\E_{Y\mid X,Z=1}\left[ Y \one_{\{Y \leq \xi_X \}}\right] +[Z-e(X)]    \E_{Y\mid X,Z=1}\left[ Y \right]          \bigg) S(O)     \       \bigg|\   X \Bigg].
\end{align*}
Using \Cref{lemma:eif_mean} for $\psi_2 =  \E_{X} \left[\mu_{1,h_*}(X)\right] = \E_{X} \left[\E_{Y\mid X, Z=1}[h_*(X,Y)Y]\right]$,
\begin{align*}
& \frac{d\psi_{2,\epsilon}}{d\epsilon}\Big|_{\epsilon=0}
 =  \E_{X}\left[  \frac{d}{d\epsilon} \E_{Y\mid X, Z=1,\epsilon}[h_{*,\epsilon}(X,Y)Y]     \big|_{\epsilon=0} \right]  + \frac{d}{d\epsilon} \E_{X,\epsilon}\left[ \mu_{1,h_*}(X) \right]     \big|_{\epsilon=0}              \\
& =     \E\Bigg[ \bigg( \frac{Z}{e(X)}\Big[     \Pi_h(X,Y)\E_{Y\mid X,Z=1}\left[ Y \one_{\{Y \leq \xi_X \}}\right]     + h_{*}(X,Y)Y -  \E_{Y\mid X, Z=1}[h_{*}(X,Y)Y] \Big]        \\
& \hspace{15pt}  \Pi_e(X,Y)\E_{Y\mid X,Z=1}\left[ Y \one_{\{Y \leq \xi_X \}}\right] +[Z-e(X)]\mathbb E_{Y\mid X,Z=1}[Y] +\E_{Y\mid X, Z=1}[h_*(X,Y)Y] -  \psi_2       \bigg)  S(O)     \Bigg],
\end{align*}
which shows that  $EIF(\psi_2)=\phi_2(O)-\psi_2$ satisfies the definition in \eqref{equ:if_definition}.

We now consider the EIF of  $\psi_1$. First, by the definition of $h_*(X,Y)$ in \Cref{prop:double_ipw},
\[
\E_{Y\mid X,Z=1}\left[h_{*}^2(X,Y)\right]  = e(X)[2-e(X)] + \lambda^2\E_{Y\mid X,Z=1}\left[g^2(X,Y)\right].
\]
Using $\lambda \bar g  (X)  = 1-e(X),$ we rewrite the derivative
\begin{align*}
& \lambda^2 \frac{d}{d\epsilon}\E_{Y\mid X,Z=1,\epsilon}\left[g_{\epsilon}^2(X,Y)\right]\big|_{\epsilon=0} \\
 = &\  2\lambda [1-e(X)] \frac{d \xi_{X,\epsilon}}{d\epsilon} \Big|_{\epsilon=0}   + \lambda^2 \E\left[ \frac{Z}{e(X)}\left(g^2  (X,Y)-  \E_{Y\mid X,Z=1}\left[g^2(X,Y)\right] \right)    S(O)\   \Big|\ X \right]        \\
= &\  \E\bigg\{ \Big[  \frac{Z}{e(X)}\Big( 2[1-e(X)]\Pi_h(X,Y)+ \lambda^2g^2  (X,Y)- \lambda^2 \E_{Y\mid X,Z=1}\left[g^2(X,Y)\right] \Big) \\
&\ + 2[1-e(X)] \Pi_e(X,Y)         \Big] S(O)\  \Big|\ X \bigg\}.
\end{align*}
Applying \Cref{lemma:eif_mean} to  $\psi_1 =      \E_{X}\left\{\nu_{1,h^*}(X)\right\}=             \E_{X} \left[\E_{Y\mid X, Z=1}[h_*^2(X,Y)]\right]$, we have
\begin{align*}
\frac{d\psi_{1,\epsilon}}{d\epsilon}\Big|_{\epsilon=0}  & =  \E_{X}\left[  \frac{d}{d\epsilon} \E_{Y\mid X, Z=1,\epsilon}[h_{*,\epsilon}^2(X,Y)]    \big|_{\epsilon=0} \right]  + \frac{d}{d\epsilon} \E_{X,\epsilon}\left[ \nu_{1,h^*}(X)	\right]    \big|_{\epsilon=0}              \\
& =     \E\Bigg[ \bigg( \frac{Z}{e(X)}\Big[    2[1-e(X)]\Pi_h(X,Y)+ \lambda^2g^2  (X,Y)- \lambda^2 \E_{Y\mid X,Z=1}\left[g^2(X,Y)\right]           \Big]    \\
& \hspace{51pt}                 + 2[1-e(X)][Z-e(X) + \Pi_e(X,Y)] +\E_{Y\mid X, Z=1}[h_*^2(X,Y)] -  \psi_1 \bigg)  S(O)     \Bigg],
\end{align*}
which shows that $EIF(\psi_1)=\phi_1(O)-\psi_1$ satisfies the definition in \eqref{equ:if_definition}.
\end{proof}

\subsection{Proof of Theorem \ref{prop:eif_large}}\label{sect:psi_3_sect}

\begin{proof}
By the definition of  $\xi_X$ in \eqref{equ:ratio},
$\bar b (X)    =\Delta_X \bar g (X)$ for
\[
\Delta_X  := \E_{Y\mid X,Z=1 }[Y]  -\theta/[1-e(X)].
\]
Applying this equality to connect the derivative of $\bar b (X) $ and $\bar g (X) $ in \Cref{sect:H_1}:
\[
\frac{d \bar b_{\epsilon} (X) }{d\epsilon} =\frac{d \Delta_{X,\epsilon} }{d\epsilon}\Big|_{\epsilon=0}\bar g (X)  + \Delta_X\frac{d \bar g_{\epsilon} (X) }{d\epsilon}  \Big|_{\epsilon=0}.
\]
Plugging in the derivative expressions, we arrive at
\begin{align*}
&\ \frac{d \xi_{X,\epsilon}}{d\epsilon} \Big|_{\epsilon=0} \E_{Y\mid X,Z=1}\left[ Y \one_{\{Y \leq \xi_X \}}\right]  + \E\left[     \frac{Z}{e(X)}\left[b  (X,Y)-  \bar b (X) \right]       S(O)    \       \Big|\   X \right]    - \frac{d \Delta_{X,\epsilon} }{d\epsilon}\Big|_{\epsilon=0}\bar g (X)   \\
=&\  \Delta_X \bigg\{ \frac{d \xi_{X,\epsilon}}{d\epsilon} \Big|_{\epsilon=0} \E_{Y\mid X,Z=1}\left[ \one_{\{Y \leq \xi_X \}}\right]  +  \E\left[     \frac{Z}{e(X)}\left[g   (X,Y)-  \bar g  (X) \right]     S(O)\   \Big|\ X \right] \bigg\}.
\end{align*}
Rearranging the terms and applying $\bar b (X)    =\Delta_X \bar g (X)$, we have
\begin{equation}\label{equ:rearrange}
        \begin{split}
\hspace{-7pt}           &\ \E_{Y\mid X,Z=1}\left[(Y-\Delta_X) \one_{\{Y \leq \xi_X \}}\right]\frac{d \xi_{X,\epsilon}}{d\epsilon} \Big|_{\epsilon=0}
+ \bar g (X)\frac{d \Delta_{X,\epsilon} }{d\epsilon}\Big|_{\epsilon=0}
\\
\hspace{-7pt}= &\               2\bar g (X) \frac{d \Delta_{X,\epsilon} }{d\epsilon}\Big|_{\epsilon=0} +  \E\left[      \frac{Z}{e(X)}\left(  \Delta_X-Y\right) g(X,Y)          S(O)    \       \Big|\   X \right]      \\
\hspace{-7pt}= &\  \E\bigg\{    \bigg(  \frac{Z}{e(X)} \left[   2\bar g (X)     \left(Y - \E_{Y\mid X,Z=1 }[Y]    \right)         + \left(  \Delta_X-Y\right) g(X,Y)      \right] \\
&\ \hspace{125pt} - 2\theta\bar g (X)  [Z-e(X)]/[1-e(X)]^2 \bigg)S(O)                     \       \bigg|\   X  \bigg\}.
        \end{split}
\end{equation}
Observe that
$\E_{Y\mid X,Z=1}\left[(Y  -  \Delta_X ) (\xi_X-Y) \one_{\{Y \leq \xi_X \}}\right]   = \bar b (X)- \Delta_X \bar g(X)
  = 0.$ Then,
\begin{equation}\label{equ:step_3}
( \xi_X  -  \Delta_X )\bar g(X)  =  \E_{Y\mid X,Z=1}\left[g^2(X,Y)\right] \ \text{ and }
\end{equation}
\begin{equation}\label{equ:reformulation}
\begin{split}
 \E_{Y\mid X,Z=1}\left[h_*^2(X,Y)\right] & =      e^2(X) + 2 e(X)\lambda_X\bar g(X) + \lambda_X^2  \E_{Y\mid X,Z=1}\left[g^2(X,Y)\right]            \\
 & =    2 e(X) -e^2(X)  + [1-e(X)]^2 ( \xi_X  -  \Delta_X )/\bar g(X).
 \end{split}
\end{equation}
Using \eqref{equ:derivative_kappa_0} and \eqref{equ:rearrange} we can write the derivative of \eqref{equ:step_3} as
\begin{align*}
         &\  [1-e(X)]^2 \frac{d}{d\epsilon}\left[  (\xi_{X,\epsilon} - \Delta_{X,\epsilon})/\bar g_{\epsilon} (X)                       \right]\big|_{\epsilon=0}       \\
= &\  \lambda_X^2 \left\{ \left( \frac{d \xi_{X,\epsilon}}{d\epsilon} \Big|_{\epsilon=0}   - \frac{d \Delta_{X,\epsilon}}{d\epsilon} \Big|_{\epsilon=0}  \right)\bar  g (X) -  (\xi_X-\Delta_X)\frac{d \bar g_{\epsilon} (X) }{d\epsilon}  \Big|_{\epsilon=0}           \right\}                                \\
= &\ - \lambda_X^2 \Bigg\{   \E_{Y\mid X,Z=1}\left[( Y-\Delta_X)\one_{\{Y\leq\xi_X\} }\right] \frac{d \xi_{X,\epsilon}}{d\epsilon} \Big|_{\epsilon=0} +  \bar g (X)  \frac{d \Delta_{X,\epsilon}}{d\epsilon} \Big|_{\epsilon=0} \\
&\ \hspace{58pt} + (\xi_X-\Delta_X)     \E\left[        \frac{Z}{e(X)}\left[g  (X,Y)-  \bar g (X) \right]       S(O)\   \Big|\  X       \right]         \Bigg\}           \\
= &\ \E\bigg\{  \bigg[ \frac{Z}{e(X)}           \Big( - 2\lambda_X [1-e(X)] \left(Y - \E_{Y\mid X,Z=1 }[Y]        \right)  -       \lambda_X^2 g^2(X,Y)  \\
&\ \hspace{78pt} + \lambda_X^2 \E_{Y\mid X,Z=1}\left[g^2(X,Y)\right] \Big)+ 2\theta[Z-e(X)]/\bar g(X)           \bigg]S(O)              \       \Big|\  X \bigg\}.
\end{align*}
The derivative of \eqref{equ:reformulation} is given by
\begin{align*}
&\ \frac{d}{d\epsilon}\left[  2e_{\epsilon}(X) - e_{\epsilon}^2 (X)             \right]\big|_{\epsilon=0}       + \frac{d}{d\epsilon}\left[  1-e_{\epsilon}(X)  \right]^2 \big|_{\epsilon=0}( \xi_X  -  \Delta_X )/\bar g(X)    \\
= &\ \E\left\{   2[Z-e(X)][1-e(X)-      \lambda_X( \xi_X  -  \Delta_X )] S(O)           \       \Big|\  X \right\}.
\end{align*}
Combing these derivatives and applying \Cref{lemma:eif_mean} to the parameter
$\psi_3 = \E_{X}\left[		\nu_{1,h^*}(X)				\right]=                 \E_{X}\left[\E_{Y\mid X, Z=1}[h_*^2(X,Y)]\right]$,
we can write the derivative $\frac{d\psi_{3,\epsilon}}{d\epsilon}\Big|_{\epsilon=0}$ as
\begin{align*}
 & \E_X\left[  \frac{d}{d\epsilon} \E_{Y\mid X, Z=1,\epsilon}[h_{*,\epsilon}^2(X,Y)]\big|_{\epsilon=0} \right]  + \frac{d}{d\epsilon} \E_{X,\epsilon}\left[ \nu_{1,h^*}(X) \right]\big|_{\epsilon=0}              \\
  = &\ \E\bigg\{  \bigg[   \frac{Z}{e(X)}         \Big( - 2\lambda_X [1-e(X)] \left(Y - \E_{Y\mid X,Z=1 }[Y]        \right)  -       \lambda_X^2 g^2(X,Y) + \lambda_X^2 \E_{Y\mid X,Z=1}\left[g^2(X,Y)\right] \Big) \\
 &\ \hspace{13pt} + 2\lambda_X[Z-e(X)] \E_{Y\mid X,Z=1}\left[(Y-\xi_X) \one_{\{Y > \xi_X \}}\right]                    + \E_{Y\mid X, Z=1}[h_*^2(X,Y)]-\psi_3   \bigg]S(O)      \bigg\},
\end{align*}
which shows that  $EIF(\psi_3)=\phi_3(O)-\psi_3$ satisfies the definition in \eqref{equ:if_definition}.
\end{proof}

\section{Proof of Proposition \ref{thm:double_bias}}\label{appendix:double_bias}
\subsection{Proof for \texorpdfstring{$\hat \phi_1$}{phi1}}
\begin{proof}
Evaluated at $\xi_X$ and $\hat  \xi_X$, $\hat f_{\lambda,X}(\hat  \xi_X )=f_{\lambda,X}(\xi_X )$. Under Assumption \ref{assumption:nuisance_error},
\begin{align*}
\hat {\E}_{Y\mid X,Z=1} \left[(\hat \xi_X-Y)      \one_{\{Y\leq  \hat  \xi_X \}}  \right] -       \E_{Y \mid X,Z=1} \left[(\xi_X-Y) \one_{\{Y\leq  \xi_X \}}        \right]   = &\frac{e(X)-\hat e (X)}{\lambda} \\
 (\hat \xi_X - \xi_X  )\E_{Y\mid X,Z=1} \big[     \one_{\{Y\leq  \hat   \xi_X \}} \mid \hat \eta \big]-  \E_{Y\mid X,Z=1}  \big[    (\xi_X-Y) (\one_{\{Y\leq   \xi_X \}} -  \one_{\{  Y\leq    \hat \xi_X \}})      \mid \hat \eta \big]            &= o_{\PP}(n^{-1/4})            \\
 O_{\PP} (\xi_X -\hat \xi_X) +     O_{\PP} \left( \P_{Y\mid X,Z=1}  \left\{  \xi_X  \wedge  \hat \xi_X \leq Y\leq  \xi_X  \vee  \hat \xi_X \mid \hat \eta     \right\}                        \right)         & =              o_{\PP}\big(n^{-1/4 }\big),
\end{align*}
which shows that  $\hat \xi_X-\xi_X =  o_{\PP}\big(n^{-1/4 }\big)$.
By the definition of $h_*(X,Y)$ in \Cref{prop:double_ipw},
\begin{equation}\label{equ:two_equality}
\begin{split}
&\E_{Y\mid X,Z=1}[h_*^2(X,Y)] = e(X)[2-e(X)] + \lambda^2\E_{Y\mid X,Z=1}[g^2(X,Y)], \\
&\mathbb{ \hat  E}_{Y\mid X,Z=1} \big[\hat h_*^2(X,Y)   \big]   =  \hat e(X) \left[       2- \hat e(X)            \right]+\lambda^2 \mathbb{\hat  E}_{Y\mid X,Z=1} [\hat g^2(X,Y)   ].
\end{split}
\end{equation}
The bias $\Bias_{Y,Z\mid X}(\hat \phi_1 \mid \hat \eta ) = \mathbb E_{Y,Z\mid X}\big[\hat \phi_1(X,Y,Z) - h_*^2(X,Y) \ \big| \ \hat \eta \big]$ can be written as
\begin{equation}\label{equ:thm_4_expression}
\begin{split}
& \left\{ e(X)/\hat e(X) -1 \right\} \big\{ 2\left[1-\hat e(X) \right]\E_{Y\mid X,Z=1}\left[\hat \Pi_h (X,Y)\right] +   \lambda^2\mathbb E_{Y\mid X,Z=1} \left[\hat g^2(X,Y) \mid \hat \eta \right] \\
&  -     \lambda^2 \hat{\mathbb E}_{Y\mid X,Z=1}\left[\hat g^2(X,Y)   \right]       \big\} + \mathbb  E_{Y,Z\mid X}\Big\{   2[1-\hat e(X)][Z-\hat e(X)]    + \lambda^2 \hat g^2(X,Y)  \\
& \hspace{5pt} -  \lambda^2 \hat{\mathbb E}_{Y\mid X,Z=1}\left[\hat g^2(X,Y)   \right] +  \hat{\mathbb E}_{Y\mid X,Z=1}\big[\hat h_*^2(X,Y)\big]  -h_*^2(X,Y)
 \  \Big|\ \hat \eta\Big\}      \\
&  +  2[1-\hat e(X)]\mathbb  E_{Y,Z\mid X}\big\{ \hat \Pi_h(X,Y)+ \hat \Pi_e(X,Y)
\mid \hat \eta\big\}.
\end{split}
\end{equation}
By $\mathbb {E}_{Y\mid X,Z=1}[\Pi_h(X,Y)] = 0$ and Assumption \ref{assumption:nuisance_error}, the first product term in
\eqref{equ:thm_4_expression},
\begin{align*}
  \left\{ e(X)/\hat e(X) -1 \right\} \big\{ & 2\left[1-\hat e(X) \right]\E_{Y\mid X,Z=1}\left[\hat \Pi_h (X,Y)\right]  \\
  &  +   \lambda^2\mathbb E_{Y\mid X,Z=1} \left[\hat g^2(X,Y) \mid \hat \eta \right] -     \lambda^2 \hat{\mathbb E}_{Y\mid X,Z=1}\left[\hat g^2(X,Y)   \right]       \big\}  = o_{\PP}(n^{-1/2}).
\end{align*}
Using \eqref{equ:two_equality}, we can write the rest of the second and third lines of \eqref{equ:thm_4_expression} as
\begin{align*}
 & 2e(X)[1-\hat e(X)]^2 -2[1-e(X)]\hat e(X)[1-\hat e(X)] + \lambda^2 \E_{Y\mid X,Z=1}[\hat g^2(X,Y) \  |\ \hat \eta] + \hat e(X) \left[   2- \hat e(X)            \right]  \\
 &  \hspace{216pt} - e(X)[2-e(X)] - \lambda^2\mathbb{  E}_{Y\mid X,Z=1} \left[ g^2(X,Y)   \right]  \\
= &\  [\hat e(X)-e(X)]^2 +  \lambda^2 \E_{Y\mid X,Z=1}[\hat g^2(X,Y) - g^2(X,Y) \  |\ \hat \eta] \\
= &\ \lambda^2 \E_{Y\mid X,Z=1}[\hat g^2(X,Y) - g^2(X,Y) \  |\ \hat \eta] + o_{\PP}\big(n^{-1/2 }\big).
\end{align*}
By the definition  of $h_*(X,Y)$ and $\hat h_*(X,Y)$, we can rewrite the last line of \eqref{equ:thm_4_expression} as\begin{align*}
&\  2[1-\hat e(X)]\mathbb  E_{Y,Z\mid X}\big\{ \hat \Pi_h(X,Y)+ \hat \Pi_e(X,Y)
\mid \hat \eta\big\}  \\
        =        &\ \frac{2\lambda \mathbb{\hat E}_{Y\mid  X,Z=1} \left[ \hat g(X,Y) \right] }{ \mathbb{\hat E}_{Y\mid  X,Z=1} \big[ \one_{ \{Y\leq \hat \xi_X\} } \big]}
                \Big\{ \mathbb{ E}_{Y\mid X,Z=1}\big[ h_*(X,Y)- \hat h_*(X,Y) \mid \hat \eta         \big] + e(X) [\hat e(X)-1]  \\
 &\ \hspace{270pt} + [1-e(X)]\hat e(X) \Big\} \\
= &\ 2\lambda^2 \mathbb{\hat E}_{Y\mid  X,Z=1, Y\leq \hat \xi_X} \left[ \hat g(X,Y) \right] \mathbb{ E}_{Y\mid X,Z=1}[ g(X,Y)- \hat g(X,Y) \mid \hat \eta   \big].
\end{align*}
Combining the remaining terms in the last two equations, we have
\begin{align*}
 &\ \lambda^2 \E_{Y\mid X,Z=1}[\hat g^2(X,Y) - g^2(X,Y) \  |\ \hat \eta] + 2\lambda^2 \mathbb{\hat E}_{Y\mid  X,Z=1, Y\leq \hat \xi_X} \left[ \hat g(X,Y) \right]
 \mathbb{ E}_{Y\mid X,Z=1}[g(X,Y)  \\
 & \hspace{340pt}  - \hat g(X,Y)  \mid \hat \eta    \big]  \\
=&\ \lambda^2 \E_{Y\mid X,Z=1}\bigg\{      \left[\hat g(X,Y) + g(X,Y) - 2 \mathbb{\hat E}_{Y\mid  X,Z=1 , Y\leq \hat \xi_X} \left[ \hat g(X,Y) \right] \right]\Big[         \hat g(X,Y) - g(X,Y)             \Big]   \  \Big|\ \hat \eta\bigg\}\\
\leq &\  \lambda^2 \E_{Y\mid X,Z=1}\bigg\{         \left[\hat g(X,Y) + g(X,Y) - 2 \mathbb{\hat E}_{Y\mid  X,Z=1 , Y\leq \hat \xi_X} \left[ \hat g(X,Y) \right] \right] \times 0 \times  \one_{\left\{ Y>  \xi_X\vee \hat \xi_X\right\} }     \  \Big|\ \hat \eta\bigg\}\\
& + \lambda^2 \E_{Y\mid X,Z=1}\bigg\{      \left[\hat \xi_X-Y +\xi_X -Y- 2\hat \xi_X + 2 \mathbb{\hat E}_{Y\mid  X,Z=1 , Y\leq \hat \xi_X} \left[ Y \right] \right]  \big[\hat \xi_X  - \xi_X\big]  \one_{\left\{Y\leq  \xi_X\wedge  \hat \xi_X\right\} }       \  \Big|\ \hat \eta\bigg\}     \\
& + \lambda^2 \E_{Y\mid X,Z=1}\bigg\{      \left[\hat g(X,Y) + g(X,Y) - 2 \mathbb{\hat E}_{Y\mid  X,Z=1 , Y\leq \hat \xi_X} \left[ \hat g(X,Y) \right] \right] \times \big|\hat \xi_X  - \xi_X\big| \\
&\hspace{290pt} \times  \one_{\left\{\xi_X\wedge\hat \xi_X< Y\leq  \xi_X\vee \hat \xi_X\right\} }        \  \Big|\ \hat \eta\bigg\}             \\
 \leq &\ 2\lambda^2 \big\{\hat \xi_X-\xi_X\big\} \big\{\mathbb{\hat E}_{Y\mid  X,Z=1 , Y\leq \hat \xi_X} \left[ Y \right] - \E_{Y\mid X,Z=1,Y\leq  \xi_X\wedge  \hat \xi_X}\left[ Y      \ \big| \  \hat \eta\right]    \big\}+O_{\PP}\big(\big[\hat \xi_X -\xi_X\big]^2\big)   \\
= &\     o_{\PP}(n^{-1/2}).
\end{align*}
The last inequality is attained by
$\PP_{Y\mid X,Z=1}\big\{\xi_X\wedge \hat \xi_X < Y\leq  \xi_X\vee \hat \xi_X  \mid \hat \eta \big\}    = O_{\PP}(|\hat \xi_X - \xi_X|).$
The last equality is obtained by $\hat \xi_X-\xi_X =  o_{\PP}\big(n^{-1/4 }\big)$ above and \eqref{equ:thm_4_proof} below.
Based on \eqref{equ:thm_4_expression} and the equations above, we have an upper bound for the bias of $\hat \phi_1:$
\begin{equation} \label{equ:bias_bound}
     \Bias(\hat \phi_1 \mid \hat \eta)
\lesssim        \big(  \| \hat e - e\| +  \| \hat \xi_X - \xi_X\| \big)  \| \hat p_{Y\mid X,Z=1}   - p_{Y\mid X,Z=1} \| + \| \hat e - e\|^2 + \|\hat \xi_X - \xi_X\|^2.
\end{equation}
\end{proof}

\subsection{Proof for \texorpdfstring{$\hat \phi_2$}{phi2}}
\begin{proof}
We write $\Bias_{Y,Z\mid X}(\hat \phi_2 \mid \hat \eta ) = \mathbb E_{Y,Z\mid X}\Big\{\hat \phi_2(X,Y,Z) - h_*(X,Y)Y \ \Big| \ \hat \eta \Big\}$ as
\begin{equation}\label{equ:thm_4_expression_2}
\begin{split}
& \left\{\frac{  e(X) }{\hat e(X)} -1 \right\} \Big\{ \hat \Pi_h (X,Y)\mathbb{\hat E}\big[Y\one_{\{Y\leq \hat \xi_X\}}\big] +  \mathbb E_{Y\mid X,Z=1} \left[\hat h_*(X,Y)Y \mid \hat \eta \right] -      \hat{\mathbb E}_{Y\mid X,Z=1}\left[\hat h_*(X,Y)Y   \right]        \Big\} \\
&  + \mathbb  E_{Y,Z\mid X}\Big\{   [Z-\hat e(X)] \mathbb{\hat E}_{Y\mid X,Z=1}\big[Y\big]
     +  \hat h_*(X,Y)Y  -h_*(X,Y)Y
 \  \Big|\ \hat \eta\Big\}      \\
&  + \mathbb  E_{Y,Z\mid X}\big\{ \hat \Pi_h(X,Y)+ \hat \Pi_e(X,Y)
\mid \hat \eta\big\}  \mathbb{\hat E}\big[Y\one_{\{Y\leq \hat \xi_X\}}\big].
\end{split}
\end{equation}
The first line of \eqref{equ:thm_4_expression_2} is $o_{\PP}(n^{-1/2}).$
The second line of \eqref{equ:thm_4_expression_2} can be written as
\begin{align*}
&\ \big\{\hat e(X) - e(X)\big\}\big\{\E_{Y\mid X,Z=1}[Y] -\mathbb{\hat E}_{Y\mid X,Z=1}[Y] \big\} + \lambda \E_{Y\mid X,Z=1} \left[\hat g(X,Y)Y- g(X,Y)Y\mid \hat \eta \right] \\
= &\  \lambda \E_{Y\mid X,Z=1} \left\{ [\hat g(X,Y)- g(X,Y)]Y\mid \hat \eta \right\}+ o_{\PP}(n^{-1/2}).
\end{align*}
We can rewrite the third line of \eqref{equ:thm_4_expression_2} as
$\lambda \mathbb{ E}_{Y\mid X,Z=1}[ g(X,Y)- \hat g(X,Y) \mid \hat \eta    \big]\mathbb{\hat E}_{Y\mid X,Z=1,Y\leq \hat \xi_X}\big[Y\big].$
Combine the remaining terms in the last two equations,
\begin{align*}
&\      \lambda \E_{Y\mid X,Z=1} \left\{   \Big[\hat g(X,Y)- g(X,Y)\Big] \left[Y - \mathbb{\hat E}_{Y\mid X,Z=1,Y\leq \hat \xi_X}\big[Y\big]  \right] \  \big|\ \hat \eta \right\}                                          \\
 \leq &\        \lambda  \E_{Y\mid X,Z=1}\left\{          0\times \left[Y  -  \hat{ \mathbb E}_{Y\mid X,Z=1,Y\leq \hat \xi_X  }  \big[Y  \big] \right]      \times \one_{\left\{Y> \xi_X \vee \hat \xi_X\right\} }\ \big|\ \hat \eta  \right\}                                                      \\
 & + \lambda  \E_{Y\mid X,Z=1}\left\{             \left[\hat \xi_X-\xi_X\right]\times \left[Y  -  \hat{ \mathbb E}_{Y\mid X,Z=1,Y\leq \hat \xi_X  }  \big[Y  \big] \right] \times   \one_{\left\{Y\leq  \xi_X\wedge \hat \xi_X\right\} } \ \big|\ \hat \eta  \right\}       \\
  & +  \lambda  \E_{Y\mid X,Z=1}\left\{           \big|\hat \xi_X -\xi_X  \big|\times \big|Y  -  \hat{ \mathbb E}_{Y\mid X,Z=1,Y\leq\hat \xi_X  }  \big[Y  \big] \big|      \times \one_{\left\{\xi_X\wedge \hat \xi_X< Y\leq  \xi_X\vee \hat \xi_X\right\} } \ \big|\ \hat \eta  \right\}  \\
\leq &\  \lambda  \big|\hat \xi_X -\xi_X\big| \big|                \mathbb E_{Y\mid X,Z=1,Y\leq \xi_X\wedge \hat \xi_X } \big[Y \mid  \hat \eta  \big] -    \hat{ \mathbb E}_{Y\mid X,Z=1,Y\leq\hat \xi_X  }  \big[Y   \big]                                  \big| + O_{\PP}\big(\big[\hat \xi_X -\xi_X\big]^2\big) = o_{\PP}\big(n^{-1/2}\big).
\end{align*}
The last line follows the same argument above \eqref{equ:bias_bound}.
Based on \eqref{equ:thm_4_expression_2} and the equations above, we arrive at an upper bound for the bias of $\hat \phi_2:$
\begin{equation} \label{equ:bias_bound_2}
         \Bias(\hat \phi_2 \mid \hat \eta)
\lesssim          \big(  \| \hat e - e\| +  \| \hat \xi_X - \xi_X\| \big)  \| \hat p_{Y\mid X,Z=1}   - p_{Y\mid X,Z=1} \| + \|\hat \xi_X - \xi_X\|^2.       \end{equation}
We now prove the remaining claim
\begin{equation}\label{equ:thm_4_proof}
\        \big| \mathbb E_{Y\mid X,Z=1,Y\leq \xi_X\wedge \hat\xi_X } \big[ Y  \mid  \hat \eta \big] - \hat{\E}_{Y\mid X,Z=1,Y\leq\hat \xi_X  } \big[Y  \big]      \big| = o_{\PP}(n^{-1/4}).
\end{equation}
We will show that the left-hand side of \eqref{equ:thm_4_proof} is upper bounded by
\begin{align*}
&\  \big| \mathbb E_{Y\mid X,Z=1,Y\leq\xi_X} \big[ Y  \big] - \hat{\E}_{Y\mid X,Z=1,Y\leq\hat \xi_X  } \big[Y  \big]      \big|       + \big| \mathbb E_{Y\mid X,Z=1,Y\leq\hat\xi_X } \big[ Y   \mid  \hat \eta \big] - \hat{\E}_{Y\mid X,Z=1,Y\leq\hat \xi_X  } \big[Y  \big]      \big|\\
\leq &\           O_{\PP}\big( \big|\hat \xi_X -\xi_X\big|\big) +   O_{\PP}\Big(\int  |\hat p_{Y\mid X,Z=1}(y) -p_{Y\mid X} (y)|dy\Big)         \equiv  E_{X,Z=1} + E_{X,2}.
\end{align*}
We rewrite the first term,
\begin{align*}
                &\      \mathbb E_{Y\mid X,Z=1,Y\leq\xi_X } \big[Y  \big] -       \hat{ \mathbb E}_{Y\mid X,Z=1,Y\leq\hat \xi_X  }  \big[Y  \big]   \\
        = &\      \frac{         \hat{\mathbb E}_{Y\mid X,Z=1} \big[\one_{\{Y\leq \hat \xi_X\}}  \big]                    \mathbb E_{Y\mid X,Z=1 } \big[Y \one_{\{Y\leq\xi_X\}} \big]                       -       \mathbb E_{Y\mid X,Z=1} \big[\one_{\{Y\leq\xi_X\}}  \big]                 \hat{\mathbb E}_{Y\mid X,Z=1 } \big[Y \one_{\{Y\leq\hat \xi_X\}} \big]                                                                                    }{                      \mathbb E_{Y\mid X,Z=1} \big[\one_{\{Y\leq\xi_X\}}  \big]                  \hat{\mathbb E}_{Y\mid X,Z=1} \big[\one_{\{Y\leq \hat \xi_X\}}  \big]            }       \\
                = &\      \frac{         \mathbb E_{Y\mid X,Z=1} \big[\one_{\{Y\leq \hat \xi_X\}}  \big]                  \mathbb E_{Y\mid X,Z=1 } \big[Y \one_{\{Y\leq\xi_X\}} \big]                       -       \mathbb E_{Y\mid X,Z=1} \big[\one_{\{Y\leq\xi_X\}}  \big]                 \mathbb E_{Y\mid X,Z=1 } \big[Y \one_{\{Y\leq \hat \xi_X \}} \big]                                        +       E_{X,2}                 }{                      \mathbb E_{Y\mid X,Z=1} \big[\one_{\{Y\leq\xi_X\}}  \big]                  \hat{\mathbb E}_{Y\mid X,Z=1} \big[\one_{\{Y\leq \hat \xi_X\}}  \big]            }       \\
                = &\    O_{\PP}\left(  \mathbb E_{Y\mid X,Z=1} \big[\one_{\{Y\leq \hat \xi_X\}  } -  \one_{\{Y\leq \xi_X\}  }   \big]               \right)         +       O_{\PP}\left(             \mathbb E_{Y\mid X,Z=1 } \big[Y \one_{\{Y\leq \xi_X \}} - Y \one_{\{Y\leq \hat \xi_X \}} \big]            \right)         + E_{X,2}       \\
                =&\ E_{X,Z=1} +E_{X,2}.
\end{align*}
Similarly, we can show that $\mathbb E_{Y\mid X,Z=1,Y\leq \hat \xi_X } \big[Y  \mid  \hat \eta \big]- \hat{\mathbb E}_{Y\mid X,Z=1,Y\leq \hat \xi_X } \big[Y  \big] =E_{X,Z=1}+E_{X,2}$. Finally, under Assumption \ref{assumption:nuisance_error}, $E_{X,Z=1}+E_{X,2} = o_{\PP }(n^{-1/4}),$ which completes the proof.
\end{proof}

\subsection{Proof for \texorpdfstring{$\hat \phi_3$}{phi3}}

\begin{proof}
We start by showing that $|\hat \xi_X - \xi_X|  = o_{\PP}(n^{-1/4}).$
First, we have
\begin{align*}
&\ \E_{Y\mid X,Z=1}\big[\hat g(X,Y)Y-     g(X,Y)Y         \mid \hat \eta  \big]  \\
= &\  \E_{Y\mid X,Z=1}\big[(\hat \xi_X-Y)Y \one_{\{Y \leq \hat \xi_X \}}  -       (\xi_X-Y)Y \one_{\{Y \leq \xi_X \}}             \mid \hat \eta  \big] \\
=       &\ \E_{Y\mid X,Z=1}\big[(\hat \xi_X-Y)Y \one_{\{Y \leq \hat \xi_X \}}- (\xi_X-Y)Y \one_{\{Y \leq \hat \xi_X \}} +(\xi_X-Y)Y \one_{\{Y \leq \hat \xi_X \}}     \\
 &\ -       (\xi_X-Y)Y \one_{\{Y \leq \xi_X \}}                     \mid \hat \eta\big]\\
=&\ O_{\PP}(\hat \xi_X - \xi_X) + O_{\PP}(\PP_{Y\mid X,Z=1}\{   \xi_X\wedge\hat \xi_X \leq Y\leq  \xi_X \vee\hat \xi_X  \mid \hat \eta \})      = O_{\PP}(|\hat \xi_X - \xi_X|).
\end{align*}
Similarly, we can show that
$\E_{Y\mid X,Z=1}\big[\hat g(X,Y)-        g(X,Y)  \mid \hat \eta  \big]  =  O_{\PP}(|\hat \xi_X - \xi_X|).$
Evaluated at the roots $\xi_X$ and $\hat \xi_X,$  $\hat f_{\theta,X}(\hat \xi_X) =  f_{\theta,X}(\xi_X)$. Under Assumption \ref{assumption:nuisance_error}, this implies that
 \begin{align*}
&\hspace{12.5pt} \E_{Y\mid X,Z=1}\left[g(X,Y)\right] \hat{\E}_{Y\mid X,Z=1}\big[\hat g(X,Y)Y\big] \\
& \hspace{145pt}  - \hat{\E}_{Y\mid X,Z=1}\big[\hat g(X,Y)\big]\E_{Y\mid X,Z=1}\left[g(X,Y)Y\right]   = o_{\PP}(n^{-1/4})        \\
& \Rightarrow \E_{Y\mid X,Z=1}\left[g(X,Y)\right]
 \E_{Y\mid X,Z=1}\big[\hat g(X,Y)Y \mid \hat \eta \big]
  \\
& \hspace{145pt} - \E_{Y\mid X,Z=1}\big[\hat g(X,Y) \mid \hat \eta \big]\E_{Y\mid X,Z=1}\left[g(X,Y)Y\right]  = o_{\PP}(n^{-1/4})      \\
&  \Rightarrow \E_{Y\mid X,Z=1}\left[g(X,Y)\right] \E_{Y\mid X,Z=1}\big[\hat g(X,Y)Y -g(X,Y)Y           \mid \hat \eta  \big]  \\
&  \hspace{110pt}   + \E_{Y\mid X,Z=1}\big[g(X,Y) - \hat g(X,Y)\mid \hat \eta\big]\E_{Y\mid X,Z=1}\left[g(X,Y)Y\right]    = o_{\PP}(n^{-1/4}).
\end{align*}
It follows from the last two equations that
\begin{align*}
         \hat \lambda_{X} -\lambda_{X} & = [1-\hat e(X)]/  \hat{ \E}_{Y\mid X,Z=1}\big[ ( \hat  \xi_X -Y ) \one_{\{Y \leq \hat \xi_X\}}  \big]-[1-e(X)]/  \E_{Y\mid X,Z=1}\big[ ( \xi_X -Y ) \one_{\{Y \leq  \xi_X\}}  \big] \\
                & =             O_{\PP}( |\hat \xi_X - \xi_X| + |\hat e(X) -e(X)|         )               = o_{\PP}(n^{-1/4}).
\end{align*}
Next, we rewrite $\Bias_{Y,Z\mid X}(\hat \phi_3 \mid \hat \eta ) = \mathbb E_{Y,Z\mid X}\Big\{\hat \phi_3(X,Y,Z) - h_*^2(X,Y) \ \Big| \ \hat \eta \Big\}$ as
\begin{equation}\label{equ:bias_6}
\begin{split}
        &       \left\{\frac{e(X)}{\hat e(X)}-1 \right\}\bigg\{ \underbrace{-2 \hat \lambda_X[1-\hat e(X)]\big(\E_{Y\mid X,Z=1}[Y] - \mathbb{\hat E}_{Y\mid X,Z=1}[Y]\big)}_{:=A} \\
        & \hspace{150pt} + \underbrace{- \hat \lambda_X^2\mathbb{ E}_{Y\mid X,Z=1}[ \hat g^2(X,Y)\mid \hat \eta ] + \hat \lambda_X^2  \mathbb{\hat E}_{Y\mid X,Z=1}[\hat g^2(X,Y)]}_{:=B} \bigg\}      \\
&
\underbrace{-2 \hat \lambda_X[1-\hat e(X)] \big(\mathbb E_{Y\mid X,Z=1}[Y] - \mathbb{\hat E}_{Y\mid X,Z=1}[Y] \big)}_{A}  \\
& \hspace{160pt} +  \underbrace{    -  \hat \lambda_X^2\mathbb{ E}_{Y\mid X,Z=1}[ \hat g^2(X,Y)\mid \hat \eta ] + \hat \lambda_X^2  \mathbb{\hat E}_{Y\mid X,Z=1}[\hat g^2(X,Y)]}_{B} \\\
& + \underbrace{2 \hat \lambda_X[e(X)-\hat e(X)] \mathbb{\hat E}_{Y\mid X,Z=1}[(Y-\hat \xi_X)\one_{\{Y>\hat \xi_X \}}]}_{:=C}     + \underbrace{\mathbb{\hat E}_{Y\mid X,Z=1}[\hat h_*^2(X,Y)] -  \mathbb E_{Y\mid X,Z=1}[ h_*^2(X,Y)] }_{:=D}.
\end{split}
\end{equation}
Under Assumption \ref{assumption:nuisance_error}, the first line of \eqref{equ:bias_6} with $A+B$ is $o_{\PP}(n^{-1/2})$.
By the definition of $h_*(X,Y)$ and $\lambda_X$ in \Cref{prop:solutions_new_1},  we have
\begin{align*}
&\E_{Y\mid X,Z=1}[h_*^2(X,Y)] = e(X)\left[        2- e(X)                 \right]  +  \lambda_X^2\E_{Y\mid X,Z=1}[g^2(X,Y)], \\
&\mathbb{ \hat  E}_{Y\mid X,Z=1} \big[\hat h_*^2(X,Y)   \big]   =  \hat e(X) \left[       2- \hat e(X)            \right]+\hat \lambda_X^2 \mathbb{\hat  E}_{Y\mid X,Z=1} [\hat g^2(X,Y)   ].
\end{align*}
Then we rewrite $D$ as
\begin{align*}
D = & \hat \lambda_X^2\mathbb{\hat E}_{Y\mid X,Z=1}[\hat g^2(X,Y)] - \lambda_X^2\ \mathbb{ E}_{Y\mid X,Z=1}[g^2(X,Y)] + [\hat e(X)-e(X)] \left[     2- \hat e(X) -e(X)              \right]         \\
 = & \underbrace{ \hat \lambda_X^2\mathbb{\hat E}_{Y\mid X,Z=1}[\hat g^2(X,Y)] - \lambda_X^2\ \mathbb{ E}_{Y\mid X,Z=1}[g^2(X,Y)]  }_{:= D_1}+ \underbrace{2[ \hat e(X) - e(X)][1- e(X)] }_{:= D_2} + o_{\PP}(n^{-1/2}).
\end{align*}
Next, we can decompose
\[
B+D_1 = 2 \hat \lambda_X^2  \mathbb{\hat E}_{Y\mid X,Z=1}[\hat g^2(X,Y)] -  \lambda_X^2 \E_{Y\mid X,Z=1}[g^2(X,Y)]         - \hat \lambda_{X}^2            \E_{Y\mid X,Z=1}[\hat g^2(X,Y)\mid \hat \eta ],
\]
into two terms, $(B+D_1)_1 $ defined as
\[
2 \hat \lambda_X^2  \mathbb{\hat E}_{Y\mid X,Z=1}[\hat g^2(X,Y)] - (\lambda_X^2+ \hat \lambda_X^2 ) \mathbb{ E}_{Y\mid X,Z=1}[ g^2(X,Y) ] - 2\hat \lambda_X [1-\hat e(X)] [\hat \xi_X-\xi_X], \\
\]
and
$(B+D_1)_2 = \hat \lambda_X^2  \mathbb{ E}_{Y\mid X,Z=1}[  g^2(X,Y) - \hat g^2(X,Y) \mid \hat \eta ] +2\hat \lambda_X [1-\hat e(X)] [\hat \xi_X-\xi_X].$

We first show that $(B+D_1)_2= o_{\PP}(n^{-1/2}).$ In the first term of $(B+D_1)_2$,
\begin{align*}
&\ \mathbb{ E}_{Y\mid X,Z=1}[   g^2(X,Y) - \hat g^2(X,Y) \mid \hat \eta ]  =       \mathbb E_{Y\mid X,Z=1}\big[      (\xi_X-Y)^2 \one_{\{Y \leq \xi_X \}}-(\hat \xi_X-Y)^2 \one_{\{Y \leq \hat \xi_X \}}    \mid \hat \eta \big]    \\
  \leq  &\      \mathbb E_{Y\mid X,Z=1}\big[      0\times \one_{\{Y > \xi_X\vee \hat \xi_X \}} \mid \hat \eta \big]       + \mathbb E_{Y\mid X,Z=1}\big[    (\hat \xi_X -\xi_X )^2 \times \one_{\{\xi_X\wedge\hat \xi_X     < Y \leq   \xi_X\vee \hat \xi_X \}} \mid \hat \eta \big]                \\
                & +\mathbb E_{Y\mid X,Z=1}\big[(\xi_X + \hat \xi_X -2Y ) (\xi_X  -\hat \xi_X ) \times \one_{\{Y \leq \xi_X\wedge \hat \xi_X \}} \mid \hat \eta \big] \\
\leq  & - (\hat \xi_X - \xi_X ) \mathbb E_{Y\mid X,Z=1}\big[(\xi_X -Y ) \one_{\{Y \leq   \xi_X\wedge \hat \xi_X \}} + ( \hat  \xi_X -Y ) \one_{\{Y \leq   \xi_X \wedge \hat \xi_X \}} \mid \hat \eta \big]  +O_{\PP}\big(\big[\hat \xi_X -\xi_X\big]^2\big).
\end{align*}
Then,  $(B+D_1)_2$ can be upper bounded by
\begin{align*}
&\ \hat \lambda_X^2  (\hat \xi_X-\xi_X)\Big(            -       \mathbb E_{Y\mid X,Z=1}\big[(\xi_X -Y ) \one_{\{Y \leq   \xi_X\wedge \hat \xi_X \}} + ( \hat  \xi_X -Y ) \one_{\{Y \leq   \xi_X \wedge \hat \xi_X \}} \mid \hat \eta \big]    \\
&\ \hspace{270pt}          2[1-\hat e(X)]/\hat \lambda_X   \Big) +  o_{\PP}(n^{-1/2})
\\
= &\   2 \hat \lambda_X^2 (\hat \xi_X-\xi_X)\big(       \mathbb{\hat E}\big[( \hat  \xi_X -Y ) \one_{\{Y \leq    \hat \xi_X \}} \big]   -       \mathbb E_{Y\mid X,Z=1}\big[(\xi_X -Y ) \one_{\{Y \leq   \xi_X \}}         \big]  \big)  = o_{\PP}(n^{-1/2}). 
\end{align*}
We next show the remaining terms $A+(B+D_1)_1 + C + D_2= o_{\PP}(n^{-1/2}).$
Using \eqref{equ:step_3}, we first rewrite $(B+D_1)_1$ above as
\begin{align*}
 &\ 2 \hat \lambda_X^2  \mathbb{\hat E}_{Y\mid X,Z=1}[\hat g^2(X,Y)] - 2\lambda_X^2  \mathbb{ E}_{Y\mid X,Z=1}[ g^2(X,Y) ] +(\lambda_X^2 -\hat\lambda_X^2) \mathbb{ E}_{Y\mid X,Z=1}[ g^2(X,Y) ] \\
 &\ - 2\hat \lambda_X [1-\hat e(X)] [\hat \xi_X-\xi_X] \\
=&\ 2\hat \lambda_X [1-\hat e (X)][\xi_X-\hat \Delta_X]  - 2\lambda_X [1-e (X)][\xi_X-\Delta_X]+(\lambda_X^2 -\hat\lambda_X^2) \mathbb{ E}_{Y\mid X,Z=1}[ g^2(X,Y) ]       \\
=&\ 2\hat \lambda_X [1-\hat e (X)][\Delta_X-\hat \Delta_X ] + 2[\xi_X-\Delta_X] \big(\hat \lambda_X [1-\hat e (X)]-\lambda_X [1-e (X)]\big)   \\
&\ +[\lambda_X^2 -\hat\lambda_X^2]\mathbb{ E}_{Y\mid X,Z=1}[ g^2(X,Y) ] \\
=&\ 2\hat \lambda_X [1-\hat e (X)][\Delta_X-\hat \Delta_X ] + 2[\xi_X-\Delta_X]  \big([1-\hat e (X)][\hat \lambda_X - \lambda_X ] + \lambda_X [e (X)-\hat e(X)]\big)    \\
&\ + [\lambda_X+\hat \lambda_X][\lambda_X - \hat \lambda_X][\xi_X-\Delta_X]\mathbb{ E}_{Y\mid X,Z=1}[ g(X,Y) ]    \\
=&\ 2\hat \lambda_X [1-\hat e (X)][\Delta_X-\hat \Delta_X ]  + 2\lambda_X [\xi_X-\Delta_X] [e (X)-\hat e(X)]    \\
&\ + \big[(\lambda_X+\hat \lambda_X)\mathbb{ E}_{Y\mid X,Z=1}[ g(X,Y) ]+ 2\hat e(X)-2 \big]\big[\lambda_X - \hat \lambda_X\big]\big[\xi_X-\Delta_X\big]   \\
=&\ 2\hat \lambda_X [1-\hat e (X)][\Delta_X-\hat \Delta_X ]  + 2\lambda_X [\xi_X-\Delta_X] [e (X)-\hat e(X)]     \\
& + 2 \big[\mathbb{ E}_{Y\mid X,Z=1}[ h_*(X,Y) ]-1 \big]\big[\lambda_X - \hat \lambda_X\big]\big[\xi_X-\Delta_X\big] +o_{\PP}(n^{-1/2})     \\
 = &\ 2\hat \lambda_X [1-\hat e (X)][\Delta_X-\hat \Delta_X ]  + 2\lambda_X [\xi_X-\Delta_X] [e (X)-\hat e(X)]+o_{\PP}(n^{-1/2}) \\
\equiv &\   (B+D_1)_{1,1} + (B+D_1)_{1,2} + o_{\PP}(n^{-1/2}).
\end{align*}
Recall that $\Delta_X  = \E_{Y\mid X,Z=1 }[Y]  -\theta/[1-e(X)]$
and  $\hat \Delta_X  = \mathbb{\hat E}_{Y\mid X,Z=1 }[Y]  - \theta/[1-\hat e(X)].$
Then,
\begin{align*}
A + (B+D_1)_{1,1} = &\ 2\hat \lambda_X \big[1-\hat e (X)\big]\big[ \mathbb{\hat E}_{Y\mid X,Z=1}[Y]-\hat \Delta_X -\E_{Y\mid X,Z=1}[Y] +\Delta_X \big] \\
=&\ 2\theta\lambda_X \big[\hat e(X) - e(X)\big]/\big[1-e(X)\big] +  o_{\PP}(n^{-1/2}).
\end{align*}
Next, we write $C+ D_2$ as
\begin{align*}
&\ 2\big[\hat e(X)- e(X)\big]\big\{\hat \lambda_X       \mathbb{\hat E}_{Y\mid X,Z=1}\big[(\hat \xi_X - Y)\one_{\{Y>\hat \xi_X \}}\big]+1-e(X)            \big\}\\
= &\ 2\lambda_X \big[\hat e(X)- e(X)\big]\big\{ \E_{Y\mid X,Z=1}\big[(\xi_X - Y)\one_{\{Y> \xi_X \}}\big]+\E_{Y\mid X,Z=1}\big[g(X,Y)\big]          \big\} +        o_{\PP}(n^{-1/2})       \\
= &\ 2 \lambda_X\big[\hat e(X)- e(X)\big]       (\xi_X-\E_{Y\mid X,Z=1}[Y])       + o_{\PP}(n^{-1/2}).
\end{align*}
Combining this expression with $A + (B+D_1)_{1,1}$ above, we have
\begin{align*}
A + (B+D_1)_{1,1} + C+ D_2 & = 2 \lambda_X     \big[\hat e(X)- e(X)\big] \big[ \xi_X - \Delta_X                \big] + o_{\PP}(n^{-1/2}) \\
& = -(B+D_1)_{1,2} + o_{\PP}(n^{-1/2}).
\end{align*}
which implies that  $A+B+C+D =  o_{\PP}(n^{-1/2}).$
\end{proof}

\section{Theory of multiplier bootstrap}\label{sect:theory.mb}

In this section, we present the theory of multiplier bootstrap, which verifies the uniform validity of the confidence bands introduced in \Cref{sect:uni}.
Following the notation introduced in \Cref{sect:notation_if,sect:uni}, we denote the uncentered EIF of $\psi_*$  by $\phi_*$
and its variance $\sigma_{*}^2.$ In one-step estimation, we estimate the EIF by $\hat \phi_*(\cdot) = \hat \phi_*(\cdot; \hat \eta)$, where $\hat \eta$ is
the nuisance estimator, e.g., $\hat
\eta = (\hat e,\hat p_{Y\mid X,Z=1}),$ in the average-case sensitivity model. We implement cross-fitting and denote the cross-fitted estimators similarly. We state the regularity conditions on $\hat \phi_*(\cdot)$ instead of every $\hat \phi_*^{(k)}(\cdot) = \hat \phi_*(\cdot; \hat \eta_{-k} )$ for all $k\in [K]$ in cross-fitting.
The following results apply to all the target parameters (bounds and sensitivity value) introduced in this paper.

\begin{theorem}\label{thm:gaussian_2}
Under Assumptions \ref{assumption:first}, \ref{assumption:nuisance_error} and the conditions
\begin{enumerate}
        \item $\sup_{\beta \in \mathcal{D}_*}|\hat \sigma_{*,\text{cf}}^2(\beta  )/\sigma_*^2(\beta ) - 1 | = o_{\PP}(1),$
        \item $\E\big\{ \big[ \sup_{\beta  \in \mathcal{D}_*}| \hat \phi_*(O;\beta )-\phi_*(O;\beta )\big| \big]^2  \  \big| \   \hat \eta              \big\}= o_{\PP}(1),$
        \item $\sup_{\beta \in \mathcal{D}_*}\Bias(\hat \phi_* \mid \hat \eta)  = o_{\PP}(n^{-1/2}),$
         \item $\hat \phi(o;\beta)$ is a Lipschitz continuous function of
        $\beta$ for any $o\in \O$,
\end{enumerate}
it holds that
$\sqrt{n}\Big[\hat\psi_{*,\text{cf}}(\cdot) - \psi_*(\cdot)\Big]/\hat \sigma_{0,\text{cf}}^{2} (\cdot)\xrightarrow[]{d}\GG_*(\cdot)=\left[\phi_*(O;\cdot) - \psi_*(\cdot)\right]/ \sigma_*(\cdot),$ which is a mean 0 Gaussian process defined in the space of bounded functions on $\mathcal{D}_*.$
\end{theorem}
The first and second conditions require our variance and influence function estimators to be consistent uniformly. The third condition is a stronger version of Propositions \ref{thm:double_bias} and \ref{thm:bias}. It requires the $\beta$-dependent nuisance parameters to be estimated with errors $o_{\PP}(n^{-1/4}), $ e.g., $\sup_{\Gamma \in \mathcal{D}}\|\hat Q(X) - Q(X) \|=o_{\PP}(n^{-1/4})$ and $\sup_{\lambda\in \mathcal{D}_{12}}\|\hat \xi_X-\xi_X\|=o_{\PP}(n^{-1/4}).$  These conditions can be derived from
bias bounds in \Cref{equ:sketch}, \eqref{equ:bias_bound} and \eqref{equ:bias_bound_2}
in \Cref{appendix:double_bias}. We will verify that all our EIFs are Lipschitz in \Cref{proof:thm_4}, which implies that the last condition holds if Assumption \ref{assumption:first}
holds and the estimators of $\beta$-dependent nuisance parameters are
Lipschitz. For example, we need $\hat p_{Y\mid X,Z=1}$ to be continuous so  that $\hat
\xi_X$ behaves smoothly w.r.t. $\lambda\in \mathcal{D}_{12}$ in the Lagrangian formulation of the average-case sensitivity model.

As explained in \Cref{sect:uni}, MB adjusts the critical value in the point-wise CIs to achieve uniform validity. In \Cref{thm:band_2} below, the first equation defines the critical value $\hat q_{*,\alpha}$
as a quantile.  The second equation means the union of the CIs $\hat C_{*,\alpha}(\beta)$ using $\hat q_{*,\alpha}$ gives a simultaneous confidence band.  We refer to \citet[Theorem 4]{kennedy2019nonparametric}
for a proof of \Cref{thm:band_2}.

\begin{theorem}\label{thm:band_2}
        In the setup of \Cref{thm:gaussian_2}, let $\hat q_{*,\alpha}$ denote the critical value of the supremum of the multiplier bootstrap process such that
\[
\P\left\{\sup_{\beta\in \mathcal{D}_*} \left|\sqrt{n}\PP_n \left[ A \left(\hat \phi_{*,\text{cf}}(O;\beta) -  \hat  \psi_{*,\text{cf}}(\beta)          \right)/\hat \sigma_{*,\text{cf}}(\beta) \right]                   \right|         \geq \hat q_{*,\alpha} \        \bigg| \         O_{[n]}  \right\} = \alpha,
\]
where  $A_{[n]}$ are i.i.d Rademacher variables drawn independently of $O_{[n]}$. Then,
\[
\PP \Big\{                \psi_{*} (\beta) \not\in \hat C_{*,\alpha}(\beta)  :=  \Big[\hat\psi_{*,\text{cf}}(\beta)\pm \hat q_{*,\alpha}\hat \sigma_{*,\text{cf}}(\beta)/\sqrt n \Big]    , \forall \hspace{1pt} \beta\in \mathcal{D}_*\Big\}= \alpha+o(1).
\]
\end{theorem}

\subsection{Proof of Theorem \ref{thm:gaussian_2}}\label{proof:thm_4}

The proof of \Cref{thm:gaussian_2} comes from
\citet[Section 8.4]{kennedy2019nonparametric}. We will restate it with the conditions proposed in our theorem. After that, we will verify the Lipschitz continuity of our influence functions.

\begin{proof}
Denote the full-sample empirical process by $\GG_n = \sqrt{n} (\PP_n - \PP).$ Define
\[
\hat    \Omega_n(\beta )  = \sqrt{n}\left[    \hat \psi_{*,\text{cf}}  (\beta ) - \psi_* (\beta )                 \right]/\hat \sigma_{*,\text{cf}}(\beta )
\ \text{ and }  \
        \Omega_n(\beta )  = \GG_n \left\{        \left[\phi_*(O;\beta) -\psi_* (\beta)                     \right]/\sigma_*(\beta)  \right\},
\]
The proof is completed by verifying two statements:
\begin{equation}\label{equ:statements}
\Omega_n(\cdot)\xrightarrow[]{d} \GG(\cdot) \in L^{\infty}(\mathcal{D}) \ \text{ and } \   \sup_{\theta\in \mathcal{D}_*} | \hat \Omega_n(\theta) - \Omega_n(\theta)        | = o_{\PP}(1).
\end{equation}
In \Cref{sect:lc}, we will verify that the function class $\mathcal F = \{\phi_*(O;\beta):\beta \in \mathcal{D}_* \}$ is Lipschitz under Assumptions \ref{assumption:first} and \ref{assumption:consistency}, which proves the first statement. The main argument for the proof is that $\mathcal F$  has a finite bracketing integral so Donsker; see \citet[Chapter 2.5.6]{vaart1996weak} and  \citet[Section 4.3]{kennedy2016semiparametric} for more details.
For the second statement, \citet{kennedy2019nonparametric} shows that for any $\beta \in \mathcal{D}_*,$
\begin{equation}\label{equ:sigma_step}
\begin{split}
\sup_{\beta\in \mathcal{D}_*} |\hat \Omega_n  (\beta ) - \Omega_n (\beta )|
& \lesssim \sup_{\beta\in \mathcal{D}_*}|\tilde \Omega_n  (\beta ) - \Omega_n (\beta )| + \sup_{\beta\in \mathcal{D}_*}|\sigma_*(\beta)/\hat \sigma_{*,\text{cf}}(\beta)-1 | \\
& =  \sup_{\beta\in \mathcal{D}_*}|\tilde \Omega_n  (\beta ) - \Omega_n (\beta )| + o_{\PP}(1),
\end{split}
\end{equation}
where
$\tilde         \Omega_n(\beta )  = \sqrt{n}\left[ \hat        \psi_{*,\text{cf}}  (\beta) - \psi_*  (\beta )                     \right]/\sigma_*(\beta) $, and the equality follows from the first condition in the theorem. Let $\GG_n^{(k)} =\sqrt{m}(\PP_m^{(k)}  - \PP) $. \citet{kennedy2019nonparametric} shows that
\begin{equation}\label{equ:sigma_step_2}
\begin{split}
&\ \tilde \Omega_n  (\beta ) - \Omega_n (\beta ) \\
= &\  \frac{\sqrt{n}
}{K\sigma (\beta)} \sum_{k=1}^{K}\left(\frac{1}{\sqrt{m}} \GG_n^{(k)}\left\{\hat \phi_*^{(k)}(O;\beta) -  \phi_*(O;\beta) \right\}  + \Bias \left\{\hat \phi_*^{(k)}(O;\beta) \ \big| \ \hat \eta_{-k}   \right\} \right) \\
\equiv &\  B_{n,1}(\beta ) + B_{n,2}(\beta ),
\end{split}
\end{equation}
By the proof in \Cref{sect:lc} and the fourth condition, we know that the function class
 $\mathcal{F}_n^{(k)} = \big\{ \hat \phi_*^{(k)}(\cdot;\beta) - \phi_*(\cdot;\beta):\beta \in \mathcal{D}_* \big\}$ is Lipschitz. Then
\begin{equation}\label{equ:proof_b_n_1}
\sup_{\beta\in \mathcal{D}_*}|B_{n,1}(\beta )| \lesssim \max_{k\in [K]}\sup_{k\in \mathcal{F}^{(k)}}|\GG_n(f)| = o_{\PP}(1),
\end{equation}
by the second condition in the theorem. Under the third condition,
\[
        \sup_{\beta \in \mathcal{D}_*}   \Bias\left\{\hat \phi_*(O;\beta  ) \mid \hat \eta \right\}               = o_{\PP}(n^{-1/2}) \Leftrightarrow  \sup_{\beta \in \mathcal{D}_*} \Bias\left\{\hat \phi_*^{(k)}(O;\beta  )\mid \hat \eta_{-k}  \right\}         = o_{\PP}(n^{-1/2}),
\]
which implies that $\sup_{\beta \in \mathcal{D}_*}B_{n,2}(\beta ) = o_{\PP}(1)$.
Taken together with
\eqref{equ:proof_b_n_1}, we verify that \eqref{equ:sigma_step_2} = $o_{\PP}(1). $ Then the second statement in \eqref{equ:statements} holds through \eqref{equ:sigma_step}.
\end{proof}

\subsection{Lipschitz continuity}\label{sect:lc}

Let $\lesssim_{\PP }$ denote smaller than up to some $\PP$-integrable function of $O$. The following definition of
Lipschitz continuity follows from \citet[Chapter 2.7.4]{vaart1996weak} and  \citet[Section 4.3]{kennedy2016semiparametric}.

\begin{lemma}\label{lemma:psi_lip}
Under Assumption \ref{assumption:first}, $\mathcal{F} = \left\{\phi(O;\Gamma): \Gamma\in
\mathcal{D} = [\Gamma_{\min }, \Gamma_{\max }]  \right\}$ is Lipschitz such that
\[
|\phi(O;\check \Gamma) - \phi(O;\Gamma)|\lesssim_{\PP }| \check \Gamma- \Gamma|,\ \forall \Gamma,\check\Gamma \in \mathcal{D}.
\]
\end{lemma}

\begin{proof}
From \Cref{thm:if_psi_plus}, the uncentered EIF $\phi_+(O;\Gamma)$ can be rewritten as
\begin{align*}
 &\  \frac{Z W_+(X)}{e(X)}\left\{      \left[Y-Q(X) \right]\one_{\{Y>Q(X)\}} - \E_{Y\mid X,Z=1}\left[\left[Y-Q(X) \right]        \one_{\{Y>Q(X)\}}       \right] \right\} \\
&\ + \left[(1-\Gamma)Z+\Gamma \right] \left( \E_{Y\mid X,Z=1}\left[\left[Y-Q(X) \right]   \one_{\{Y>Q(X)\}}       \right] \ + Q(X) \P_{Y\mid X,Z=1}\{Y>Q(X)\}\right).
\end{align*}
Recall from \Cref{prop:quantile_balancing}
 that $Q(X)$ is the $\Gamma/(1+\Gamma)$-quantile. Let $\check Q(X)$
denote the $\check \Gamma/(1+\check \Gamma)$-quantile.
Under Assumption \ref{assumption:first},  $p_{Y\mid X,Z=1}(Q(X))\neq 0$ for any $\Gamma\in \mathcal{D}$. Then,
\begin{equation}\label{equ:lemma_4_derivative}
\frac{d Q(X) }{d \Gamma} = (1+\Gamma)^{-2}/p_{Y\mid X,Z=1}(Q(X))<\infty,
\end{equation}
which shows that $Q(X)$ is a Lipschitz function of $\Gamma$. So is $\PP_{Y\mid X,Z=1}\{Y> Q(X)\}$ because \[
|\PP_{Y\mid X,Z=1}\{Y>\check Q(X)\} - \PP_{Y\mid X,Z=1}\{Y>Q(X)\} |\lesssim  |\check Q(X) - Q(X) |.
\]
Similarly, it is straightforward to show that
\begin{align*}
 \big |\big[Y-\check Q(X) \big]\one_{\{Y-\check Q(X)> 0 )\}} - \left[Y-Q(X) \right]\one_{\{Y-Q(X)>0\}} \big |
\leq &\  \big|\big[Y-\check Q(X) \big] - \left[Y-Q(X) \right]\big|  \\
= &\ |\check Q(X)  - Q(X)  |.
\end{align*}
Then, $\left[Y-Q(X) \right]\one_{\{Y>Q(X)\}}$ and
$\E_{Y\mid X,Z=1}\left[\left[Y-Q(X) \right]       \one_{\{Y>Q(X)\}}       \right]$ are Lipschitz functions of $\Gamma.$
So are $(1-\Gamma)Z + \Gamma $ and $W_+(X) = (1-\Gamma)e(X) + \Gamma $ because their derivatives with respect to  $\Gamma$ are upper bounded by 1. We have proven that all the components in $\phi_+(O;\Gamma)$ are  Lipschitz.  Under Assumption \ref{assumption:nuisance_error}, given $O$, all of them are also bounded functions of $\Gamma\in \mathcal{D}.$
As the products and sums of
these bounded Lipschitz functions, $\phi_+(O;\Gamma)$ is a Lipschitz function of $\Gamma$.
Next, we rewrite the uncentered EIF  $\phi_-$ as
\begin{align*}
\phi_-(O;\Gamma)
=   \frac{Z W_-(X)}{e(X)}\left\{      \left[Y-Q(X) \right]\one_{\{Y<Q(X)\}} - \E_{Y\mid X,Z=1}\left[\left[Y-Q(X) \right]        \one_{\{Y<Q(X)\}}       \right] \right\} \\
+ \left[(1-\Gamma)Z+\Gamma \right] \big( \E_{Y\mid X,Z=1}\left[\left[Y-Q(X) \right]   \one_{\{Y<Q(X)\}}       \right] + Q(X) \PP_{Y\mid X,Z=1}\{Y<Q(X)\} \big).
\end{align*}
Using the same proof for $\phi_+(O;\Gamma)$ above, we can show that $\phi_-(O;\Gamma)  $ is a Lipschitz function of $\Gamma$. Then,
$\phi(O;\Gamma) = \phi_+(O;\Gamma) + \phi_-(O;\Gamma) $ is also a Lipschitz function of $\Gamma$.
\end{proof}

\begin{lemma}\label{lemma:psi_12_lip}
Under Assumption \ref{assumption:first},
$\mathcal{F}_j = \left\{\phi_j(O;\lambda ): \lambda \in \mathcal{D}_{12} =  [\lambda_{\min },\lambda_{\max}]  \right\}$ is Lipschitz such that
\[
|\phi_j(O;\check \lambda ) - \phi_j(O; \lambda )|\lesssim_{\PP}|\check \lambda - \lambda  |, \ \forall  \lambda,\check\lambda \in \mathcal{D}_{12} \ \text{ and }  j=1,2.
\]
\end{lemma}

\begin{proof}
The uncentered EIFs $\phi_1(O;\lambda )$ and $ \phi_2(O;\lambda )$ are given in \Cref{prop:eif_small}.
By the definition of $\xi_X$ in \Cref{prop:double_ipw}, we can view $\lambda$ as a strictly decreasing function of $\xi_X$:
\[
\zeta_1(\xi_X ) := [1-e(X)]/ \mathbb E_{Y\mid X,Z=1}\left[(\xi_X - Y) \one_{\{Y\leq \xi_X\}}\right].
\]
Its inverse $\xi_X = \zeta_1^{-1}(\lambda) $ is a strictly decreasing function of $\lambda$, with a negative and bounded derivative
\begin{equation}\label{equ:lemma_7_step}
\left(\frac{d \zeta_1(\xi_X )}{d \xi_X } \right)^{-1} = -  \frac{1-e(X)}{\mathbb E_{Y\mid X,Z=1} \big[ \one_{\{Y\leq \zeta_1^{-1}(\lambda)\}} \big]} \mathbb E_{Y\mid X,Z=1}^2 \big[(\zeta_1^{-1}(\lambda)- Y) \one_{\{Y\leq \zeta_1^{-1}(\lambda)\}}\big],
\end{equation}
for any $\lambda \in [\lambda_{\min},\lambda_{\max}].$ This implies that $\xi_X = \zeta_1^{-1}(\lambda) $ is a Lipschitz function of $\lambda$, and that
\[
|(\check \xi_{X}-Y)\one_{\{Y\leq \check \xi_{X}\}} - (\xi_{X}-Y)\one_{\{Y\leq \xi_{X}\}} | \leq |       (\check \xi_{X}-Y) -    (\xi_{X}-Y)             |  = |\check \xi_{X} - \xi_{X}|
\lesssim_{\PP} |  \check \lambda          - \lambda       |.
\]
Since the product of bounded Lipschitz functions is Lipschitz,
$h_*(X,Y),$ $h_*^2(X,Y), h_*(X,Y)Y$ and their expectations are also Lipschitz functions of $\lambda.$
For any $\xi_X \in [ \zeta_1^{-1}(\lambda_{\max}), \zeta_1^{-1}(\lambda_{\min}) ],$
$1/\mathbb E_{Y\mid X,Z=1}  \big[\one_{\{Y\leq \xi_X \}}  \big]$ and $\mathbb E_{Y\mid X,Z=1}  \big[Y \one_{\{Y\leq \xi_X \}}  \big]$ have bounded derivatives with respect to $\xi_X$. Then their derivatives with respect to $\lambda$ are also bounded by the chain rule and the bounded derivative in \eqref{equ:lemma_7_step}.  Since all the components of the uncentered EIFs
$ \phi_1(O;\lambda ) $ and $ \phi_2(O;\lambda ) $ are bounded Lipschitz functions of  $\lambda\in \mathcal{D}_{12}$, we know that they are also Lipschitz functions of $\lambda\in \mathcal{D}_{12}$.
\end{proof}

\begin{lemma}\label{lemma:psi_3_lip}
Under Assumption \ref{assumption:first},
 $\mathcal{F}_3 = \left\{\phi_3(O;\theta ): \theta \in \mathcal{D}_{3} =  [\theta_{\min },\theta_{\max}]  \right\}$ is Lipschitz such that
 \[
 |\phi_j(O;\check \theta ) - \phi_j(O; \theta)|\lesssim_{\PP}|\check \theta - \theta  |,
 \ \forall \theta,\check\theta \in \mathcal{D}_{3}.
 \]
\end{lemma}

\begin{proof}
    At the end of \Cref{appendix:solutions_new_1}, we show that the function $f_{\theta,X}(\xi)$ in \Cref{prop:solutions_new_1} has a positive derivative w.r.t $\xi.$ Thus, the root $\xi_X$ is a strictly decreasing function of $\theta$.
    Similarly,
$\lambda_X = [1-e(X)]/\mathbb E\left[(\xi_X - Y) \one_{\{Y\leq \xi_X\}}\right]$
is a strictly decreasing function of $\xi_X$, i.e., a strictly increasing function of $\theta$.
Under the assumptions, their derivative w.r.t $\theta\in \mathcal{D}_3$ is bounded. Hence, $\xi_X$ and $\lambda_X$ are Lipschitz functions of $\theta.$
The rest of the proof follows the same steps in \Cref{lemma:psi_12_lip}, verifying each component of the uncentered EIF $\phi_3(O;\theta )$ in \Cref{prop:eif_large} is Lipschitz; details are omitted here.
\end{proof}

\end{document}